\documentclass[12pt, draftclsnofoot, journal, letterpaper, onecolumn]{IEEEtran}

\IEEEoverridecommandlockouts
\usepackage{amsfonts}
\usepackage[dvips]{graphicx}
\usepackage{times}
\usepackage{cite}
\usepackage{amsmath}
\usepackage{cases}
\usepackage{array}
\usepackage{dsfont}
\usepackage{amssymb}

\usepackage{stfloats}
\usepackage{slashbox}
\usepackage{graphicx}
\usepackage{footnote}
\usepackage{booktabs}
\usepackage{array}
\usepackage{multirow}
\usepackage{bm}
\usepackage{empheq}
\usepackage[labelformat=simple]{subcaption}

\usepackage{amsthm}
\usepackage{algorithm}
\usepackage{algorithmic}
\usepackage{color}
\theoremstyle{plain}
\newtheorem{theorem}{Theorem}
\newtheorem{corollary}{Corollary}
\newtheorem{lemma}{Lemma}
\newtheorem{proposition}{Proposition}
\theoremstyle{definition}

\newtheorem{definition}{Definition}
\newtheorem{remark}{Remark}
\begin{document}

\title{Fundamental Tradeoff between Storage and Latency in Cache-Aided Wireless Interference Networks}
\author{\IEEEauthorblockN{Fan Xu, Meixia Tao, \IEEEmembership{Senior Member,~IEEE}, and Kangqi Liu, \IEEEmembership{Student Member,~IEEE}}\\
\thanks{This paper was presented in part at IEEE ISIT 2016. }
\thanks{The authors are with the Department of Electronic Engineering, Shanghai Jiao Tong University, Shanghai, 200240, China (Emails: xxiaof@sjtu.edu.cn, mxtao@sjtu.edu.cn, k.liu.cn@ieee.org).}}
\maketitle

\begin{abstract}
This paper studies the fundamental tradeoff between storage and latency in a general wireless interference network with caches equipped at all transmitters and receivers. The tradeoff is characterized by an information-theoretic metric, \emph{normalized delivery time} (NDT), which is the worst-case delivery time of the actual traffic load at a transmission rate specified by degrees of freedom (DoF) of a given channel. We obtain both an achievable upper bound and a theoretical lower bound of the minimum NDT for any number of transmitters, any number of receivers, and any feasible cache size tuple. We show that the achievable NDT is exactly optimal in certain cache size regions, and is within a bounded multiplicative gap to the theoretical lower bound in other regions. In the achievability analysis, we first propose a novel cooperative transmitter/receiver coded caching strategy. It offers the freedom to adjust file splitting ratios for NDT minimization. We then propose a delivery strategy which transforms the considered interference network into a new class of cooperative X-multicast channels. It leverages local caching gain, coded multicasting gain, and transmitter cooperation gain (via interference alignment and interference neutralization) opportunistically. Finally, the achievable NDT is obtained by solving a linear programming problem. This study reveals that with caching at both transmitter and receiver sides, the network can benefit simultaneously from traffic load reduction and transmission rate enhancement, thereby effectively reducing the content delivery latency.

\end{abstract}

\begin{IEEEkeywords}
Wireless cache network, coded caching, content delivery, multicast, and interference management.
\end{IEEEkeywords}
%\vspace{0.5cm}

%\begin{IEEEkeywords}
%Multi-input multi-output, signal alignment, degrees of freedom.
%\end{IEEEkeywords}

\section{Introduction}
Over the last decades, mobile data traffic has been shifting from connection-centric services, such as voice, e-mails, and web browsing,
to emerging content-centric services, such as video streaming, push media, application download/updates, and mobile TV \cite{5g,liu2014content,cisco}. These contents are typically produced well ahead of transmission and can be requested by multiple users at possibly different times. This allows us to cache the contents at the edge of networks, e.g., base stations and user devices, during periods of low network load. The local availability of contents at the network edge has significant potential of reducing user access latency and alleviating wireless traffic. Recently, there have been increasing interests from both academia and industry in characterizing the impact of caching on wireless networks \cite{fem,edge,air,business,editorial}.

Caching in a shared link with one server and multiple cache-enabled users is first studied by Maddah-Ali and Niesen in \cite{fundamentallimits}. It is shown that caching at user ends, also known as coded caching, brings not only local caching gain but also global caching gain. The latter is achieved by a carefully designed cache placement and coded delivery strategy, which can create multicast chances for content delivery even if users demand different files. The idea of coded caching in \cite{fundamentallimits} is then extended to the distributed network in \cite{decentralized}, which achieves a rate close to the optimal centralized scheme. Taking file popularity into consideration, the authors in \cite{nonuniformdemands,nonuniformdemand10,nonuniformdemand11} introduced order-optimal coded caching schemes for average traffic load performance. In \cite{cache&csit}, the authors considered the wireless broadcast channel with imperfect channel state information at the transmitter (CSIT) and showed that the gain of coded caching can offset the loss due to the imperfect CSIT. Besides the shared link or broadcast channels mentioned above, coded caching is also investigated in other network topologies, such as hierarchical cache networks \cite{hierarchical}, device-to-device cache networks \cite{D2D}, and multi-server networks \cite{multi-server}.

Caching at transmitters is studied in \cite{basestationcooperation,hypergraph,upperbound,lowerbound,simeone22,simeone} to exploit the opportunities for transmitter cooperation and interference management. In specific, the authors in \cite{basestationcooperation} exploited the multiple-input multiple-output (MIMO) cooperation gain via joint beamforming by caching the same erasure-coded packets at all edge nodes in a backhaul-limited multi-cell network. The authors in \cite{hypergraph} studied the degrees of freedom (DoF) and clustered cooperative beamforming in cellular networks with edge caching via a hypergraph coloring problem. The authors in \cite{upperbound} studied the transmitter cache strategy in a $3\times3$ cache-aided interference channel under an information-theoretical framework. It is shown that splitting contents into different parts and caching each part in different transmitters can turn the interference channel into broadcast channel, X channel, or hybrid channel and hence increase the system throughput via interference management. The authors in \cite{lowerbound} presented a lower bound of delivery latency in a general interference network with transmitter cache and showed that the scheme in \cite{upperbound} is optimal in certain region of cache size. The authors in \cite{simeone22} studied a $2\times2$ cloud and cache aided wireless network,  and characterized the optimal tradeoff between cache storage size and content delivery time.  Then, the authors in \cite{simeone} extended the network in \cite{simeone22} to the general cloud and cache aided wireless network, and showed that their proposed transmission strategy achieves the optimality within a constant factor 2.

The above literature reveals that caching at the receiver side can bring local caching gain and coded multicasting gain, and that caching at the transmitter side can induce transmitter cooperation for interference management and load balancing. In modern wireless communication systems, storage space has been proliferating in both base stations and smart mobile devices \cite{fem}. During the off-peak traffic time, both the base stations and mobile users can download certain files from the core network into their local caches in advance. When the users submit content requests in the peak traffic time afterwards, the locally cached contents can be utilized to relieve the burden of the network traffic and reduce the delivery latency. It is thus of both theoretical importance and practical interest to investigate the impact of caching at both transmitter and receiver sides.

In this paper, we aim to study the fundamental limits of caching in a general wireless interference network with caches equipped at all transmitters and receivers as shown in Fig.~\ref{Fig model}. The performance metric to characterize the gains of caching varies in the existing works. For the broadcast channel with receiver cache, the authors in \cite{fundamentallimits} characterized the gain by \textit{memory-rate tradeoff}, where the \textit{rate} is defined as the normalized load of the shared link with respect to the file size in the delivery phase. For the interference channel with transmitter cache, the authors in \cite{upperbound} characterized the gain by the standard DoF from the information-theoretic studies. In \cite{lowerbound}, the authors introduced the \textit{storage-latency tradeoff}, where the latency is defined as the relative delivery time with respect to an ideal baseline system with unlimited cache and no interference in the high signal-to-noise ratio (SNR) region. In our considered wireless interference network with both transmitter and receiver caches, the standard DoF is unable to capture the potential reduction in the traffic load due to receiver cache, and the \textit{rate} is unable to capture the potential DoF enhancement due to cache-induced transmitter cooperation. Interestingly, the latency-oriented performance metric in \cite{lowerbound} can reflect not only the load reduction due to receiver cache but also the DoF enhancement due to transmitter cache, since it evaluates the delivery time of the actual \textit{load} at a transmission rate specified by the given \textit{DoF}. As such, we adopt the \textit{storage-latency tradeoff} to characterize the fundamental limits of caching in this work. In specific, we measure the performance by \emph{normalized delivery time} (NDT) as defined in \cite{lowerbound,simeone22,simeone}, denoted as $\tau(\mu_R, \mu_T)$, which is a function of the normalized receiver cache size $\mu_R$ and the normalized transmitter cache size $\mu_T$.

Our preliminary results on the latency-storage tradeoff study in the special case with $3$ transmitters and $3$ receivers are presented in \cite{mine}. Note that an independent work on the similar problem with both transmitter and receiver caches is studied in \cite{bothcache}. After the initial submission of this work, another similar work is studied in \cite{niesen}. We shall discuss the differences with \cite{bothcache,niesen} at appropriate places throughout the paper. The main contributions and results of this work are listed as follows:

$\bullet$ \textit{A novel file splitting and caching strategy}: We propose a novel file splitting and caching strategy for any number of transmitters and receivers, and at any feasible normalized cache size tuples. This strategy is more general than the existing file splitting and caching strategy in \cite{fundamentallimits,upperbound,bothcache,niesen}. It offers the freedom to adjust the file splitting ratios for caching gain optimization.

$\bullet$ \textit{Achievable storage-latency tradeoff}: Based on the proposed file splitting and caching strategy, we obtain an achievable upper bound of the minimum NDT for the general $N_T\times N_R$ interference networks by solving a linear programming problem of file splitting ratios. The achievable NDT is for any number of transmitters $N_T\ge2$, any number of receivers $N_R\ge2$, and any feasible normalized cache size tuples $(\mu_R,\mu_T)$. The main idea is to design the delivery phase carefully so that the network topology can be opportunistically changed to a new class of cooperative X-multicast channels, which includes X channel, broadcast channel, and  multicast channel as special cases. Interference neutralization and interference alignment are used to increase the system DoF of these channels. Our analysis shows that the transmitter cooperation gain, local caching gain, and coded multicasting gain can be leveraged opportunistically in different cache size regions. Our analysis also shows that the optimal file spitting ratios are not unique. The multiple choices offer the freedom to choose a proper caching scheme according to practical limitations, such as subpacketization overhead and receiver complexity.

$\bullet$ \textit{Lower bound of storage-latency tradeoff}: We also obtain a lower bound of the minimum NDT for the general $N_T\times N_R$ interference network by using genie-message approach. With this lower bound, we show that the achievable NDT upper bound is optimal in certain cache size regions. In other regions, the multiplicative gap between the upper and lower bounds is within 2 when $N_T\ge N_R$, within 12 when $N_T<N_R$ and $\mu_T\ge\frac{1}{N_T}$, and within $\frac{N_T+N_R-1}{N_T}$ when $N_T<N_R$ and $\mu_T<\frac{1}{N_T}$.

%  \item \textit{Achievable DoF of new channels}: We obtain the achievable DoF of three new channels, i.e. the partially cooperative X channel, the hybrid X-multicast channel, and the partially cooperative hybrid X-multicast channel, each with $N_T\ge2$ transmitters and $N_R=3$ receivers. All the three achievable DoFs approach to 3 when $N_T \to\infty$. Interference neutralization, interference alignment, and real interference alignment are used in different channels to obtain the achievable DoF. These DoF results are necessary to obtain the achievable NDT results mentioned above.

The remainder of this paper is organized as follows. Section \ref{section model} introduces our system model and performance metric. Section \ref{section main results} presents the main results of this paper. Section \ref{section cache placement} describes the cache placement strategy. Section \ref{section delivery} illustrates the content delivery strategy. Section \ref{section discussion} presents some discussions. Section \ref{section converse} proves the lower bound of NDT, and Section \ref{section conclusion} concludes this paper.

Notations: $(\cdot)^T$ denotes the transpose. $[K]$ denotes set $\{1,2,\cdots,K\}$. $\lfloor x\rfloor$ denotes the largest integer no greater than $x$. $(x_j)^K_{j=1}$ denotes vector $(x_1,x_2,\cdots,x_K)^T$. $(x)^+$ denotes the maximum of $x$ and 0, i.e. $(x)^+=\max\{0,x\}$. $A_{1\sim S}$ denotes set $\{A_1,A_2,\cdots,A_S\}$. $\mathcal{CN}(m,\sigma^2)$ denotes the complex Gaussian distribution with mean of $m$ and variance of $\sigma^2$.

\section{System Description and Performance Metric}\label{section model}
\subsection{System Description}
Consider a general cache-aided wireless interference network with $N_T$ ($\ge2$) transmitters and $N_R$ ($\ge2$) receivers as illustrated in Fig.~\ref{Fig model}, where each node is equipped with a cache memory of finite size.\footnote{This work focuses on the general interference network with at least two transmitters and at least two receivers. We do not consider the special cases with either one transmitter (broadcast channel) or one receiver (multiple-access channel).} Each node is assumed to have single antenna. The communication link between each transmitter and each receiver experiences channel fading, and is corrupted with additive white Gaussian noise. The communication at each time slot $t$ over this network is modeled by
\begin{align}
Y_q(t)=\sum_{p=1}^{N_T}h_{qp}(t)X_p(t)+Z_q(t),q=1,2,\cdots,N_R,\notag
\end{align}
where $Y_q(t)\in \mathbb{C}$ denotes the received signal at receiver $q$, $X_p(t)\in \mathbb{C}$ denotes the transmitted signal at transmitter $p$, $h_{qp}(t)\in \mathbb{C}$ denotes the channel coefficient from transmitter $p$ to receiver $q$ which is assumed to be identically and independently (i.i.d.) distributed as some continuous distribution, and $Z_q(t)$ denotes the noise at receiver $q$ distributed as $\mathcal{CN}(0,1)$.
\begin{figure}[tbp]
\begin{centering}
\includegraphics[scale=0.3]{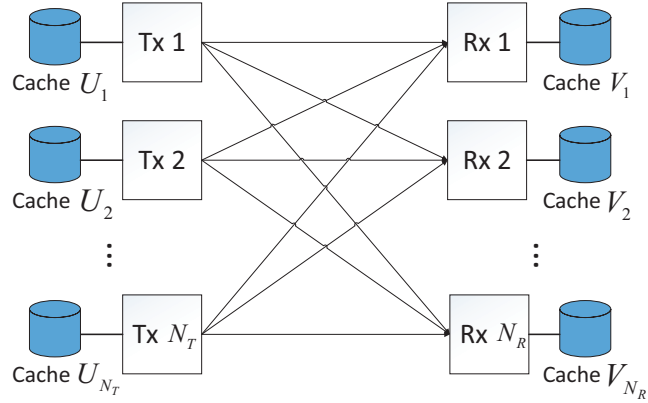}
\vspace{-0.1cm}
\caption{Cache-aided wireless interference network with $N_T$ transmitters and $N_R$ receivers.}\label{Fig model}
\vspace{-0.3cm}
\end{centering}
\end{figure}

Consider a database consisting of $L$ files, denoted by $\{W_1,W_2,\cdots,W_L\}$. Throughout this study, we consider $L\ge N_R$ so that each receiver can request a distinct file. Each file is chosen independently and uniformly from $[2^F]=\{1,2,\cdots,2^F\}$ randomly, where $F$ is the file size in bits. Each transmitter has a local cache able to store $M_TF$ bits and each receiver has a local cache able to store $M_RF$ bits. The \textit{normalized cache sizes} at each transmitter and receiver are defined, respectively, as
\begin{eqnarray}
\mu_T\triangleq\frac{M_T}{L},\qquad\mu_R\triangleq\frac{M_R}{L}.\notag
\end{eqnarray}

The network operates in two phases, \textit{cache placement phase} and \textit{content delivery phase}. During the cache placement phase, each transmitter $p$ designs a caching function $\phi_{p,i}$ that maps each file $W_i$ into its cached content $U_{p,i}$ as
\begin{eqnarray}
U_{p,i}\triangleq \phi_{p,i} (W_i), \forall i\in [L].\notag
\end{eqnarray}
Define the overall cached content at transmitter $p$ as $U_{p}\triangleq \bigcup_{i\in[L]}U_{p,i}$. The mapping $\{\phi_{p,i}\}$ is such that $H(U_{p})\le M_TF$ in order to satisfy the cache capacity constraint at each transmitter. Each receiver $q$ also designs a caching function $\psi_{q,i}$ that maps each file $W_i$ into its cached content $V_{q,i}$ as
\begin{eqnarray}
V_{q,i}\triangleq \psi_{q,i} (W_i), \forall i\in [L].\notag
\end{eqnarray}
Define the overall cached content at receiver $q$ as $V_q\triangleq \bigcup_{i\in[L]}V_{q,i}$. The mapping $\psi_{q,i}$ is such that $H(V_q)\le M_RF$ in order to satisfy the cache capacity constraint at each receiver. The caching functions $\{\phi_{p,i},\psi_{q,i}\}$ are assumed to be known globally at all nodes. Note that in this paper, we restrict our study to the caching functions that do not allow for inter-file coding, but can allow for arbitrary coding within each file. Similar assumptions have been made in \cite{lowerbound,simeone}.

In the delivery phase, each receiver $q$ requests a file $W_{d_q}$ from the database. We denote ${\bf d}\triangleq(d_q)^{N_R}_{q=1}\in[L]^{N_R}$ as the demand vector. Each transmitter $p$ has an encoding function
\begin{eqnarray}
\Lambda_p:[2^{\lfloor FM_T\rfloor}]\times[L]^{N_R}\times\mathbb{C}^{N_T\times N_R}\to\mathbb{C}^T.\notag
\end{eqnarray}
Transmitter $p$ uses $\Lambda_p$ to map its cached content $U_p$, receiver demands ${\bf d}$, and channel realization $\mathbf{H}$ to the codeword $(X_p[t])^T_{t=1}\triangleq\Lambda_p(U_p,{\bf d},\mathbf{H})$, where $T$ is the block length of the code. Note that $T$ may depend on the receiver demand ${\bf d}$ and channel realization $\mathbf{H}$. Each codeword $(X_p[t])_{t=1}^T$ has an average transmit power constraint $P$. Each receiver $q$ has a decoding function
\begin{eqnarray}
\Gamma_q:[2^{\lfloor FM_R\rfloor}]\times\mathbb{C}^T\times\mathbb{C}^{N_T\times N_R}\times[L]^{N_R}\to[2^F].\notag
\end{eqnarray}
We denote $(Y_q[t])_{t=1}^T$ as the signal vector received at receiver $q$. Upon receiving $(Y_q[t])^T_{t=1}$, each receiver $q$ uses $\Gamma_q$ to decode $\hat{W}_q\triangleq\Gamma_q(V_q,(Y_q[t])^T_{t=1},\mathbf{H},{\bf d})$ of its desired file $W_{d_q}$ using its cached content $V_q$ and the channel realization $\mathbf{H}$ as side information. The worst-case error probability is
\begin{eqnarray}
P_\epsilon=\max\limits_{{\bf d}\in[L]^{N_R}}\max\limits_{q\in[N_R]}\mathbb{P}(\hat{W}_q\ne W_{d_q}).\notag
\end{eqnarray}
The given caching and coding scheme $\{\phi_{p,i},\psi_{q,i},\Lambda_p,\Gamma_q\}$ is said to be feasible if, for almost all channel realizations, $P_\epsilon\to 0$ when $F\to\infty$.

Note that the cache placement phase and the content delivery phase take place on different timescales. In general, cache placement is in a much larger timescale (e.g. on a daily or hourly basis) while content delivery is in a much shorter timescale. As such, the caching functions designed in the cache placement phase are unaware of the future content requests, but the coding functions during the content delivery phase are dependent on the caching functions.

\subsection{Performance Metric}
In this work, we adopt the following latency-oriented performance metric as in \cite{lowerbound,simeone22,simeone}.\footnote{The performance metric NDT is first proposed in \cite{lowerbound} for wireless networks with transmitter cache only. It is then scaled by the number of receivers and renamed as fractional delivery time (FDT) by taking receiver cache into account in \cite{mine} as well as the initial submission of this paper. During the paper revision, we have removed the scaling and changed back to NDT for consistency with \cite{lowerbound}.}

\begin{definition}[\cite{lowerbound}]\label{definition ndt}
The \textit{normalized delivery time} (NDT) for a given feasible caching and coding scheme at a given normalized cache size tuple $(\mu_R, \mu_T)$ is defined as
\begin{eqnarray}
\tau(\mu_R,\mu_T)\triangleq\lim_{P\to\infty}\lim_{F\to\infty}\sup\frac{\max\limits_{{\bf d}}T}{F/\log P}.\notag
\end{eqnarray}
Moreover, the minimum NDT is defined as
\begin{eqnarray}
\tau^*(\mu_R,\mu_T)=\inf\{\tau(\mu_R,\mu_T):\tau(\mu_R,\mu_T)\rm \;is\;achievable\}.\notag
\end{eqnarray}
\end{definition}

Note that $F/\log P$ is the delivery time of transmitting one file of $F$ bits in a point-to-point baseline system with Gaussian noise in the high SNR regime. An NDT of $\tau^*$ thus indicates that the worst-case time required to serve any possible demand vector $\mathbf{d}$ is $\tau^*$ times of this reference time period.

\begin{remark}[Interpretation of NDT]\label{remark fdt calculation}
Let $R$ denote the worst-case traffic load per user with respect to the file size $F$. Since the per-user capacity of the network in the high SNR regime can be approximately given by $(d\cdot \log P +o(\log P))$, where $d$ is the per-user DoF, the worst-case delivery time can be rewritten as $\max\limits_{{\bf d}}T=\frac{RF}{d\cdot \log P +o(\log P)}$. Then, by Definition \ref{definition ndt}, NDT can be expressed more conveniently as
\begin{align}
\tau=R/d.\label{eqn NDTtau}
\end{align}
In the special case with transmitter cache only, we have $\tau(\mu_R=0, \mu_T) = 1/d$. As a result,  NDT characterizes the asymptotic delivery time of the actual per-user traffic \emph{load} $R$  at a transmission rate specified by the \emph{per-user DoF} $d$ when $P\to\infty$ and $F\to \infty$, and hence is particularly suitable to measure the performance of the wireless networks with both transmitter and receiver caches.
\end{remark}

\begin{remark}[Feasible region of NDT]\label{remark feasible region}
The NDT introduced above is able to measure the fundamental tradeoff between the cache storage and content delivery latency. However, not all normalized cache sizes are feasible. Given fixed $L$ and $M_T$, all the transmitters together can store at most $N_TM_TF$ bits of files, which leaves at least $LF-N_TM_TF$ bits of files to be stored in all receivers. Thus we must have $M_RF\ge LF-N_TM_TF$. This results in the following feasible region for the normalized cache sizes:
\begin{align}
\left\{
\begin{array}{ll}0\le\mu_R,\mu_T\le1\\ \mu_R+N_T\mu_T\ge1
\end{array}.\label{eqn boundary}
\right.
\end{align}
Throughout this paper, we study the NDT in the above feasible region. Note that the works in \cite{bothcache,niesen} are limited to the region $\{\frac{1}{N_T}\le\mu_T\le1,0\le\mu_R\le1\}$.
\end{remark}

%\begin{figure}[tbp]
%\begin{centering}
%\includegraphics[scale=0.5]{region.eps}
%\vspace{-0.1cm}
%\caption{Feasible domain of NDT (divided into 5 regions).}\label{Fig region}
%\vspace{-0.3cm}
%\end{centering}
%\end{figure}
\section{Main Results}\label{section main results}
In this section, we first present our main results on the fundamental storage-latency tradeoff for the general cache-aided wireless interference network. Then, we present the results in some special cases and discuss the connections with existing works.
\subsection{General Results}

\begin{theorem}[Achievable NDT] \label{thm 1}
For the cache-aided interference network with $N_T\ge2$ transmitters, $N_R\ge2$ receivers, and $L\ge N_R$ files, where each transmitter has a cache of normalized size $\mu_T$ and each receiver has a cache of normalized size $\mu_R$, the minimum NDT is upper bounded by the optimal solution of the following linear programming (LP) problem:
\begin{align}
\mathcal{P}_1:\quad&\tau_U(\mu_R,\mu_T)\triangleq \notag\\
\min_{\{a_{r,t}:(r,t)\in \mathcal{A}\}}&\sum_{r=0}^{N_R-1}\sum_{t=1}^{N_T}\frac{\binom{N_R-1}{r}\binom{N_T}{t}}{d_{r,t}}a_{r,t},\label{eqn tmin}\\
\textrm{s.t.} \qquad&\sum_{r=0}^{N_R} \sum_{t=1}^{N_T}\binom{N_R}{r}\binom{N_T}{t}a_{r,t}+a_{N_R,0}=1,\label{eqn tmin1}\\
&\sum_{r=1}^{N_R} \sum_{t=1}^{N_T}\binom{N_R-1}{r-1}\binom{N_T}{t}a_{r,t}+a_{N_R,0}\le\mu_R,\label{eqn tmin2}\\
&\sum_{r=0}^{N_R} \sum_{t=1}^{N_T}\binom{N_R}{r}\binom{N_T-1}{t-1}a_{r,t}\le\mu_T,\label{eqn tmin3}\\
&0\le a_{r,t} \le 1,\forall (r,t)\in \mathcal{A}
\end{align}
where $\mathcal{A}\triangleq\{(r,t): r+N_Rt\ge N_R,0\le r\le N_R, 0\le t\le N_T,r,t\in\mathds{Z}\}$, $\{a_{r,t}\}$ are the (file splitting) variables to be optimized, and $d_{r,t}$ is given by
\begin{align}
d_{r,t}
=\left\{
\begin{array}{ll}
1, &r+t\ge N_R\\
\frac{\binom{N_R-1}{r}\binom{N_T}{t}t}{\binom{N_R-1}{r}\binom{N_T}{t}t+1}, &r+t= N_R-1\\
\max\left\{d_{1},\frac{r+t}{N_R}\right\}, &r+t\le N_R-2
\end{array}
\right.\label{eqn tau ip}
\end{align}
with $d_1$ being
\begin{align}
d_1\triangleq\max\limits_{1\le t'\le t}\left\{\frac{\binom{N_R-1}{r}\binom{N_T}{t'}\binom{N_R-r-1}{t'-1}t'}
{\binom{N_R-1}{r}\binom{N_T}{t'}\binom{N_R-r-1}{t'-1}t'+\binom{N_R-1}{r+1}\binom{N_R-r-2}{t'-1}\binom{N_T}{t'-1}}\right\}
\end{align}
\end{theorem}

The LP problem in Theorem \ref{thm 1} can be solved efficiently by some linear equation substitution and other manipulations. The closed-form and optimal solutions in the special cases when $N_T=N_R=2$ and $N_T=N_R=3$ are given in Corollary \ref{coro 2x2} and Corollary \ref{coro 3x3}, respectively, in Section \ref{section delivery}.

\begin{theorem}[Lower bound of NDT]\label{thm 2}
For the cache-aided interference network with $N_T\ge2$ transmitters, $N_R\ge2$ receivers, and $L\ge N_R$ files, where each transmitter has a cache of normalized size $\mu_T$ and each receiver has a cache of normalized size $\mu_R$, the minimum NDT is lower bounded by
\begin{align}
\tau^*(\mu_R,\mu_T)\ge\tau_{L1}\triangleq\max\limits_{\substack{l=1,\cdots,\min\{N_T,N_R\}\\s_1=0,1,\ldots,l\\s_2=0,1,\ldots,N_R-l}}\frac{1}{l}
\Big\{&(s_1+s_2)-(N_T-l)s_2\mu_T\notag\\
&\left.-\left(\frac{2s_2+s_1+1}{2}\cdot s_1+s_2^2\right)\mu_R\right\},\label{eqn converse}
\end{align}
when the caching functions allow for arbitrary intra-file coding, and by
\begin{align}
\tau^*(\mu_R,\mu_T)\ge\tau_{L2}\triangleq\max\limits_{\substack{l=1,\cdots,\min\{N_T,N_R\}\\s_1=0,1,\ldots,l\\s_2=0,1,\ldots,N_R-l}}\frac{1}{l}
\Big\{&(s_1+s_2)-(N_T-l)s_2\mu_T\notag\\
&\left.-\left(\frac{2s_2+s_1+1}{2}\cdot s_1+s_2^2\right)\mu_R\right.\notag\\
&\left.+\left(\frac{2s_2+s_1}{2}(s_1-1)+s_2^2\right)(1-N_T\mu_T)^+\right\}\label{eqn taul2}
\end{align}
when the intra-file coding is not allowed.
\end{theorem}

The proof of Theorem \ref{thm 1} will be given in Section \ref{section cache placement} and \ref{section delivery}, and the proof of  Theorem \ref{thm 2} will be given in Section \ref{section converse}.

It can be seen from the above theorems that both the upper and lower bounds of the minimum NDT are convex\footnote{Please refer to \cite[Lemma 1]{simeone} for more detailed analysis of the convexity of NDT.} and non-increasing functions of normalized cache sizes $\mu_R$ and $\mu_T$. The following two corollaries state the relations of the two bounds, the proofs of which will be given in Appendix B and Appendix C, respectively.

\begin{corollary}[Optimality]\label{coro optimality}
The achievable NDT is optimal (i.e., coincides with the lower bound) when $(\mu_R,\mu_T)$ satisfies any of the following conditions:
\begin{enumerate}
  \item $N_R\mu_R+N_T\mu_T\ge N_R$: the optimal NDT is $\tau^*=1-\mu_R$;
  \item $(\mu_R,\mu_T)=(0,1)$: the optimal NDT is $\tau^*=\frac{N_R}{\min\{N_T,N_R\}}$;
  \item $(\mu_R,\mu_T)=(0,1/N_T)$: the optimal NDT is $\tau^*=\frac{N_T+N_R-1}{N_T}$;
  \item $\mu_R + N_T\mu_T=1$ when there is no intra-file coding in the caching functions: the optimal NDT is $\tau^*=\frac{N_T+N_R-1}{N_T}(1-\mu_R)$.
\end{enumerate}
\end{corollary}

\begin{corollary}[Gap of NDT]\label{coro gap}
The multiplicative gap between the upper bound and the lower bound of the minimum NDT is within 2 when $N_T\ge N_R$,  within 12 when $N_T<N_R,\mu_T\ge\frac{1}{N_T}$, and within $\frac{N_T+N_R-1}{N_T}$ when $N_T<N_R,\mu_T<\frac{1}{N_T}$.
\end{corollary}

\subsection{Special Cases}
\subsubsection{Transmitter cache only ($\mu_R=0$)}

\

In the special case when $\mu_R=0$ (transmitter cache only), the achievable NDT for the $3\times 3$ network in Theorem 1 reduces to
\begin{align}
\tau_U(0,\mu_T)=
\left\{
\begin{array}{ll}
13/6-3\mu_T/2, & 1/3\le\mu_T\le2/3\\
3/2-\mu_T/2, &2/3<\mu_T\le1
\end{array}.\label{eqn 33ali}
\right.
\end{align}
In \cite{upperbound}, the authors obtained the inverse of an achievable sum DoF of this network. By Remark \ref{definition ndt}, it is seen that our result \eqref{eqn 33ali} is consistent with that in \cite{upperbound}.

Also, when $\mu_R=0$, by setting $s_1=l,s_2=N_R-l$ in \eqref{eqn converse} of Theorem \ref{thm 2}, a loosened lower bound is
\begin{align}
\tau^*\ge \max\limits_{l=1,2,\cdots,\min\{N_T,N_R\}}\frac{1}{l}&\left(
N_R-(N_T-l)(N_R-l)\mu_T\right).\notag
\end{align}
which is the same as the lower bound $\delta^*(\mu)$ in \cite{lowerbound}.

\subsubsection{Full transmitter cache ($\mu_T=1$)}

\

When $\mu_T=1$, each transmitter can cache all the files and hence can fully cooperate with each other. The network can thus be viewed as a virtual broadcast channel as in \cite{fundamentallimits} except that the server (transmitter) has $N_T$ distributed antennas.

When $N_T\ge N_R$, by Corollary \ref{coro optimality}, the optimal NDT of the $N_T \times N_R$ network is $\tau^*=1-\mu_R$. This can be achieved by letting $a^*_{N_R,0}=\mu_R,a^*_{0,N_R}=\frac{1-\mu_R}{\binom{N_T}{N_R}}$ and others being 0 in \eqref{eqn tmin}. By comparing to the result in \cite{fundamentallimits}, i.e. the lower convex envelop of points $\tau=\frac{N_R(1-\mu_R)}{1+N_R\mu_R}$ at $\mu_R\in\{0,1/N_R,2/N_R,\ldots,1\}$, it is seen that our NDT is better when $0\le\mu_R<1-\frac{1}{N_R}$ and the same when $1-\frac{1}{N_R}\le \mu_R \le1$.

When $N_T<N_R$, by Theorem \ref{thm 1}, an achievable NDT (not necessarily optimal) at certain $\mu_R$ is
\begin{align}
\tilde{\tau}_U=
\left\{
\begin{array}{ll}
\frac{N_R(1-\mu_R)}{N_T+N_R\mu_R},&\mu_R\in\{0,\frac{1}{N_R},\frac{2}{N_R}\ldots,\\
&\qquad\quad\frac{N_R-N_T-2}{N_R}\}\\
1-\mu_R+\frac{1}{N_T\binom{N_R}{N_R\mu_R}},&\mu_R=\frac{N_R-N_T-1}{N_R}\\
1-\mu_R,&\mu_R\in\{\frac{N_R-N_T}{N_R},\frac{N_R-N_T+1}{N_R},\\
&\qquad\quad \frac{N_R-N_T+2}{N_R},\ldots,1\}
\end{array},\notag
\right.
\end{align}
by letting $a_{N_R\mu_R,N_T}=\frac{1}{\binom{N_R}{N_R\mu_R}}$ and others being 0 in \eqref{eqn tmin}. By comparing to the result in \cite{fundamentallimits}, i.e., $\tau=\frac{N_R(1-\mu_R)}{1+N_R\mu_R}$ at these points, it is seen that our NDT is better when $\mu_R\in\{0,\frac{1}{N_R},\ldots,\frac{N_R-2}{N_R}\}$ and the same when $\mu_R\in\{\frac{N_R-1}{N_R},1\}$.

The above performance improvements are all due to transmitter cooperation gain.

\section{File Splitting and Cache Placement}\label{section cache placement}
In this section, we propose a novel file splitting and cache placement scheme for any given normalized cache sizes $\mu_R$ and $\mu_T$ and any transmitter and receiver node numbers $N_T$ and $N_R$. This scheme is the basis of the proofs of all the achievable NDTs.

In this work, we treat all the files equally without taking file popularity into account. Thus, each file will be split and cached in the same manner. Without loss of generality, we focus on the splitting and caching of file $W_i$ for any $1\le i\le L$. Since each bit of the file is either cached or not cached at every node, there are $2^{N_T+N_R}$ possible cache states for each bit. Not every cache state is, however, legitimate. In specific, every bit of the file must be cached in at least one node. In addition, every bit that is not cached simultaneously in all receivers must be cached in at least one transmitter.\footnote{This is because we do not allow receiver cooperation and all the messages must be sent from the transmitters.} As such, the total number of feasible cache states for each bit is given by $2^{N_T+N_R}-\binom{N_T}{0}\sum^{N_R-1}_{r=0}\binom{N_R}{r}=\sum_{r=0}^{N_R} \sum_{t=1}^{N_T}\binom{N_R}{r}\binom{N_T}{t}+1$. Now we can partition each $W_i$ into $\sum_{r=0}^{N_R} \sum_{t=1}^{N_T}\binom{N_R}{r}\binom{N_T}{t}+1$ subfiles exclusively, each associated with one unique cache state and with possibly different length.

Define receiver subset $\Phi\subseteq [N_R]$ and transmitter subset $\Psi \subseteq [N_T]$. Then, denote $W_{i,\textrm{R}_\Phi, \textrm{T}_\Psi}$ as the subfile of $W_i$ cached in receiver subset $\Phi$ and transmitter subset $\Psi$. For example, $W_{i,\textrm{R}_{12},\textrm{T}_{12}}$ is the subfile cached in receivers 1 and 2 and transmitters 1 and 2,  $W_{i,\textrm{R}_\emptyset,\textrm{T}_{123}}$ is the subfile cached in none of the receivers but in transmitters 1, 2 and 3. Similarly, we denote $W_{i,\textrm{R}_\Phi}$ as the collection of the subfiles of $W_i$ that are cached in receiver subset $\Phi$, i.e., $W_{i,\textrm{R}_\Phi}=\bigcup\limits_{\Psi}W_{i,\textrm{R}_\Phi,\textrm{T}_\Psi}$. We assume that the subfiles that are cached in the same number of transmitters and the same number of receivers have the same size. Due to the symmetry of all the nodes as well as the independence of all files, this assumption is valid and does not lose any generality. Thus, we denote the size of $W_{i,\textrm{R}_\Phi,\textrm{T}_\Psi}$ by $a_{r,t}F$, where $r=|\Phi|, t=|\Psi|$, and $a_{r,t}\in[0,1]$ is the file splitting ratio to be optimized later. For example, the size of $W_{i,\textrm{R}_{12},\textrm{T}_{12}}$ is $a_{2,2}F$, and the size of $W_{i,\textrm{R}_\emptyset,\textrm{T}_{123}}$ is $a_{0,3}F$. Here, the file splitting ratios $\{a_{r,t}\}$ should satisfy the following constraints:

\begin{empheq}[left=\empheqlbrace]{align}
&\sum_{r=0}^{N_R} \sum_{t=1}^{N_T}\binom{N_R}{r}\binom{N_T}{t}a_{r,t}+a_{N_R,0}=1,\label{eqn:total cache}\\
&\sum_{r=1}^{N_R} \sum_{t=1}^{N_T}\binom{N_R-1}{r-1}\binom{N_T}{t}a_{r,t}+a_{N_R,0}\le\mu_R,\label{eqn:receiver cache}\\
&\sum_{r=0}^{N_R} \sum_{t=1}^{N_T}\binom{N_R}{r}\binom{N_T-1}{t-1}a_{r,t}\le\mu_T.\label{eqn:transmitter cache}
\end{empheq}
Constraint \eqref{eqn:total cache} comes from the file size limit. This is because for each file, the number of its subfiles cached in $r$ out of $N_R$ receivers and $t$ out of $N_T$ transmitters is given by $\binom{N_R}{r}\binom{N_T}{t}$ and they all have same length of $a_{r,t}F$ bits, for $r =0,1,\cdots,N_R$ and $t=1,2,\cdots,N_T$ or $(r=N_R,t=0)$. Constraint \eqref{eqn:receiver cache} comes from the receiver cache size limit. This is because for each receiver, the total number of subfiles it caches is given by $\sum_{r=1}^{N_R} \sum_{t=1}^{N_T}\binom{N_R-1}{r-1}\binom{N_T}{t}+1$. Among them, there are $\binom{N_R-1}{r-1}\binom{N_T}{t}$ subfiles with length of $a_{r,t}F$ bits, for $r =1,2,\cdots,N_R$ and $t=1,2,\cdots,N_T$, and there is only one subfile with length of $a_{N_R,0}F$ bits. Likewise, constraint \eqref{eqn:transmitter cache} comes from the transmitter cache size limit. This is because for each transmitter, the total number of subfiles it caches is given by $\sum_{r=0}^{N_R} \sum_{t=1}^{N_T}\binom{N_R}{r}\binom{N_T-1}{t-1}$. Among them, there are $\binom{N_R}{r}\binom{N_T-1}{t-1}$ subfiles with length of $a_{r,t}F$ bits for $r =0,1,\cdots,N_R$ and $t=1,2,\cdots,N_T$.

\begin{remark}[Integer points and equal file splitting]\label{remark integer point}
Consider the special case where $(\mu_R, \mu_T)$ satisfies $N_R\mu_R = m$ and $N_T \mu_T = n$ with $m$ and $n$ being any integers. These normalized cache size values are referred to as \textit{integer points}, where every bit of each file can be cached simultaneously at $m$ receivers and $n$ transmitters \emph{on average}. The authors in \cite{bothcache} proposed to split each file so that each bit is cached \emph{exactly} at $m$ receivers and $n$ transmitters. We refer to this file spitting scheme as \emph{equal file splitting}. For example, in a $3\times 3$ interference network at integer point ($\mu_R = 1/3$, $\mu_T = 2/3$), they proposed to partition each file equally into nine disjoint subfiles, each of fractional size $a_{1,2}=1/9$, then place each subfile at exactly one receiver and two transmitters. Such file spitting and cache placement method is, however, not unique. Alternatively, we can partition each file into two subfiles, one cached at all three transmitters but not any receiver with fractional size $a_{0,3}=2/3$ and the other cached at all three receivers but not any transmitter with fractional size $a_{3,0}=1/3$. As we will show in Section \ref{section delivery}, the two file splitting and caching strategies achieve the same NDT. Through this example, it can be seen that our proposed file splitting and cache placement strategy is more general. It offers the freedom to adjust the file splitting ratios for caching gain optimization as will be discussed in Section \ref{section delivery}.
\end{remark}

\section{Delivery Scheme and Caching Optimization}\label{section delivery}
In this section, we prove the achievability of the NDT in Theorem \ref{thm 1}. The main idea of the proof is to design the delivery phase given the file splitting and caching strategy presented in Section \ref{section cache placement}, and then to compute and minimize the achievable NDT by optimizing the file splitting ratios. We consider the worst-case scenario where each receiver requests a distinct file. When some receivers request the same file, the proposed delivery strategy can still be applied either directly or by treating the requests as being different. Without loss of generality, we assume that receiver $q$ ($q=1,2,\ldots,N_R$) desires $W_q$ in the delivery phase. In specific, receiver $q$ wants subfiles $\{W_{q,\textrm{R}_\Phi,\textrm{T}_\Psi}: q\notin \Phi\}$. We divide these subfiles into $N_RN_T$ groups according to the number of transmitters and receivers where they are cached, or equivalently, their fractional file sizes $\{a_{r,t}\}$. There are $N_R\binom{N_R-1}{r}\binom{N_T}{t}$ subfiles in the group associated with $a_{r,t}$, and each receiver desires $\binom{N_R-1}{r}\binom{N_T}{t}$ subfiles of them. Each group of subfiles is delivered individually in the time division manner. Without loss of generality, we present the delivery strategy of an arbitrary group of subfiles with fractional file size $a_{r,t}$ in this section, for $0\le r\le N_R-1,1\le t\le N_T$. As will be clear in the following subsection, the cache states of these subfiles transform the original interference network into cooperative X-multicast channels, and therefore exploit transmitter cooperation gain and coded multicasting gain, apart from local caching gain.

\subsection{Delivery of Subfiles in the Group with $a_{r,t}$}
Note that each subfile in the same group with fractional file size $a_{r,t}$ is desired by one receiver, and already cached at $r$ different receivers and $t$ different transmitters. Coded multicasting approach can be used in the delivery phase through bit-wise XOR, similar with \cite{fundamentallimits}. In specific, given an arbitrary receiver subset $\Phi^+$ with size $|\Phi^+|=r+1$ and an arbitrary transmitter subset $\Psi$ with size $t$, each transmitter in $\Psi$ generates the coded message $\bigoplus\limits_{q\in\Phi^+} W_{q,\textrm{R}_{\Phi^+\backslash\{q\}},\textrm{T}_\Psi}$ desired by all receivers in $\Phi^+$. Note that each receiver $q$ in $\Phi^+$ has cached subfiles $W_{q',\textrm{R}_{\Phi^+\backslash\{q'\}},\textrm{T}_\Psi}$ for $q'\in\Phi^+\setminus\{q\}$, and thus can successfully decode its desired subfile $W_{q,\textrm{R}_{\Phi^+\backslash\{q\}},\textrm{T}_\Psi}$ from the coded message $\bigoplus\limits_{q\in\Phi^+} W_{q,\textrm{R}_{\Phi^+\backslash\{q\}},\textrm{T}_\Psi}$. Through this coded multicasting approach, $r+1$ different subfiles are combined into a single coded message via XOR, and there are only $\binom{N_R}{r+1}\binom{N_T}{t}$ coded messages to be transmitted in total, each available at $t$ transmitters and desired by $r+1$ receivers. We define the channel with such message flow formally as below.
\begin{definition}\label{def channel}
The channel characterized as follows is referred to as the $\binom{N_T}{t}\times\binom{N_R}{r+1}$ cooperative X-multicast channel:
\begin{enumerate}
  \item there are $N_R$ receivers and $N_T$ transmitters;
  \item each set of $r+1$ ($r< N_R$) receivers forms a receiver multicast group;
  \item each set of $t$ ($t\le N_T$) transmitters forms a transmitter cooperation group;
  \item each transmitter cooperation group has an independent message to send to each receiver multicast group.
\end{enumerate}
\end{definition}
In the special case when $(r,t)=(0,1)$ (or $(r,t)=(0,N_T)$), the cooperative X-multicast channel reduces to the X channel (or MISO broadcast channel). When $t=1$, the channel reduces to the $(r+1)$-multicast X-channel defined in \cite{niesen}. The achievable DoF of the $\binom{N_T}{t}\times\binom{N_R}{r+1}$ cooperative X-multicast channel is presented in the following lemma.
\begin{lemma} \label{lemma dof}
The achievable per-user DoF of the $\binom{N_T}{t}\times\binom{N_R}{r+1}$ cooperative X-multicast channel is
\begin{align}
d_{r,t}=\left\{
\begin{array}{ll}
1, &r+t\ge N_R\\
\frac{\binom{N_R-1}{r}\binom{N_T}{t}t}{\binom{N_R-1}{r}\binom{N_T}{t}t+1}, &r+t= N_R-1\\
\max\left\{d'_{r,t},\frac{r+t}{N_R}\right\},& r+t\le N_R-2
\end{array},
\right.\label{lemma d}
\end{align}
where
\begin{align}
d'_{r,t}\triangleq\max\limits_{1\le t'\le t}\left\{\frac{\binom{N_R-1}{r}\binom{N_T}{t'}\binom{N_R-r-1}{t'-1}t'}
{\binom{N_R-1}{r}\binom{N_T}{t'}\binom{N_R-r-1}{t'-1}t'+\binom{N_R-1}{r+1}\binom{N_R-r-2}{t'-1}\binom{N_T}{t'-1}}\right\}
\end{align}
\end{lemma}
\begin{proof}
We present the main idea of the proof here. The detailed proof is given in Appendix A. Since each message can be cooperatively transmitted by $t$ transmitters, the interference it may cause to a maximum of $t-1$ undesired receivers can be neutralized through interference neutralization. When the actual number of undesired receivers for each message, $N_R-r-1$, does not exceed $t-1$, i.e. $N_R \le r+t$, by interference neutralization, each receiver only receives its desired messages with all undesired messages neutralized out. Therefore, a per-user DoF of 1 can be achieved. On the other hand, when $N_R>r+t$, each message will still cause interference to $N_R-r-t$ undesired receivers after interference neutralization. In this case, asymptotic  interference alignment is further applied  by partitioning the interference messages into groups and aligning the interferences from the same group in a same subspace at each undesired receiver, so as to achieve the per-user DoF in \eqref{lemma d}.
\end{proof}

\begin{remark}[Sum DoF]\label{rmk sum dof}
Since each message is desired by $r+1$ receivers in the $\binom{N_T}{t}\times\binom{N_R}{r+1}$ cooperative X-multicast channel, the achievable sum DoF of this channel is given by $d_{\textrm{sum}}=\frac{N_R}{r+1}d_{r,t}$, where $d_{r,t}$ is given in Lemma \ref{lemma dof}.
\end{remark}

\begin{remark}[Optimality of DoF]\label{rmk dof}
The achievable per-user DoF in Lemma \ref{lemma dof} is optimal in certain cases. In specific,
\begin{enumerate}
  \item When $(r,t)=(0,1)$, Lemma \ref{lemma dof} reduces to  $d_{0,1}=\frac{N_T}{N_T+N_R-1}$, which is optimal for the X channel \cite{Xchannel}.
  \item When $(r,t)=(0,N_T)$, Lemma \ref{lemma dof} reduces to $d_{0,N_T}=\min\{\frac{N_T}{N_R},1\}$, which is optimal for the MISO broadcast channel \cite{BCchannel}.
  \item When $r+t\ge N_R$, Lemma \ref{lemma dof} reduces to $d_{r,t}=1$, which is optimal for the considered channel. The converse can be proved easily by using a cut-set bound at each receiver.
\end{enumerate}
Note that the optimality of the DoF results in the above cases is a part of the reason that the achievable NDT is optimal under the conditions in Corollary \ref{coro optimality}. More specifically, it is observed from the proof of Corollary \ref{coro optimality} in Appendix B that the optimal file splitting ratios satisfy $a^*_{0,N_R}>0,a^*_{N_R-N_T,N_T}>0,a^*_{0,N_T}>0,a^*_{0,1}>0$ in the corresponding cache size regions.
\end{remark}

In the special case when $t=1$, Lemma \ref{lemma dof} reduces to the DoF of the $(r+1)$-multicast X-channel in \cite{niesen}, i.e., $d_{r,1}=\frac{N_T\binom{N_R-1}{r}}{N_T\binom{N_R-1}{r}+\binom{N_R-1}{r+1}}$.  When $(N_R,N_T)=(3,3)$ and $(r,t)=(0,2)$, Lemma \ref{lemma dof} reduces to the DoF of the cache-aided interference channel in \cite{upperbound} when $(\mu_R,\mu_T)=(0,2/3)$, i.e., $d_{0,2}=6/7$. When $(N_R,N_T)=(3,3)$ and $(r,t)=(1,1)$, $(r,t)=(1,2)$, $(r,t)=(1,3)$, Lemma \ref{lemma dof} reduces to the DoF of the hybrid X-multicast channel, partially cooperative X-multicast channel, and fully cooperative X-multicast channel in \cite{mine}, i.e. $d_{1,1}=6/7$, $d_{1,2}=1$, and $d_{1,3}=1$, respectively.

Since the channel formed by the delivery of the group with $a_{r,t}$ is the $\binom{N_T}{t}\times \binom{N_R}{r+1}$ cooperative X-multicast channel, and there are $\binom{N_R}{r+1}\binom{N_T}{t}$ coded messages to deliver, with each receiver desiring $\binom{N_R-1}{r}\binom{N_T}{t}$ of them, by Lemma \ref{lemma dof}, we can obtain the NDT of this group directly as $\tau_{r,t}=\frac{\binom{N_R-1}{r}\binom{N_T}{t}}{d_{r,t}}a_{r,t}$.

\subsection{Optimization of Splitting Ratios}

Summing up the NDTs obtained in the previous subsection for all groups, we can obtain the total NDT in the delivery phase:
\begin{align}
\tau=\sum_{r=0}^{N_R-1}\sum_{t=1}^{N_T}\frac{\binom{N_R-1}{r}\binom{N_T}{t}}{d_{r,t}}a_{r,t}.\label{eqn tau v}
\end{align}
We then optimize the file splitting ratios $\{a_{r,t}\}$ to minimize the total NDT subject to constraints \eqref{eqn:total cache}\eqref{eqn:receiver cache}\eqref{eqn:transmitter cache}. This is expressed as the LP problem shown in Theorem 1, where the constraints \eqref{eqn tmin1}\eqref{eqn tmin2}\eqref{eqn tmin3} are the same as \eqref{eqn:total cache}\eqref{eqn:receiver cache}\eqref{eqn:transmitter cache}, and $d_{r,t}$ in \eqref{eqn tau ip} is the same as \eqref{lemma d}. Thus, Theorem \ref{thm 1} is proved.

In the following corollaries, we present the closed-form and optimal solutions in Theorem \ref{thm 1} when $N_T=N_R=2$ and $N_T=N_R=3$ by using linear equation substitutions and other manipulations. The detailed computation for $N_T=N_R=2$ is given in Appendix D. The computation for $N_T=N_R=3$ is similar and omitted.

\begin{corollary}\label{coro 2x2}
For the cache-aided $2\times2$ interference network, the minimum NDT is upper bounded by
\begin{align}
\tau^*(\mu_R,\mu_T)\le\tau_U=
\left\{
\begin{array}{ll}
1-\mu_R, & (\mu_R,\mu_T)\in \mathcal{R}^1_{22}\\
2-2\mu_R-\mu_T, & (\mu_R,\mu_T)\in \mathcal{R}^2_{22}
\end{array},
\right.
\end{align}
where $\{\mathcal{R}_{22}^i\}^2_{i=1}$ are given below and sketched in Fig.~\ref{Fig regionnt2}.
\begin{align}
\left\{
\begin{array}{ll}
\mathcal{R}^1_{22}=\left\{(\mu_R,\mu_T): \mu_R+\mu_T\ge1,\mu_R\le1,\mu_T\le1\right\}\\
\mathcal{R}^2_{22}=\left\{(\mu_R,\mu_T): \mu_R+\mu_T<1,\mu_R\ge0,\mu_R+2\mu_T\ge1\right\}
\end{array}\notag
\right.
\end{align}
\end{corollary}

\begin{proof}
See Appendix D.
\end{proof}

\begin{figure}[tbp]
\centering
\includegraphics[scale=0.4]{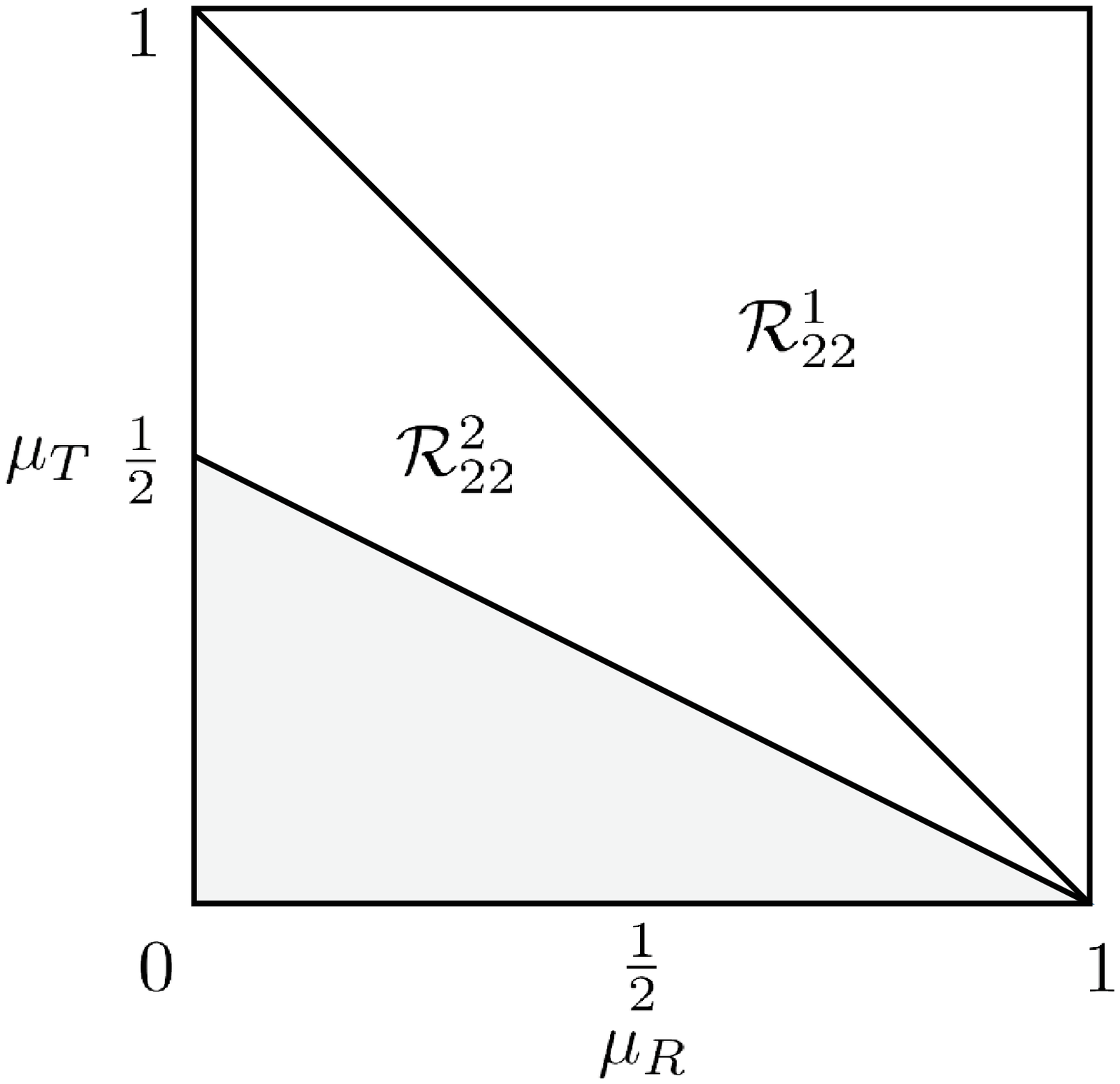}
\caption{Cache size regions in the $2\times 2$ network.}
\label{Fig regionnt2}
\end{figure}

\begin{corollary}\label{coro 3x3}
For the cache-aided $3\times3$ interference network, the minimum NDT is upper bounded by
\begin{align}
\tau^*(\mu_R,\mu_T)\le\tau_U=
\left\{
\begin{array}{ll}
1-\mu_R, & (\mu_R,\mu_T)\in \mathcal{R}^1_{33}\\
\frac{4}{3}-\frac{4}{3}\mu_R-\frac{1}{3}\mu_T, & (\mu_R,\mu_T)\in \mathcal{R}^2_{33}\\
\frac{3}{2}-\frac{5}{3}\mu_R-\frac{1}{2}\mu_T, & (\mu_R,\mu_T)\in \mathcal{R}^3_{33}\\
\frac{13}{6}-\frac{8}{3}\mu_R-\frac{3}{2}\mu_T, & (\mu_R,\mu_T)\in \mathcal{R}^4_{33}\\
\frac{8}{3}-\frac{8}{3}\mu_R-3\mu_T, & (\mu_R,\mu_T)\in \mathcal{R}^5_{33}
\end{array}\label{eqn 3x3 t}
\right.
\end{align}
where $\{\mathcal{R}^i_{33}\}^5_{i=1}$ are given below and sketched in Fig.~\ref{Fig region33}.
\begin{align}
\left\{
\begin{array}{ll}
\mathcal{R}^1_{33}=\{(\mu_R,\mu_T): \mu_R+\mu_T\ge1, \mu_R\le1, \mu_T\le1\}\\
\mathcal{R}^2_{33}=\{(\mu_R,\mu_T): \mu_R+\mu_T<1, 2\mu_R+\mu_T\ge1, \\
\qquad\qquad\qquad\qquad\mu_R+2\mu_T>1\}\\
\mathcal{R}^3_{33}=\{(\mu_R,\mu_T): 3\mu_R+3\mu_T\ge2, 2\mu_R+\mu_T<1, \\
\qquad\qquad\qquad\qquad\mu_R\ge0\}\\
\mathcal{R}^4_{33}=\{(\mu_R,\mu_T): 3\mu_R+3\mu_T<2, \mu_R\ge0, 3\mu_T>1\}\\
\mathcal{R}^5_{33}=\{(\mu_R,\mu_T): 3\mu_T\le1, \mu_R+2\mu_T\le1, \\
\qquad\qquad\qquad\qquad\mu_R+3\mu_T\ge1\}
\end{array}.\notag
\right.
\end{align}
\end{corollary}

\begin{proof}
The proof is similar to that of Corollary \ref{coro 2x2} and hence ignored. We only present the optimal solution of file splitting ratios in Appendix E.
\end{proof}

\begin{figure}[tbp]
\centering
\includegraphics[scale=0.4]{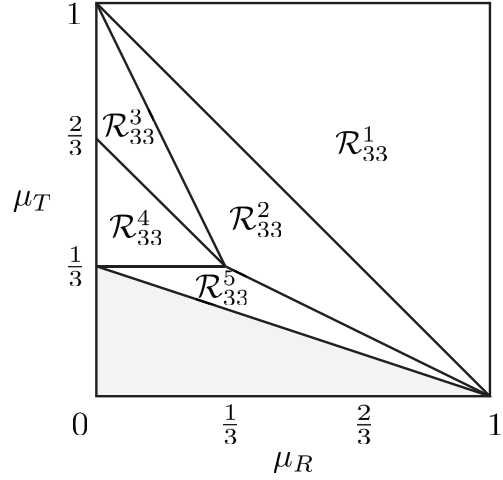}
\caption{Cache size regions in the $3\times 3$ network.}\label{Fig region33}
\end{figure}

It can be seen from Corollary \ref{coro 2x2} and Corollary \ref{coro 3x3} that the achievable NDT is a piece-wise linearly decreasing function of $\mu_R$ and $\mu_T$. The number of piece-wise regions depends on $N_R$ and $N_T$.

\section{Discussion on the Achievable Scheme}\label{section discussion}
In this section, we provide some discussions on our proposed caching and delivery scheme, which will offer some insights into the impact of caching in the considered interference networks.

\subsection{On Caching at Integer Points}
Consider an arbitrary integer point $(\mu_R=m/N_R,\mu_T=n/N_T)$ with $m$ and $n$ being any integers. We first evaluate the achievable NDT by adopting the equal file splitting strategy.\footnote{See Remark \ref{remark integer point} in Section \ref{section cache placement} for the definition of integer points.} In the equal file splitting strategy, each file is split into $\binom{N_R}{m}\binom{N_T}{n}$ equal-sized subfiles, each cached in $m$ receivers and $n$ transmitters. Then, we have $a_{m,n}=\frac{1}{\binom{N_R}{m}\binom{N_T}{n}}$ and all the rest $a_{r,t}=0$. The proposed delivery scheme introduced before then transforms the network topology to the $\binom{N_T}{n}\times \binom{N_R}{m+1}$ cooperative X-multicast channel. By \eqref{eqn tau v}, the achievable NDT can be expressed in a unified form as
\begin{align}
\tau_{m,n}=\frac{1-\mu_R}{d_{m,n}},\label{eqn equal tau}
\end{align}
where $d_{m,n}$ is the per-user DoF of the formed cooperative X-multicast channel given in \eqref{lemma d}. By comparing \eqref{eqn equal tau} and \eqref{eqn tmin}, the achievable NDT at any $(\mu_R,\mu_T)$ in Theorem \ref{thm 1} can be regarded as the  convex envelope of the achievable NDTs $\{\tau_{m,n}\}$ at all integer points with combination coefficients $\{\beta_{m,n}\triangleq\binom{N_T}{n}\binom{N_R}{m}a_{m,n}:(m,n)\in \mathcal{A}\}$.

Based on Remark \ref{rmk sum dof}, we can rewrite \eqref{eqn equal tau} as
\begin{align}
\tau_{m,n}=\frac{N_R(1-\mu_R)}{(m+1)d_{\textrm{sum}}}.\label{eqn tausumdof}
\end{align}
The expression of NDT in \eqref{eqn tausumdof} reveals the gains of caching more explicitly. The term $(1-\mu_R)$ denotes the receiver local caching gain, since each receiver has already cached a fraction $\mu_R$ of its desired file. The term $(m+1)$ denotes the coded multicasting gain, since by our caching and delivery scheme, each coded message is needed by $m+1$ different receivers.  The DoF term, $d_{\textrm{sum}}$, reflects the cache-induced transmitter cooperation gain via interference neutralization and interference alignment.

At an arbitrary cache size tuple $(\mu_R,\mu_T)$, these gains can also be exploited and reflected by the file splitting ratios $\{a_{r,t}\}$ of the corresponding cache states. When the optimal solution of the LP problem satisfies $a^*_{r,t}>0$ for some $(r,t)$, it means that there exist subfiles cached in $r$ receivers and $t$ transmitters in the cache placement phase. As shown in our proposed delivery scheme in Section \ref{section delivery}-A, local caching gain is exploited when $a^*_{r,t}>0$ for some $r>0$, coded multicasting gain is exploited when $a^*_{r,t}>0$ for some $0<r<N_R$, and transmitter cooperation gain is exploited when $a^*_{r,t}>0$ for some $t\ge1$. For example, we can only exploit local caching gain and transmitter cooperation gain in cache size region $\mathcal{R}^1_{33}$ in the $3\times3$ network, since our solution satisfies $a^*_{0,3},a^*_{3,0}>0$ and all the rest $a^*_{r,t}=0$ as shown in Appendix E.

\subsection{On the Optimal File Splitting Ratios}
It is important to realize from the previous two sections that the optimal file slitting ratios for NDT minimization at given cache size $(\mu_R, \mu_T)$ are not unique. Mathematically, this is quite expected since a linear programming problem like \eqref{eqn tmin} in general does not have unique solutions. However, the physical meaning of each solution can vary dramatically. Let us take the \textit{integer point} $(\mu_R=\frac{1}{3},\mu_T=\frac{2}{3})$ in the $3\times3$ network for example. According to Corollary \ref{coro 3x3}, there are two optimal solutions for the file splitting ratios. One is $a^*_{0,3} =\frac{2}{3}$, $a^*_{3,0}=\frac{1}{3}$ with the rest $a^*_{r,t}=0$. This solution means that each file is split into two subfiles, one has fractional size $a^*_{0,3} =\frac{2}{3}$ and is cached simultaneously at all three transmitters but none of the receivers, the other subfile has fractional size $a^*_{3,0}=\frac{1}{3}$ and is cached simultaneously at all three receivers but none of the transmitters. From the proposed delivery scheme in Section \ref{section delivery}, it is seen that this solution enjoys both receiver local caching gain and transmitter cooperation gain. In particular, the transmitter cooperation turns the interference network into a MISO broadcast channel with per-user DoF of 1.

Another feasible solution is $a^*_{1,2}=\frac{1}{9}$ with all the rest $a^*_{r,t}=0$. In this solution, each file is split into 9 subfiles, each with the same fractional size $a^*_{1,2}=\frac{1}{9}$ and cached at one receiver and two transmitters. From the proposed delivery scheme in Section \ref{section delivery}, this solution enjoys the coded multicasting gain and transmitter cooperation gain by turning the network topology into a partially cooperative X-multicast channel. Together with Corollary \ref{coro optimality} in Section \ref{section main results}, both file splitting schemes are globally optimal in terms of achieving the minimum NDT $\tau^* =\frac{2}{3}$.

It is interesting to note that the second file splitting scheme is the same as the one proposed in \cite{bothcache}. However, the delivery strategy is different. In \cite{bothcache}, the network topology is turned into a partially cooperative interference channel with side information, and interference neutralization is used to achieve the sum DoF of 3. Given that each receiver already caches 3 out of the 9 subfiles of its desired file and only needs the other 6 subfiles in the delivery phase, we can compute the total delivery time as $T = \frac{3 \times 6 \times a_{1,2} F }{ 3\times \log P}$. As such, the corresponding NDT is $\tau =\frac{2}{3}$, which is the same as ours.

In general, we find that at integer points $(\mu_R=\frac{m}{N_R},\mu_T=\frac{n}{N_T})$, with $m+n\ge N_R$, the optimal file splitting ratios and the delivery scheme are not unique. In specific, when $N_T\ge N_R$, besides the equal file splitting strategy, the optimal ratios can also be $a^*_{N_R,0}=\mu_R,a^*_{0,N_R}=\frac{1-\mu_R}{\binom{N_T}{N_R}}$. This solution means that each file is split into $1+\binom{N_T}{N_R}$ subfiles, one has fractional size $a^*_{N_R,0}=\mu_R$ and is cached simultaneously at all receivers but none of the transmitters, and each of the other subfiles has fractional size $a^*_{0,N_R}=\frac{1-\mu_R}{\binom{N_T}{N_R}}$ and is cached simultaneously at $N_R$ out of $N_T$ different transmitters but none of the receivers. According to the delivery schemes proposed in Section \ref{section delivery}, the network topology becomes the $\binom{N_T}{N_R}\times \binom{N_R}{1}$ cooperative X-multicast channel whose per-user DoF is 1. Thus, we can obtain the NDT as
\begin{align}
\tau=\frac{1-\mu_R}{d}=1-\mu_R.\label{eqn bc tau}
\end{align}
Note that bit-wise XOR is not used in the delivery phase, thus this scheme does not exploit coded multicasting gain. Comparing \eqref{eqn bc tau} to \eqref{eqn equal tau} achieved by equal file splitting strategy, it can be seen that the transmitter cooperation gain obtained in this scheme has the same contribution as the combined coded-multicasting and transmitter cooperation gain in the equal file splitting strategy.

When $N_T<N_R$, bit-wise XOR is applied in the delivery phase to achieve the optimal NDT, because the limited number of transmitters becomes a bottleneck. In this case, the optimal solution can be $a^*_{N_R,0}=1-\frac{N_R}{N_T}(1-\mu_R),a^*_{N_R-N_T,N_T}=\frac{\frac{N_R}{N_T}(1-\mu_R)}{\binom{N_R}{N_R-N_T}}$. This solution means that each file is split into $1+\binom{N_R}{N_R-N_T}$ subfiles, one has fractional size $a^*_{N_R,0}=1-\frac{N_R}{N_T}(1-\mu_R)$ and is cached simultaneously at all receivers but none of the transmitters, and each of the other subfiles has fractional size $a^*_{N_R-N_T,N_T}=\frac{\frac{N_R}{N_T}(1-\mu_R)}{\binom{N_R}{N_R-N_T}}$ and is cached simultaneously at all transmitters and $N_R-N_T$ out of $N_R$ different receivers. In the delivery phase, only the subfiles with fractional size $a^*_{N_R-N_T,N_T}$ are transmitted, and the local caching gain, coded-multicasting gain and transmitter cooperation gain are all exploited.

The multiple choices of file splitting ratios offer more freedoms to choose an appropriate caching and delivery scheme according to different limitations in practical systems, such as transmitter or receiver computation complexity or file splitting constraints.

\subsection{On the Differences with \textrm{\cite{bothcache} and \cite{niesen}}}
Although the similar caching problem is considered in \cite{bothcache}, their performance metric, caching scheme, and conclusion are different from ours. First, we adopt the NDT as the performance metric while \cite{bothcache} used the standard DoF. As we noted in Remark \ref{remark fdt calculation} in Section \ref{section model}, NDT is particularly suitable for the considered network because it reflects not only the load reduction due to receiver cache but also the DoF enhancement due to transmitter cache. In specific, we can express the NDT as $\tau = \frac{R}{d}$ as in \eqref{eqn NDTtau}, where $R$ is the per-user traffic load normalized by each file size. To illustrate this in detail, we consider the integer points $(\mu_R=\frac{1}{3},\mu_T=\frac{2}{3})$ and $(\mu_R=\frac{2}{3},\mu_T=\frac{1}{3})$ in the $3\times3$ interference network. In \cite{bothcache}, they have per-user DoF of 1 at both points. However, the actual delivery time at these two points is different. At point $(\mu_R=\frac{1}{3},\mu_T=\frac{2}{3})$, each file is split into 9 equal-sized subfiles, each cached at one receiver and two transmitters. This corresponds to $a_{1,2}=\frac{1}{9}$ and the rest $a_{r,t}$'s are all 0. Thus, each receiver caches 3 out of 9 subfiles of its desired file and only needs the other 6 subfiles in the delivery phase. The corresponding NDT is $\tau=\frac{6\times a_{1,2}}{1}=\frac{2}{3}$. On the other hand, at point $(\mu_R=\frac{2}{3},\mu_T=\frac{1}{3})$, each file is also split into 9 equal-sized subfiles, but each cached at two receivers and one transmitter. This corresponds to $a_{2,1}=\frac{1}{9}$ and the rest $a_{r,t}$'s are all 0. Thus, each receiver caches 6 out of 9 subfiles of its desired file and only needs the other 3 subfiles in the delivery phase, yielding the corresponding NDT $\tau=\frac{3\times a_{2,1}}{1}=\frac{1}{3}$. Clearly, the DoF alone is unable to fully capture the gains of joint transmitter and receiver caching as NDT does.

Second, the file splitting ratios in \cite{bothcache} are pre-determined at each given cache size tuple $(\mu_R, \mu_T)$ as noted in Remark \ref{remark integer point} of Section \ref{section cache placement}. However, our file splitting ratios are obtained by solving an LP problem at each given cache size tuple and thus are provably optimal under the given caching strategy.

Another difference between our work and \cite{bothcache} is that the transmission scheme in \cite{bothcache} is restricted to one-shot linear processing, while we allow asymptotic interference alignment and interference neutralization to explore the optimal transmission DoF. Due to this difference, the achievable NDT in \cite{bothcache} is different from ours. In specific, the achievable NDT in \cite{bothcache} at an arbitrary integer point $(\mu_R,\mu_T)$ is given by
\begin{align}
\tau=\frac{N_R(1-\mu_R)}{\min\{N_R,N_R\mu_R+N_T\mu_T\}},\label{eqn discussion ali ndt}
\end{align}
based on \cite[Theorem 1]{bothcache}. The achievable NDT in our scheme is given by \eqref{eqn equal tau}. It can be seen that our achievable NDT in \eqref{eqn equal tau} is smaller (hence better) than \eqref{eqn discussion ali ndt} in \cite{bothcache}. For example, consider the integer point $(\mu_R=0,\mu_T=1/3)$ in the $3\times 3$ interference network. According to Corollary \ref{coro optimality}, the achievable NDT in our scheme is optimal and given by $\tau^*=5/3$, which is better than the NDT $\tau=3$ achieved in \cite{bothcache}.

The caching problem with all transmitters and receivers equipped with cache is also considered in a later work \cite{niesen}. Note that the performance metric, 1/DoF, adopted in \cite{niesen} is equivalent to NDT and thus we are able to compare the result directly. In \cite{niesen}, each subfile is only cached at one distinct transmitter during the cache placement phase for all $\mu_T\ge 1/N_T$. Hence it cannot exploit the cache-induced transmitter cooperation through interference neutralization as our scheme when $\mu_T> 1/N_T$. As a result, the achievable NDT in \cite{niesen} is larger (hence worse) than ours at cache size region $\mu_T> 1/N_T$. For example, consider integer point $(\mu_R=1/3,\mu_T=2/3)$ in the $3\times 3$ interference network. According to Corollary \ref{coro optimality}, the achievable NDT in our scheme is optimal and given by $\tau^*=2/3$, which is better than the NDT $\tau=7/9$ achieved in \cite{niesen}.

Last but not least, \cite{bothcache,niesen} are limited to the cache size region $\mu_T\ge\frac{1}{N_T}$ which is only a subset of the feasible cache size region \eqref{eqn boundary} considered in this work.

\section{Lower Bound of the Minimum NDT}\label{section converse}
In this section, we present the proof of the lower bound of the minimum NDT in Theorem \ref{thm 2}. The method of the proof is an extension of the approach in \cite{lowerbound} by taking receiver caches into account.

\begin{figure*}[tbp]
\begin{centering}
\includegraphics[scale=0.27]{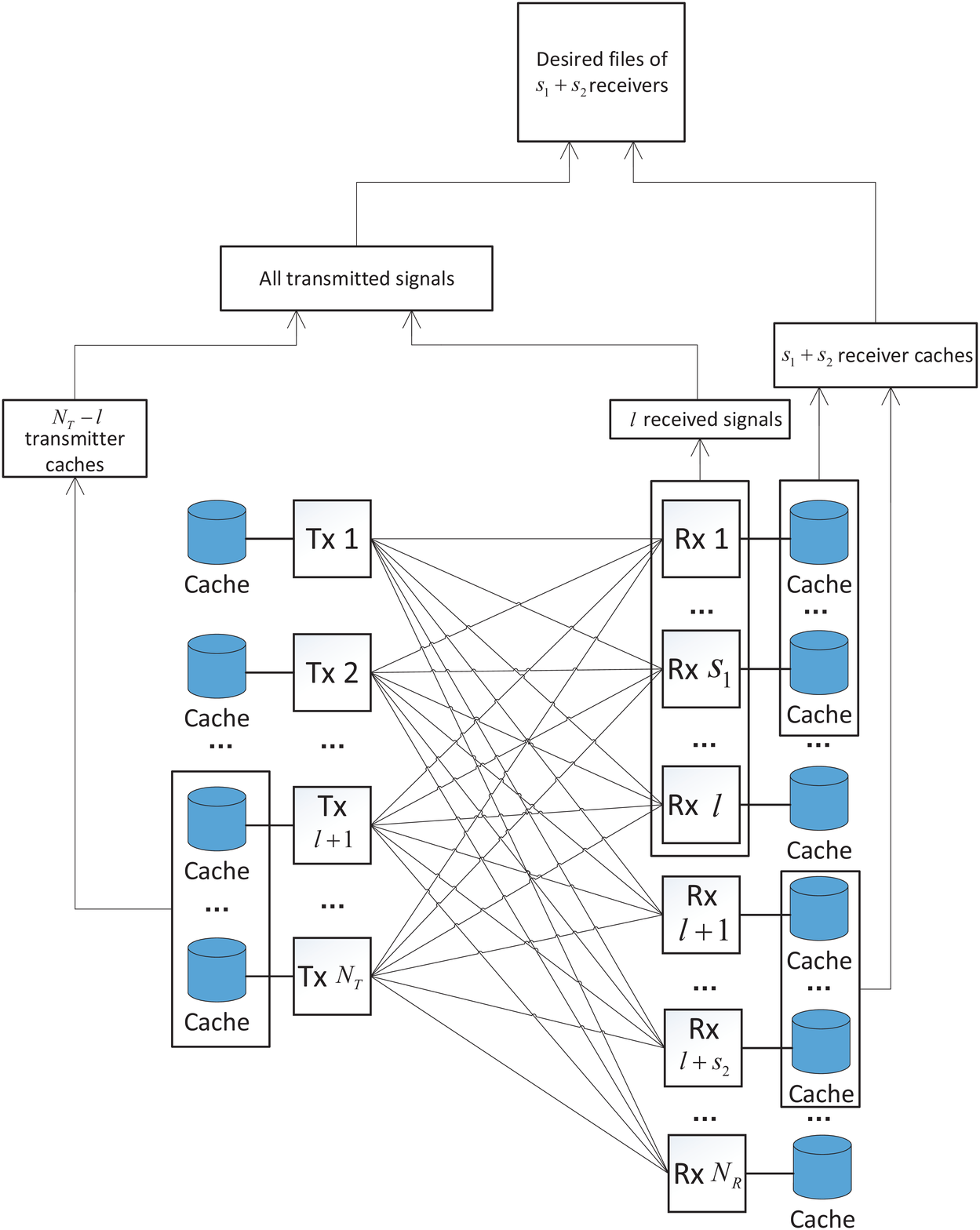}
\caption{Illustration of the proof of the converse.}\label{Fig converse2}
\end{centering}
\end{figure*}

The proof of the lower bound is based on the following statement. As illustrated in Fig.~\ref{Fig converse2}, we divide all the receivers into two groups, where the first group contains $l$ arbitrary receivers and the second group contains the remaining $N_R-l$ receivers. Then, we select $s_1$ ($s_1 \le l$) receivers from the first group as $\mathcal{S}_1$, and $s_2$ ($s_2 \le N_R-l$) receivers from the second group as $\mathcal{S}_2$. \emph{By converse assumption, given the local caches from any $N_T-l$ transmitters, the received signals from the $l$ receivers in the first group, and the local caches from the receivers in $\mathcal{S}_1$ and $\mathcal{S}_2$, then the desired files requested by the receivers in $\mathcal{S}_1$ and $\mathcal{S}_2$ are decodable in the high SNR regime.} More specifically, given the transmitter caches from the $N_T-l$ transmitters, the transmitted signals of these $N_T-l$ transmitters can be constructed. Then, given the $N_T-l$ transmitted signals and the $l$ received signals in the first group, the remaining $l$ unknown transmitted signals can be obtained almost surely \cite[Lemma 3]{lowerbound}, neglecting noise in the high SNR regime. With all the transmitted signals, the received signals of all the receivers can be obtained. Together with the local caches from $\mathcal{S}_1$ and $\mathcal{S}_2$, their desired files are decodable.

To begin the proof, let $\mathbf{d}=(d_1,d_2,\ldots,d_{N_R})$ denote a distinct user demand vector. Using \cite[Lemma 2]{lowerbound} and from the statement above, we obtain 
\begin{align}
&H(W_{d_1\sim d_{s_1}},W_{d_{l+1}\sim d_{l+s_2}}|Y_{1\sim l},U_{l+1\sim N_T},V_{1\sim s_1}, V_{l+1\sim l+s_2},W_{d_{s_1+1}\sim d_{l}},W_{d_{l+s_2+1}\sim d_{N_R}},W_{1\sim L}\setminus W_\mathbf{d})\notag\\
=&F\varepsilon_F+T\varepsilon_P\log P,\label{eqn converse begin}
\end{align}
where $W_\mathbf{d}\triangleq\{W_{d_1},W_{d_2},\ldots,W_{d_{N_R}}\}$, $\varepsilon_P$ is a function of power $P$, and satisfies $\lim_{P\to\infty}\varepsilon_P=0$. Then, the entropy of desired files $\{W_{d_1},\ldots,W_{d_{s_1}},W_{d_{l+1}},\ldots,W_{d_{l+s_2}}\}$ can be expressed as 
\begin{align}
&(s_1+s_2)F\notag\\
=&H(W_{d_1\sim d_{s_1}},W_{d_{l+1}\sim d_{l+s_2}}|W_{d_{s_1+1}\sim d_{l}},W_{d_{l+s_2+1}\sim d_{N_R}},W_{1\sim L}\setminus W_\mathbf{d})\notag\\
=&H(W_{d_1\sim d_{s_1}},W_{d_{l+1}\sim d_{l+s_2}}|Y_{1\sim l},U_{l+1\sim N_T},V_{1\sim s_1},V_{l+1\sim l+s_2},W_{d_{s_1+1}\sim d_{l}},W_{d_{l+s_2+1}\sim d_{N_R}},W_{1\sim L}\setminus W_\mathbf{d})\notag\\
&+I(W_{d_1\sim d_{s_1}},W_{d_{l+1}\sim d_{l+s_2}};Y_{1\sim l},U_{l+1\sim N_T},V_{1\sim s_1},V_{l+1\sim l+s_2}|W_{d_{s_1+1}\sim d_{l}},W_{d_{l+s_2+1}\sim d_{N_R}},W_{1\sim L}\setminus W_\mathbf{d})\notag\\
=&I(W_{d_1\sim d_{s_1}},W_{d_{l+1}\sim d_{l+s_2}};Y_{1\sim l},U_{l+1\sim N_T},V_{1\sim s_1},V_{l+1\sim l+s_2}|W_{d_{s_1+1}\sim d_{l}},W_{d_{l+s_2+1}\sim d_{N_R}},W_{1\sim L}\setminus W_\mathbf{d})\notag\\
&+F\varepsilon_F+T\varepsilon_P\log P\label{eqn converse 1}.
\end{align}
Note that there are $\binom{L}{N_R}N_R!$ distinct demand vectors in total. Then, we have 
\begin{align}
&(s_1+s_2)F\notag\\
=&\frac{1}{\binom{L}{N_R}N_R!}\sum_\mathbf{d}H(W_{d_1\sim d_{s_1}},W_{d_{l+1}\sim d_{l+s_2}}|W_{d_{s_1+1}\sim d_{l}},W_{d_{l+s_2+1}\sim d_{N_R}},W_{1\sim L}\setminus W_\mathbf{d})\notag\\
=&\frac{1}{\binom{L}{N_R}N_R!}\sum_\mathbf{d}I(W_{d_1\sim d_{s_1}},W_{d_{l+1}\sim d_{l+s_2}};Y_{1\sim l},U_{l+1\sim N_T},V_{1\sim s_1},V_{l+1\sim l+s_2}|W_{d_{s_1+1}\sim d_{l}}W_{d_{l+s_2+1}\sim d_{N_R}},\notag\\
&\qquad\qquad\qquad\quad W_{1\sim L}\setminus W_\mathbf{d})+F\varepsilon_F+T\varepsilon_P\log P\label{eqn converse 2}.
\end{align}

In what follows, we prove the bounds \eqref{eqn converse} and \eqref{eqn taul2}, respectively according to whether intra-file coding is allowed or not in the caching functions.

\subsection{Lower Bound \eqref{eqn converse} with Arbitrary Intra-file Coding}
The mutual information in \eqref{eqn converse 2} is upper bounded by 
\begin{subequations}\label{eqn converse 3}
\begin{align}
&\frac{1}{\binom{L}{N_R}N_R!}\sum_\mathbf{d}I(W_{d_1\sim d_{s_1}},W_{d_{l+1}\sim d_{l+s_2}};Y_{1\sim l},U_{l+1\sim N_T},V_{1\sim s_1},V_{l+1\sim l+s_2}|W_{d_{s_1+1}\sim d_{l}},\notag\\
&\qquad\qquad\qquad\quad W_{d_{l+s_2+1}\sim d_{N_R}},W_{1\sim L}\setminus W_\mathbf{d})\notag\\
=&\frac{1}{\binom{L}{N_R}N_R!}\sum_\mathbf{d}I(W_{d_1\sim d_{s_1}},W_{d_{l+1}\sim d_{l+s_2}};Y_{1\sim l}|W_{d_{s_1+1}\sim d_{l}},W_{d_{l+s_2+1}\sim d_{N_R}},W_{1\sim L}\setminus W_\mathbf{d})\notag\\
&+\frac{1}{\binom{L}{N_R}N_R!}\sum_\mathbf{d}I(W_{d_1\sim d_{s_1}},W_{d_{l+1}\sim d_{l+s_2}};U_{l+1\sim N_T},V_{1\sim s_1},V_{l+1\sim l+s_2}|Y_{1\sim l},W_{d_{s_1+1}\sim d_{l}},\notag\\
&\qquad\qquad\qquad\qquad W_{d_{l+s_2+1}\sim d_{N_R}},W_{1\sim L}\setminus W_\mathbf{d})\label{eqn converse 31}\\
=&\frac{1}{\binom{L}{N_R}N_R!}\sum_\mathbf{d}h(Y_{1\sim l}|W_{d_{s_1+1}\sim d_{l}},W_{d_{l+s_2+1}\sim d_{N_R}},W_{1\sim L}\setminus W_\mathbf{d})\notag\\
&+\frac{1}{\binom{L}{N_R}N_R!}\sum_\mathbf{d}H(U_{l+1\sim N_T},V_{1\sim s_1},V_{l+1\sim l+s_2}|Y_{1\sim l},W_{d_{s_1+1}\sim d_{l}},W_{d_{l+s_2+1}\sim d_{N_R}},W_{1\sim L}\setminus W_\mathbf{d})\notag\\
&-\frac{1}{\binom{L}{N_R}N_R!}\sum_\mathbf{d}h(Y_{1\sim l}|W_{1\sim L})-\frac{1}{\binom{L}{N_R}N_R!}\sum_\mathbf{d}H(U_{l+1\sim N_T},V_{1\sim s_1},V_{l+1\sim l+s_2}|Y_{1\sim l},W_{1\sim L})\label{eqn converse 32}\\
\le &\frac{1}{\binom{L}{N_R}N_R!}\sum_\mathbf{d}h(Y_{1\sim l}|W_{d_{s_1+1}\sim d_{l}},W_{d_{l+s_2+1}\sim d_{N_R}},W_{1\sim L}\setminus W_\mathbf{d})\notag\\
&+\frac{1}{\binom{L}{N_R}N_R!}\sum_\mathbf{d}H(U_{l+1\sim N_T},V_{1\sim s_1},V_{l+1\sim l+s_2}|Y_{1\sim l},W_{d_{s_1+1}\sim d_{l}},W_{d_{l+s_2+1}\sim d_{N_R}},W_{1\sim L}\setminus W_\mathbf{d})\label{eqn converse 33}\\
\le &\frac{1}{\binom{L}{N_R}N_R!}\sum_\mathbf{d}lT\log(2\pi e(c\cdot P+1))\notag\\
&+\frac{1}{\binom{L}{N_R}N_R!}\sum_\mathbf{d}H(U_{l+1\sim N_T}|V_{1\sim s_1},V_{l+1\sim l+s_2},Y_{1\sim l},W_{d_{s_1+1}\sim d_{l}},W_{d_{l+s_2+1}\sim d_{N_R}},W_{1\sim L}\setminus W_\mathbf{d})\notag\\
&+\frac{1}{\binom{L}{N_R}N_R!}\sum_\mathbf{d}H(V_{1\sim s_1},V_{l+1\sim l+s_2}|Y_{1\sim l},W_{d_{s_1+1}\sim d_{l}},W_{d_{l+s_2+1}\sim d_{N_R}},W_{1\sim L}\setminus W_\mathbf{d})\label{eqn converse 34}\\
\le &lT\log(2\pi e(c\cdot P+1))+\frac{1}{\binom{L}{N_R}N_R!}\sum_\mathbf{d}H(U_{l+1\sim N_T}|W_{d_1\sim d_l},W_{d_{l+s_2+1}\sim d_{N_R}},W_{1\sim L}\setminus W_\mathbf{d})\notag\\
&+\frac{1}{\binom{L}{N_R}N_R!}\sum_\mathbf{d}H(V_{1\sim s_1},V_{l+1\sim l+s_2}|Y_{1\sim l},W_{d_{s_1+1}\sim d_{l}},W_{d_{l+s_2+1}\sim d_{N_R}},W_{1\sim L}\setminus W_\mathbf{d})+F\varepsilon_F \label{eqn converse 35}
\end{align}
\end{subequations}
Here, \eqref{eqn converse 31} and \eqref{eqn converse 32} follow from the definition of mutual information; \eqref{eqn converse 34} follows from \cite[Lemma 1]{lowerbound}; \eqref{eqn converse 35} follows from the Fano's inequality and the fact that conditioning reduces entropy.

The second term in \eqref{eqn converse 35} is upper bounded by 
\begin{subequations}\label{eqn 1}
\begin{align}
&\frac{1}{\binom{L}{N_R}N_R!}\sum_{\mathbf{d}}H(U_{l+1\sim N_T}|W_{d_1\sim d_l},W_{d_{l+s_2+1\sim d_{N_R}}},W_{1\sim L}\setminus W_{\mathbf{d}})\label{eqn 11}\\
\le &\frac{1}{\binom{L}{N_R}N_R!}\sum_{\mathbf{d}}\sum_{l+1\le p\le N_T}H(U_{p}|W_{d_1\sim d_l},W_{d_{l+s_2+1\sim d_{N_R}}},W_{1\sim L}\setminus W_{\mathbf{d}})\label{eqn 12}\\
= &\frac{1}{\binom{L}{N_R}N_R!}\sum_{l+1\le p\le N_T}\sum_{\mathbf{d}}\sum_{i\in\{d_{l+1},\ldots,d_{l+s_2}\}}H(U_{p,i})\label{eqn 13}\\
= &\frac{1}{\binom{L}{N_R}N_R!}\sum_{l+1\le p\le N_T}\sum_{1\le i\le L}s_2\binom{L-1}{N_R-1}(N_R-1)!H(U_{p,i})\label{eqn 14}\\
\le &\frac{1}{\binom{L}{N_R}N_R!}(N_T-l)s_2\binom{L-1}{N_R-1}(N_R-1)!M_TF\label{eqn 15}\\
=&(N_T-l)s_2\mu_TF\label{eqn 16}.
\end{align}
\end{subequations}
Here, \eqref{eqn 13} comes from the fact that only files $\{W_i: i\in\{d_{l+1},\ldots,d_{l+s_2}\}\}$ are unknown; \eqref{eqn 15} comes from the fact that each transmitter has $M_TF$ bits of caching storage. Note that \eqref{eqn 1} can be viewed equivalently as the fact that each transmitter can cache $\mu_TF$ bits of each file on average. This argument is also used for the upper bound of the third term in \eqref{eqn converse 35} below.

Further, the third term in \eqref{eqn converse 35} is upper bounded by 
\begin{subequations}\label{eqn converse 4}
\begin{align}
&\frac{1}{\binom{L}{N_R}N_R!}\sum_{\mathbf{d}}H(V_{1\sim s_1},V_{l+1\sim l+s_2}|Y_{1\sim l},W_{d_{s_1+1}\sim d_{l}},W_{d_{l+s_2+1}\sim d_{N_R}},W_{1\sim L}\setminus W_\mathbf{d})\notag\\
=&\frac{1}{\binom{L}{N_R}N_R!}\sum_{\mathbf{d}}H(V_1|Y_{1\sim l},W_{d_{s_1+1}\sim d_{l}},W_{d_{l+s_2+1}\sim d_{N_R}},W_{1\sim L}\setminus W_\mathbf{d})\notag\\
&+\frac{1}{\binom{L}{N_R}N_R!}\sum_{\mathbf{d}}H(V_{2\sim s_1},V_{l+1\sim l+s_2}|Y_{1\sim l},V_1, W_{d_{s_1+1}\sim d_{l}},W_{d_{l+s_2+1}\sim d_{N_R}},W_{1\sim L}\setminus W_\mathbf{d})\label{eqn converse 41}\\
\le &(s_1+s_2)\mu_RF\notag\\
&+\frac{1}{\binom{L}{N_R}N_R!}\sum_{\mathbf{d}}\sum_{q=2}^{s_1}H(V_{q}|Y_{1\sim l},V_{1\sim q-1},W_{d_1\sim d_{q-1}},W_{d_{s_1+1}\sim d_{l}},W_{d_{l+s_2+1}\sim d_{N_R}},W_{1\sim L}\setminus W_\mathbf{d})\notag\\
&+\frac{1}{\binom{L}{N_R}N_R!}\sum_{\mathbf{d}}H(V_{l+1\sim l+s_2}|Y_{1\sim l},V_{1\sim s_1},W_{d_{1}\sim d_{l}},W_{d_{l+s_2+1}\sim d_{N_R}},W_{1\sim L}\setminus W_\mathbf{d})+F\varepsilon_F\label{eqn converse 42}\\
\le &(s_1+s_2)\mu_RF+\sum_{q=2}^{s_1}(s_1+s_2-q+1)\mu_RF\notag\\
&+\frac{1}{\binom{L}{N_R}N_R!}\sum_{\mathbf{d}}\sum_{q=l+1}^{l+s_2}H(V_q|Y_{1\sim l},V_{1\sim s_1},W_{d_{1}\sim d_{l}},W_{d_{l+s_2+1}\sim d_{N_R}},W_{1\sim L}\setminus W_\mathbf{d})+F\varepsilon_F\label{eqn converse 43}\\
\le &(s_1+s_2)\mu_RF+\sum_{q=2}^{s_1}(s_1+s_2-q+1)\mu_RF+s_2^2\mu_RF+F\varepsilon_F\label{eqn converse 44}\\
=&\left(\frac{2s_2+s_1+1}{2}s_1+s_2^2\right)\mu_RF+F\varepsilon_F.\label{eqn converse 45}
\end{align}
\end{subequations}
Here, \eqref{eqn converse 42} comes from the Fano's inequality and the fact that receiver 1 can cache $\mu_RF$ bits of each file on average, which is similar to \eqref{eqn 1}; \eqref{eqn converse 43} comes from the fact that each receiver $q$ ($2\le q\le s_1$) can cache $\mu_RF$ bits of unknown files $W_{d_q\sim d_{s_1}},W_{d_{l+1}\sim d_{l+s_2}}$ on average, similar to \eqref{eqn 1}; \eqref{eqn converse 44} comes from the fact that each receiver $q$ ($l+1\le q\le l+s_2$) can cache $\mu_RF$ bits of unknown files $W_{l+1\sim l+s_2}$ on average, similar to \eqref{eqn 1}.

Combining \eqref{eqn converse 2}\eqref{eqn converse 35}\eqref{eqn 16}\eqref{eqn converse 45}, we have  
\begin{align}
(s_1+s_2)F\le &lT\log(2\pi e(c\cdot P+1))+(N_T-l)s_2\mu_TF+\left(\frac{2s_2+s_1+1}{2}s_1+s_2^2\right)\mu_RF\notag\\
&+F\varepsilon_F+T\epsilon_P\log P.\label{eqn converse 5}
\end{align}
Dividing $F$ on both sides of \eqref{eqn converse 5} and letting $F\to\infty$ and $P\to\infty$, we have

\begin{align}
\tau&=\lim_{P\to\infty}\lim_{F\to\infty}\frac{T\log P}{F}\notag\\
&\ge\frac{1}{l}\left\{ (s_1+s_2)-(N_T-l)s_2\mu_T-\left(\frac{2s_2+s_1+1}{2}\cdot s_1+s_2^2\right)\mu_R\right\}.\label{eqn converse 6}
\end{align}

By optimizing the bound in \eqref{eqn converse 6} over all possible choices of $s_1=0,1,\ldots,l$, $s_2=0,1,\ldots,N_R-l$, and $l=1,2,\ldots,\min\{N_T,N_R\}$, \eqref{eqn converse} is proved.

\subsection{Lower Bound \eqref{eqn taul2} without Intra-file Coding}
Next, we consider the proof of \eqref{eqn taul2} where neither intra-file coding nor inter-file coding are allowed. In this case, files can only be split and cached at receivers and transmitters without any coding. If the cache sizes at transmitters are not enough, receivers must cache some common bits of the files to guarantee the feasibility of the scheme. Thus, \eqref{eqn converse 43} and \eqref{eqn converse 44} can be further tightened. In specific, the third term in \eqref{eqn converse 35} can be upper bounded by
\begin{subequations}\label{eqn converse 8}
\begin{align}
&\frac{1}{\binom{L}{N_R}N_R!}\sum_{\mathbf{d}}H(V_{1\sim s_1},V_{l+1\sim l+s_2}|Y_{1\sim l},W_{d_{s_1+1}\sim d_{l}},W_{d_{l+s_2+1}\sim d_{N_R}},W_{1\sim L}\setminus W_\mathbf{d})\notag\\
=&\frac{1}{\binom{L}{N_R}N_R!}\sum_{\mathbf{d}}H(V_1|Y_{1\sim l},W_{d_{s_1+1}\sim d_{l}},W_{d_{l+s_2+1}\sim d_{N_R}},W_{1\sim L}\setminus W_\mathbf{d})\notag\\
&+\frac{1}{\binom{L}{N_R}N_R!}\sum_{\mathbf{d}}H(V_{2\sim s_1},V_{l+1\sim l+s_2}|Y_{1\sim l},V_1,W_{d_{s_1+1}\sim d_{l}},W_{d_{l+s_2+1}\sim d_{N_R}},W_{1\sim L}\setminus W_\mathbf{d})\label{eqn converse 81}\\
\le &(s_1+s_2)\mu_RF\notag\\
&+\frac{1}{\binom{L}{N_R}N_R!}\sum_{\mathbf{d}}\sum_{q=2}^{s_1}H(V_{q}|Y_{1\sim l},V_{1\sim q-1},W_{d_1\sim d_{q-1}},W_{d_{s_1+1}\sim d_{l}},W_{d_{l+s_2+1}\sim d_{N_R}},W_{1\sim L}\setminus W_\mathbf{d})\notag\\
&+\frac{1}{\binom{L}{N_R}N_R!}\sum_{\mathbf{d}}H(V_{l+1\sim l+s_2}|Y_{1\sim l},V_{1\sim s_1},W_{d_1\sim d_l},W_{d_{l+s_2+1}\sim d_{N_R}},W_{1\sim L}\setminus W_\mathbf{d})+F\varepsilon_F\label{eqn converse 82}\\
\le &(s_1+s_2)\mu_RF+\sum_{q=2}^{s_1}(s_1+s_2-q+1)\left(\mu_R-(1-N_T\mu_T)^+\right)F\notag\\
&+\frac{1}{\binom{L}{N_R}N_R!}\sum_{\mathbf{d}}\sum_{q=l+1}^{l+s_2}H(V_q|Y_{1\sim l},V_{1\sim s_1},W_{d_1\sim d_l},W_{d_{l+s_2+1}\sim d_{N_R}},W_{1\sim L}\setminus W_\mathbf{d})+F\varepsilon_F\label{eqn converse 84}\\
\le&\sum_{q=2}^{s_1}(s_1+s_2-q+1)\left(\mu_R-(1-N_T\mu_T)^+\right)F+s_2^2\left(\mu_R-(1-N_T\mu_T)^+\right)F\notag\\
&+(s_1+s_2)\mu_RF+F\varepsilon_F\label{eqn converse 85}\\
=&\left(\frac{2s_2+s_1+1}{2}s_1+s_2^2\right)\mu_RF-\left(\frac{2s_2+s_1}{2}(s_1-1)+s_2^2\right)(1-N_T\mu_T)^+F+F\varepsilon_F.\label{eqn converse 86}
\end{align}
\end{subequations}
Here, \eqref{eqn converse 82} follows from the Fano's inequality and the fact that given receiver 1 has $M_RF$ cache storage, it can cache $\mu_RF$ bits of each file on average, which is similar to \eqref{eqn 1}; \eqref{eqn converse 84} follows from the fact that at least $(1-N_T\mu_T)^+F$ cached bits of each file on average are common for all receivers to guarantee the feasibility of the scheme if $N_TM_TF< LF$. Thus, each receiver $q$ ($2\le q\le s_1$) can cache $(\mu_R-(1-N_T\mu_T)^+)F$ bits of unknown files $W_{d_q\sim d_{s_1}},W_{d_{l+1}\sim d_{l+s_2}}$ on average, similar to \eqref{eqn 1}; \eqref{eqn converse 85} follows from the fact that each receiver $q$ ($l+1\le q\le l+s_2$) can cache $(\mu_R-(1-N_T\mu_T)^+)F$ bits of unknown files $W_{d_{l+1}\sim d_{l+s_2}}$ on average, similar to \eqref{eqn 1}.

Combining \eqref{eqn converse 2}\eqref{eqn converse 35}\eqref{eqn 16}\eqref{eqn converse 86}, we have 
\begin{align}
(s_1+s_2)F\le &lT\log(2\pi e(c\cdot P+1))+(N_T-l)s_2\mu_TF+\left(\frac{2s_2+s_1+1}{2}s_1+s_2^2\right)\mu_RF\notag\\
&-\left(\frac{2s_2+s_1}{2}(s_1-1)+s_2^2\right)(1-N_T\mu_T)^+F+F\varepsilon_F+T\epsilon_P\log P.\label{eqn converse 9}
\end{align}
Dividing $F$ on both sides of \eqref{eqn converse 9} and letting $F\to\infty$ and $P\to\infty$, we have
\begin{align}
\tau&=\lim_{P\to\infty}\lim_{F\to\infty}\frac{T\log P}{F}\notag\\
&\ge\frac{1}{l}\left\{ (s_1+s_2)-(N_T-l)s_2\mu_T-\left(\frac{2s_2+s_1+1}{2}\cdot s_1+s_2^2\right)\mu_R\right.\notag\\
&\qquad\quad\left.+\left(\frac{2s_2+s_1}{2}(s_1-1)+s_2^2\right)(1-N_T\mu_T)^+\right\}.\label{eqn converse 10}
\end{align}
By optimizing the bound in \eqref{eqn converse 10} over all possible choices of $s_1=0,1,\ldots,l$, $s_2=0,1,\ldots,N_R-l$, and $l=1,2,\cdots,\min\{N_T,N_R\}$, \eqref{eqn taul2} is proved.

\section{Conclusions}\label{section conclusion}
In this paper, we have characterized the normalized delivery time for a general $N_T \times N_R$ interference network where both the transmitter and receiver sides are equipped with caches. We have obtained both the achievable upper bound and the theoretical lower bound of the minimum NDT for any  $N_T\ge2$, any $N_R\ge2$, and any normalized cache size tuple $(\mu_R, \mu_T)$ in the feasible region. The achievable bound is expressed as the optimal solution of a linear programming problem which can be solved efficiently.  The closed-form expressions for the $2\times2$ and $3\times3$ networks show that it is a piece-wise linearly decreasing function of the normalized cache sizes. The achievable NDT is exactly optimal in a number of special cases and is within a bounded multiplicative gap to the lower bound in general cases. In specific, the gap is a constant in most cases, and is bounded by $\frac{N_T+N_R-1}{N_T}$ in the case when  $\mu_T<1/N_T$ (the accumulated cache size at all transmitters is not enough to cache the entire file library) and $N_T<N_R$.   The proposed cache placement strategy involves generic file splitting with adjustable ratios. The proposed delivery strategy transforms the interference network into a new class of cooperative X-multicast channels. We derived the achievable DoF of this new channel via interference neutralization and interference alignment techniques. Our analysis shows that the proposed caching method can leverage receiver local caching gain, coded multicasting gain, and transmitter cooperation gain opportunistically.  Analysis also shows that the optimal file splitting ratios are not unique.

\newpage

\section*{Appendix A: Proof of Lemma \ref{lemma dof}}
Consider the $\binom{N_T}{t}\times\binom{N_R}{r+1}$ cooperative X-multicast channel defined in Section \ref{section delivery}-A. There are $\binom{N_R}{r+1}\binom{N_T}{t}$ messages in total.  We denote message $W_{\mathcal{R},\mathcal{T}}$ as the message desired by receiver multicast group $\mathcal{R}$ with $|\mathcal{R}|=r+1$ and cached at transmitter cooperation group $\mathcal{T}$ with $|\mathcal{T}|=t$. For example, $W_{{[r+1]},{[t]}}$ is desired by receivers $\{1,2,\ldots,r+1\}$, and available at transmitters $\{1,2,\ldots,t\}$. We divide the proof of Lemma \ref{lemma dof} into three parts as follows, according to the relationship between $N_R$ and $r+t$.

\subsection{$r+t\ge N_R$}
In this case, we divide the total $\binom{N_R}{r+1}\binom{N_T}{t}$ messages into $\binom{N_T}{t}$ groups, such that the messages in the same group are available at the same transmitter cooperation group. Each receiver desires $\binom{N_R-1}{r}$ messages out of the $\binom{N_R}{r+1}$ messages in each group. We then deliver each message group sequentially in a time division manner. Here, we take the group associated with transmitter cooperation set $\mathcal{T}=\{1,2,\ldots,t\}$ as an example to illustrate the achievable transmission scheme. All the other transmitter sets can use the same method.

Denote $x_{\mathcal{R},{\mathcal{T}}}$ as the transmitted symbol encoded from message $W_{{\mathcal{R}},{\mathcal{T}}}$. Each $x_{\mathcal{R},{\mathcal{T}}}$ is wanted by the receiver set $\mathcal{R}$ and unwanted by the receiver set $\bar{\mathcal{R}}=[N_R]\setminus\mathcal{R}$. Note that each transmitted symbol can be cancelled at $\min\{N_R-r-1,t-1\}$ undesired receivers by interference neutralization among the $t$ cooperating transmitters in $\mathcal{T}$. Since $r+t\ge N_R$, we only let $N_R-r$ ($\le t$) transmitters in $\mathcal{T}$ cooperatively transmit the symbol $x_{\mathcal{R},{\mathcal{T}}}$, and deactivate the rest $t+r-N_R$ transmitters. In this case, each symbol can still be neutralized at all the $N_R-r-1$ undesired receivers. Without loss of generality, transmitters $\{1,2,\ldots,N_R-r\}$ are selected for cooperative transmission.

We use a $\rho\triangleq\binom{N_R-1}{r}$-symbol extension to transmit the $\binom{N_R}{r+1}$ messages in this group. Note that $\rho$ is also the total number of messages desired by each receiver. In each time slot $u\in[\rho]$, the received signal at an arbitrary receiver $q\in[N_R]$, denoted as $y_q(u)$, is given by (neglecting the noise)
\begin{align}
y_q(u)=&\sum_{p=1}^{N_R-r}h_{qp}(u)\underbrace{\sum_{\mathcal{R}:|\mathcal{R}|=r+1}v_{{\mathcal{R}},{\mathcal{T}},p}(u)x_{{\mathcal{R}},{\mathcal{T}}}}_{\binom{N_R}{r+1} \textrm{ terms}}\notag\\
=&\underbrace{\sum_{\mathcal{R}:|\mathcal{R}|=r+1,\mathcal{R}\ni q}\left[\sum_{p=1}^{N_R-r}h_{qp}(u)v_{{\mathcal{R}},{\mathcal{T}},p}(u)\right]x_{{\mathcal{R}},{\mathcal{T}}}}_{\textrm{wanted}}+\underbrace{\sum_{\mathcal{R}:|\mathcal{R}|=r+1,\mathcal{R}\not\ni q}\left[\sum_{p=1}^{N_R-r}h_{qp}(u)v_{{\mathcal{R}},{\mathcal{T}},p}(u)\right]x_{{\mathcal{R}},{\mathcal{T}}}}_{\textrm{unwanted}}\label{eqn proof1 receive signal}
\end{align}
where $h_{qp}(u)$ is the channel realization, and $v_{{\mathcal{R}},{\mathcal{T}},p}(u)$ is the precoder of symbol $x_{{\mathcal{R}},{\mathcal{T}}}$ at transmitter $p$. For each undesired receiver $q\in\bar{\mathcal{R}}$ of symbol $x_{{\mathcal{R}},{\mathcal{T}}}$, to apply interference neutralization, we need
\begin{align}
\sum_{p=1}^{N_R-r}h_{qp}(u)v_{{\mathcal{R}},{\mathcal{T}},p}(u)=0,\forall q\in\bar{\mathcal{R}}, \forall u\in[\rho].\label{eqn proof1 neutralization}
\end{align}

We now design the precoders $\{v_{{\mathcal{R}},{\mathcal{T}},p}(u)\}$ to meet \eqref{eqn proof1 neutralization}. Consider the symbol $x_{{\mathcal{R}},{\mathcal{T}}}$ desired by an arbitrary receiver multicast group $\mathcal{R}=\{R_1,R_2,\ldots,R_{r+1}\}$. The undesired receiver set of $x_{{\mathcal{R}},{\mathcal{T}}}$ is $\bar{\mathcal{R}}=\{R_{r+2},\ldots,R_{N_R}\}$. Here, $(R_1,R_2,\ldots,R_{N_R})$ represents an arbitrary permutation of receiver index $(1,2,\ldots,N_R)$. Consider the following $(N_R-r)\times(N_R-r)$ matrix:
\begin{align}
\left(
\begin{matrix}
h_{R_{r+2},1}&h_{R_{r+2},2}&\cdots&h_{R_{r+2},N_R-r}\\
h_{R_{r+3},1}&h_{R_{r+3},2}&\cdots&h_{R_{r+3},N_R-r}\\
\cdots&\cdots&\cdots&\cdots\\
h_{R_{N_R},1}&h_{R_{N_R},2}&\cdots&h_{R_{N_R},N_R-r}\\
a_{1}&a_{2}&\cdots&a_{N_R-r}
\end{matrix}
\right)\triangleq \mathbf{H}_{\bar{\mathcal{R}}},\label{eqn neutralization matrix}
\end{align}
for any $\{a_{1},a_2,\ldots,a_{N_R-r}\}$. Define $c_p$ as the cofactor of $a_p$ such that the determinant of $\mathbf{H}_{\bar{\mathcal{R}}}$ can be expressed as
\begin{align}
\sum_{p=1}^{N_R-r}a_{p}c_{p}=\det(\mathbf{H}_{\bar{\mathcal{R}}}).\label{eqn neutralization determinant}
\end{align}
Then, we design $v_{{\mathcal{R}},{\mathcal{T}},p}(u)$ as $v_{{\mathcal{R}},{\mathcal{T}},p}(u)=c_p(u)$ by taking channel realization $\{h_{qp}(u)\}$ into \eqref{eqn neutralization matrix} and \eqref{eqn neutralization determinant}.

By such construction, the condition in \eqref{eqn proof1 neutralization} is satisfied. For example, at receiver $R_{r+2}$, we have
\begin{align}
&\sum_{p=1}^{N_R-r}h_{R_{r+2},p}(u)v_{{\mathcal{R}},{\mathcal{T}},p}(u)\notag\\
=&\sum_{p=1}^{N_R-r}h_{R_{r+2},p}(u)c_{p}(u)\notag\\
=&
\begin{vmatrix}
h_{R_{r+2},1}(u)&h_{R_{r+2},2}(u)&\cdots&h_{R_{r+2},N_R-r}(u)\\
h_{R_{r+3},1}(u)&h_{R_{r+3},2}(u)&\cdots&h_{R_{r+3},N_R-r}(u)\\
\cdots&\cdots&\cdots&\cdots\\
h_{R_{N_R},1}(u)&h_{R_{N_R},2}(u)&\cdots&h_{R_{N_R},N_R-r}(u)\\
h_{R_{r+2},1}(u)&h_{R_{r+2},2}(u)&\cdots&h_{R_{r+2},N_R-r}(u)
\end{vmatrix}
=0.
\end{align}

After interference neutralization, the received signal in \eqref{eqn proof1 receive signal} can be rewritten as
\begin{align}
y_q(u)&=\sum_{\mathcal{R}:|\mathcal{R}|=r+1,\mathcal{R}\ni q}\tilde{h}_{q,\mathcal{T}}^{\bar{\mathcal{R}}}(u)x_{{\mathcal{R}},{\mathcal{T}}},\label{eqn proof1 receive signal after neutralization}
\end{align}
where
\begin{align}
\tilde{h}_{q,\mathcal{T}}^{\bar{\mathcal{R}}}(u)\triangleq\sum_{p=1}^{N_R-r}h_{q,p}(u)v_{{\mathcal{R}},{\mathcal{T}},p}(u)
=\begin{vmatrix}
h_{R_{r+2},1}(u)&h_{R_{r+2},2}(u)&\cdots&h_{R_{r+2},N_R-r}(u)\\
h_{R_{r+3},1}(u)&h_{R_{r+3},2}(u)&\cdots&h_{R_{r+3},N_R-r}(u)\\
\cdots&\cdots&\cdots&\cdots\\
h_{R_{N_R},1}(u)&h_{R_{N_R},2}(u)&\cdots&h_{R_{N_R},N_R-r}(u)\\
h_{q,1}(u)&h_{q,2}(u)&\cdots&h_{q,N_R-r}(u)
\end{vmatrix}.\label{eqn equivalent channel1}
\end{align}

To successfully decode the $\rho$ desired messages of receiver $q$,  we need to assure that the following $\rho\times \rho$ received signal matrix is full-rank with probability 1:
\begin{align}
\begin{pmatrix}
\tilde{h}_{q,\mathcal{T}}^{\bar{\mathcal{R}}_1}(1)&\tilde{h}_{q,\mathcal{T}}^{\bar{\mathcal{R}}_2}(1)&\cdots&\tilde{h}_{q,\mathcal{T}}^{\bar{\mathcal{R}}_{\rho}}(1)\\
\tilde{h}_{q,\mathcal{T}}^{\bar{\mathcal{R}}_1}(2)&\tilde{h}_{q,\mathcal{T}}^{\bar{\mathcal{R}}_2}(2)&\cdots&\tilde{h}_{q,\mathcal{T}}^{\bar{\mathcal{R}}_{\rho}}(2)\\
\cdots&\cdots&\cdots&\cdots\\
\tilde{h}_{q,\mathcal{T}}^{\bar{\mathcal{R}}_1}(\rho)&\tilde{h}_{q,\mathcal{T}}^{\bar{\mathcal{R}}_2}(\rho)&\cdots&\tilde{h}_{q,\mathcal{T}}^{\bar{\mathcal{R}}_{\rho}}(\rho)
\end{pmatrix},\label{eqn matrix1}
\end{align}
where $\{\bar{\mathcal{R}}_1,\bar{\mathcal{R}}_2,\ldots,\bar{\mathcal{R}}_{\rho}\}$ denotes the undesired receiver sets of the $\binom{N_R-1}{r}$ messages intended for receiver $q$. Given that the construction method of $\{v_{{\mathcal{R}},{\mathcal{T}},p}(u)\}$ and the formation of $\{\tilde{h}_{q,\mathcal{T}}^{\bar{\mathcal{R}}}(u)\}$ are the same at each time slot $u$ as in \eqref{eqn neutralization matrix}, \eqref{eqn neutralization determinant} and \eqref{eqn equivalent channel1}, using \cite[Lemma 3]{CoMP}, we only need to prove the linear independence of polynomials $\{\tilde{h}_{q,\mathcal{T}}^{\bar{\mathcal{R}}_1},\tilde{h}_{q,\mathcal{T}}^{\bar{\mathcal{R}}_2},\ldots,\tilde{h}_{q,\mathcal{T}}^{\bar{\mathcal{R}}_{\rho}}\}$ as functions of $\{h_{qp}\}$. Since $\tilde{h}_{q,\mathcal{T}}^{\bar{\mathcal{R}}}$ is the determinant of the matrix \eqref{eqn equivalent channel1} and each $\tilde{h}_{q,\mathcal{T}}^{\bar{\mathcal{R}}}$ has a unique undesired receiver set $\bar{\mathcal{R}}$, it is easy to see that polynomials $\{\tilde{h}_{q,\mathcal{T}}^{\bar{\mathcal{R}}_1},\tilde{h}_{q,\mathcal{T}}^{\bar{\mathcal{R}}_2},\ldots,\tilde{h}_{q,\mathcal{T}}^{\bar{\mathcal{R}}_{\rho}}\}$ are linearly independent. Using \cite[Lemma 3]{CoMP}, we assure that the received signal matrix \eqref{eqn matrix1} is full-rank with probability 1. Therefore, receiver $q$ can successfully decode its $\binom{N_R-1}{r}$ desired messages in $\binom{N_R-1}{r}$ time slots. Similar arguments can be applied to other receivers. Therefore, a per-user DoF of 1 is achieved.

\subsection{$r+t=N_R-1$}
Next, let us consider the case when $r+t=N_R-1$. Since each message can only be canceled at $t-1$ undesired receivers by interference neutralization while there are $N_R-r-1=t$ undesired receivers in total, each message will still cause interference to one undesired receiver. In this case, asymptotic interference alignment is further applied. To be specific, let each message $W_{{\mathcal{R}},{\mathcal{T}}}$ be encoded into a $t N^{\binom{N_R-1}{r+1}\binom{N_T}{t}}\times 1$ symbol vector $\mathbf{x}_{{\mathcal{R}},{\mathcal{T}}}=((\mathbf{x}_{{\mathcal{R}},{\mathcal{T}}}^1)^T,(\mathbf{x}_{{\mathcal{R}},{\mathcal{T}}}^2)^T,\ldots, (\mathbf{x}_{{\mathcal{R}},{\mathcal{T}}}^t)^T)^T$, where $N\in \mathds{Z}^+$, and $\mathbf{x}_{{\mathcal{R}},{\mathcal{T}}}^i$ ($i\in[t]$) is an $N^{\binom{N_R-1}{r+1}\binom{N_T}{t}}\times 1$ vector. We use an $S\triangleq S_0+(N+1)^{\binom{N_R-1}{r+1}\binom{N_T}{t}}$-symbol extension, where $S_0\triangleq\binom{N_R-1}{r}\binom{N_T}{t}t N^{\binom{N_R-1}{r+1}\binom{N_T}{t}}$. Note that $S_0$ is also the total number of symbols desired by each receiver. Unlike the previous method in Case A which used message grouping, here, we transmit all the $\binom{N_R}{r+1}\binom{N_T}{t}$ messages together. In each time slot $u\in[S]$, the received signal at an arbitrary receiver $q\in[N_R]$, denoted as $y_q(u)$, is given by (neglecting the noise)
\begin{align}
y_q(u)&=\sum_{\mathcal{R}:|\mathcal{R}|=r+1}\sum_{\mathcal{T}:|\mathcal{T}|=t}\sum_{i=1}^t\left[\sum_{p\in \mathcal{T}}h_{qp}(u)\left(\mathbf{v}_{{\mathcal{R}},{\mathcal{T}},p}^i(u)\right)^T\right]\mathbf{x}_{{\mathcal{R}},{\mathcal{T}}}^i,\label{eqn proof22 receive signal}
\end{align}
where $h_{qp}(u)$ is the channel realization, and $\mathbf{v}_{{\mathcal{R}},{\mathcal{T}},p}^i(u)$ is the $N^{\binom{N_R-1}{r+1}\binom{N_T}{t}}\times 1$ precoding vector of symbol vector $\mathbf{x}_{{\mathcal{R}},{\mathcal{T}}}^i$ at transmitter $p$.

We first elaborate the interference neutralization strategy. Consider an arbitrary symbol vector $\mathbf{x}_{{\mathcal{R}},{\mathcal{T}}}^i$ desired by receiver multicast group $\mathcal{R}=\{R_1,R_2,\ldots,R_{r+1}\}$, transmitted by transmitter cooperation group $\mathcal{T}=\{T_{1},T_{2},\ldots,T_{t}\}$, and whose undesired receiver set is $\bar{\mathcal{R}}=[N_R]\setminus\mathcal{R}=\{\bar{R}_1,\bar{R}_2,\ldots,\bar{R}_t\}$. We assume that $\mathbf{x}_{{\mathcal{R}},{\mathcal{T}}}^i$ will be neutralized at receiver set $\bar{\mathcal{R}}_i\triangleq\bar{\mathcal{R}}\setminus\{\bar{R}_i\}$. Then, the precoder must satisfy
\begin{align}
\sum_{p\in \mathcal{T}}h_{qp}(u)v_{{\mathcal{R}},{\mathcal{T}},p,n}^i(u)=0, \forall q\in \bar{\mathcal{R}}_i, \forall n\in[N^{\binom{N_R-1}{r+1}\binom{N_T}{t}}], \forall u\in[S]\label{eqn proof2 neutralization}
\end{align}
where $v_{{\mathcal{R}},{\mathcal{T}},p,n}^i(u)$ is the $n$-th element of $\mathbf{v}_{{\mathcal{R}},{\mathcal{T}},p}^i(u)$. Consider the following $t\times t$ matrix:
\begin{align}
\left(
\begin{matrix}
h_{\bar{R}_1,T_1}&h_{\bar{R}_1,T_2}&\cdots&h_{\bar{R}_1,T_t}\\
h_{\bar{R}_2,T_1}&h_{\bar{R}_2,T_2}&\cdots&h_{\bar{R}_2,T_t}\\
\cdots&\cdots&\cdots&\cdots\\
h_{\bar{R}_{i-1},T_1}&h_{\bar{R}_{i-1},T_2}&\cdots&h_{\bar{R}_{i-1},T_t}\\
h_{\bar{R}_{i+1},T_1}&h_{\bar{R}_{i+1},T_2}&\cdots&h_{\bar{R}_{i+1},T_t}\\
\cdots&\cdots&\cdots&\cdots\\
h_{\bar{R}_{t},T_1}&h_{\bar{R}_{t},T_2}&\cdots&h_{\bar{R}_{t},T_t}\\
a_{1}&a_{2}&\cdots&a_{t}
\end{matrix}
\right)\triangleq \mathbf{H}_{\bar{\mathcal{R}}_i,\mathcal{T}},\label{eqn proof2 neutralization matrix}
\end{align}
for any $\{a_{1},a_2,\ldots,a_{t}\}$. Define $c_p$ as the cofactor of $a_p$ such that
\begin{align}
\sum_{p=1}^{t}a_{p}c_{p}=\det(\mathbf{H}_{\bar{\mathcal{R}}_i,\mathcal{T}}).\label{eqn proof2 neutralization determinant}
\end{align}
We then design $v_{{\mathcal{R}},{\mathcal{T}},p,n}^i(u)$ as
\begin{align}
v_{{\mathcal{R}},{\mathcal{T}},p,n}^i(u)=\alpha_{{\mathcal{R}},{\mathcal{T}}}^{\bar{\mathcal{R}}_i}(u)c_p(u) z_{{\mathcal{R}},{\mathcal{T}},n}^{\bar{\mathcal{R}}_i}(u),\label{eqn proof2 v}
\end{align}
where $\alpha_{{\mathcal{R}},{\mathcal{T}}}^{\bar{\mathcal{R}}_i}(u)$ is chosen i.i.d. from a continuous distribution for all $\{\mathcal{R,T},\bar{\mathcal{R}}_i,u\}$, $c_p(u)$ is the cofactor $c_p$ by taking channel realization $\{h_{qp}(u)\}$ into  \eqref{eqn proof2 neutralization matrix} and \eqref{eqn proof2 neutralization determinant}, and $z_{{\mathcal{R}},{\mathcal{T}},n}^{\bar{\mathcal{R}}_i}(u)$ will be determined later. By such construction, the condition in \eqref{eqn proof2 neutralization} is satisfied. For example, at receiver $\bar{R}_1$, we have
\begin{align}
&\sum_{p\in \mathcal{T}}h_{\bar{R}_1,p}(u)\alpha_{{\mathcal{R}},{\mathcal{T}}}^{\bar{\mathcal{R}}_i}(u)c_p(u)z_{{\mathcal{R}},{\mathcal{T}},n}^{\bar{\mathcal{R}}_i}(u)\notag\\
=\,&\alpha_{{\mathcal{R}},{\mathcal{T}}}^{\bar{\mathcal{R}}_i}(u)z_{{\mathcal{R}},{\mathcal{T}},n}^{\bar{\mathcal{R}}_i}(u)\sum_{p\in \mathcal{T}}h_{\bar{R}_1,p}(u)c_p(u)\notag\\
=\,&\alpha_{{\mathcal{R}},{\mathcal{T}}}^{\bar{\mathcal{R}}_i}(u)z_{{\mathcal{R}},{\mathcal{T}},n}^{\bar{\mathcal{R}}_i}(u)\cdot
\begin{vmatrix}
h_{\bar{R}_1,T_1}(u)&h_{\bar{R}_1,T_2}(u)&\cdots&h_{\bar{R}_1,T_t}(u)\\
h_{\bar{R}_2,T_1}(u)&h_{\bar{R}_2,T_2}(u)&\cdots&h_{\bar{R}_2,T_t}(u)\\
\cdots&\cdots&\cdots&\cdots\\
h_{\bar{R}_{i-1},T_1}(u)&h_{\bar{R}_{i-1},T_2}(u)&\cdots&h_{\bar{R}_{i-1},T_t}(u)\\
h_{\bar{R}_{i+1},T_1}(u)&h_{\bar{R}_{i+1},T_2}(u)&\cdots&h_{\bar{R}_{i+1},T_t}(u)\\
\cdots&\cdots&\cdots&\cdots\\
h_{\bar{R}_{t},T_1}(u)&h_{\bar{R}_{t},T_2}(u)&\cdots&h_{\bar{R}_{t},T_t}(u)\\
h_{\bar{R}_1,T_1}(u)&h_{\bar{R}_1,T_2}(u)&\cdots&h_{\bar{R}_1,T_t}(u)
\end{vmatrix}\notag\\
=\,&0
\end{align}

By the above construction of precoders, it can be seen that symbol vectors $\mathbf{x}_{{\mathcal{R}},{\mathcal{T}}}^i$ unwanted by receiver $q\in\bar{\mathcal{R}}_i$ are all neutralized. Then, the received signal in \eqref{eqn proof22 receive signal} at an arbitrary receiver $q\in[N_R]$ can be rewritten as
\begin{align}
y_q(u)=&\sum_{\mathcal{R}:|\mathcal{R}|=r+1,\mathcal{R}\ni q}\sum_{\mathcal{T}:|\mathcal{T}|=t}\sum_{i=1}^t\left[\sum_{p\in \mathcal{T}}h_{qp}(u)\left(\mathbf{v}_{{\mathcal{R}},{\mathcal{T}},p}^i(u)\right)^T\right]\mathbf{x}_{{\mathcal{R}},{\mathcal{T}}}^i\notag\\
&+\sum_{\substack{\bar{\mathcal{R}}_i,\mathcal{R}:\\|\mathcal{R}|=r+1,\mathcal{R}\cup\bar{\mathcal{R}}_i\not\ni q}}\sum_{\mathcal{T}:|\mathcal{T}|=t}\left[\sum_{p\in \mathcal{T}}h_{qp}(u)\left(\mathbf{v}_{{\mathcal{R}},{\mathcal{T}},p}^i(u)\right)^T\right]\mathbf{x}_{{\mathcal{R}},{\mathcal{T}}}^i,\label{eqn proof2 receive signal after neutralization}
\end{align}
where the first term is the desired messages of receiver $q$ and the second term is the residual interferences.

Now, we aim to apply asymptotic interference alignment to align the interference term in \eqref{eqn proof2 receive signal after neutralization} at the same sub-space. Consider the following monomial set:
\begin{align}
\mathcal{M}_q[N]=\left\{\prod_{\substack{\mathcal{R},\bar{\mathcal{R}}_i,\mathcal{T}:\\\mathcal{R}\cup\bar{\mathcal{R}}_i\not\ni q}}\left[\alpha_{{\mathcal{R}},{\mathcal{T}}}^{\bar{\mathcal{R}}_i}\tilde{h}_{q,\mathcal{T}}^{\bar{\mathcal{R}}_i}\right]^{s_{{\mathcal{R}},{\mathcal{T}}}^{i}}: 1\le s_{{\mathcal{R}},{\mathcal{T}}}^{i}\le N\right\},\label{eqn monomial set 1}
\end{align}
where $\tilde{h}_{q,\mathcal{T}}^{\bar{\mathcal{R}}_i}$ is defined as
\begin{align}
\tilde{h}_{q,\mathcal{T}}^{\bar{\mathcal{R}}_i}\triangleq\sum_{p=T_1}^{T_t}h_{q,p}c_{p}=
\begin{vmatrix}
h_{\bar{R}_1,T_1}&h_{\bar{R}_1,T_2}&\cdots&h_{\bar{R}_1,T_t}\\
h_{\bar{R}_2,T_1}&h_{\bar{R}_2,T_2}&\cdots&h_{\bar{R}_2,T_t}\\
\cdots&\cdots&\cdots&\cdots\\
h_{\bar{R}_{i-1},T_1}&h_{\bar{R}_{i-1},T_2}&\cdots&h_{\bar{R}_{i-1},T_t}\\
h_{\bar{R}_{i+1},T_1}&h_{\bar{R}_{i+1},T_2}&\cdots&h_{\bar{R}_{i+1},T_t}\\
\cdots&\cdots&\cdots&\cdots\\
h_{\bar{R}_{t},T_1}&h_{\bar{R}_{t},T_2}&\cdots&h_{\bar{R}_{t},T_t}\\
h_{q,T_1}&h_{q,T_2}&\cdots&h_{q,T_t}
\end{vmatrix}.\label{eqn proof2 equivalent channel}
\end{align}
The cardinality of $\mathcal{M}_q[N]$ is $N^{\binom{N_R-1}{r+1}\binom{N_T}{t}}$. For each element $v_{{\mathcal{R}},{\mathcal{T}},p,n}^i(u)$ in $\mathbf{v}_{{\mathcal{R}},{\mathcal{T}},p}^i(u)$ satisfying $\mathcal{R}\cup\bar{\mathcal{R}}_i\not\ni q$, the element $z_{{\mathcal{R}},{\mathcal{T}},n}^{\bar{\mathcal{R}}_i}(u)$ in \eqref{eqn proof2 v} is given by a unique monomial $m_{{\mathcal{R}},{\mathcal{T}},n}^{\bar{\mathcal{R}}_i}(u)$ in $\mathcal{M}_q[N]$ by taking $\{h_{qp}(u)\}$ and $\{\alpha_{{\mathcal{R}},{\mathcal{T}}}^{\bar{\mathcal{R}}_i}(u)\}$ into \eqref{eqn monomial set 1}. Then it can be seen that for elements $v_{{\mathcal{R}},{\mathcal{T}},p,n}^i(u)$ in $\mathbf{v}_{{\mathcal{R}},{\mathcal{T}},p}^i(u)$ satisfying $\mathcal{R}\cup\bar{\mathcal{R}}_i\not\ni q$, the summation $\sum_{p\in \mathcal{T}}h_{qp}(u)v_{{\mathcal{R}},{\mathcal{T}},p,n}^i(u)$ is $\alpha_{{\mathcal{R}},{\mathcal{T}}}^{\bar{\mathcal{R}}_i}(u)\tilde{h}_{q,\mathcal{T}}^{\bar{\mathcal{R}}_i}(u)m_{{\mathcal{R}},{\mathcal{T}},n}^{\bar{\mathcal{R}}_i}(u)$ and satisfies
\begin{align}
\alpha_{{\mathcal{R}},{\mathcal{T}}}^{\bar{\mathcal{R}}_i}(u)\tilde{h}_{q,\mathcal{T}}^{\bar{\mathcal{R}}_i}(u)m_{{\mathcal{R}},{\mathcal{T}},n}^{\bar{\mathcal{R}}_i}(u)\in\mathcal{M}_q[N+1](u),\notag
\end{align}
Therefore, the interferences at receiver $q$ are aligned together.

The received signal in \eqref{eqn proof2 receive signal after neutralization} can be rewritten as
\begin{align}
y_q(u)=&\sum_{\mathcal{R}:|\mathcal{R}|=r+1,\mathcal{R}\ni q}\sum_{\mathcal{T}:|\mathcal{T}|=t}\sum_{i=1}^t\left[\sum_{p\in \mathcal{T}}h_{qp}(u)\left(\mathbf{v}_{{\mathcal{R}},{\mathcal{T}},p}^i(u)\right)^T\right]\mathbf{x}_{{\mathcal{R}},{\mathcal{T}}}^i+\sum_{m(u)\in \mathcal{M}_q[N+1](u)}m(u)x_{m(u)},\label{eqn proof2 receive signal after alignment}
\end{align}
where $x_{m(u)}$ is the sum of interference symbols whose received factor is $m(u)$ at receiver $q$. To successfully  decode the $\binom{N_R-1}{r}\binom{N_T}{t}$ desired messages of receiver $q$, we need to assure that the $S\times S$ received signal matrix whose column vectors are
\begin{align}
&\left\{\left(\alpha_{{\mathcal{R}},{\mathcal{T}}}^{\bar{\mathcal{R}}_i}(u)\tilde{h}_{q,\mathcal{T}}^{\bar{\mathcal{R}}_i}(u)m_{{\mathcal{R}},{\mathcal{T}},n}^{\bar{\mathcal{R}}_i}(u)\right)_{u=1}^S:|\mathcal{R}|=r+1,\mathcal{R}\ni q,|\mathcal{T}|=t,i\in[t],n\in[N^{\binom{N_R-1}{r+1}\binom{N_T}{t}}]\right\}\notag\\
&\cup \left\{\left(m(u)\right)_{u=1}^S:m(u)\in\mathcal{M}_q[N+1](u)\right\}
\end{align}
is full-rank with probability 1.

Since the construction method of $\{\mathbf{v}_{{\mathcal{R}},{\mathcal{T}},p}^i(u)\}$ and the formation of $\{\alpha_{{\mathcal{R}},{\mathcal{T}}}^{\bar{\mathcal{R}}_i}(u)\tilde{h}_{q,\mathcal{T}}^{\bar{\mathcal{R}}_i}(u)m_{{\mathcal{R}},{\mathcal{T}},n}^{\bar{\mathcal{R}}_i}(u)\}$ and $\{m(u)\}$ are the same at each time slot $u$, based on \cite[Lemma 3]{CoMP}, we only need to prove the linear independence of these polynomials as functions of $\{h_{qp}\}$ and $\{\alpha_{{\mathcal{R}},{\mathcal{T}}}^{\bar{\mathcal{R}}_i}\}$, which is given by
\begin{align}
&\left\{\alpha_{{\mathcal{R}},{\mathcal{T}}}^{\bar{\mathcal{R}}_i}\tilde{h}_{q,\mathcal{T}}^{\bar{\mathcal{R}}_i}m_{{\mathcal{R}},{\mathcal{T}},n}^{\bar{\mathcal{R}}_i}:\mathcal{R}\ni q,|\mathcal{R}|=r+1,|\mathcal{T}|=t,i\in[t],n\in[N^{\binom{N_R-1}{r+1}\binom{N_T}{t}}\right\}\notag\\
&\cup\{m:m\in\mathcal{M}_q[N+1]\}.\notag
\end{align}
It can be seen that $\{\alpha_{{\mathcal{R}},{\mathcal{T}}}^{\bar{\mathcal{R}}_i}: \mathcal{R}\cup\bar{\mathcal{R}}_i\not \ni q',\forall\mathcal{T}\}$ only exist in the polynomials whose transmitted symbols are the interference of receiver $q'$. Thus, polynomials with different $q'$ are linearly independent. Next, let us consider the polynomials of desired symbols of receiver $q$ corresponding to the same $q'$. The polynomials are given in the following set:
\begin{align}
&\left\{\alpha_{{\mathcal{R}},{\mathcal{T}}}^{\bar{\mathcal{R}}_i}\tilde{h}_{q,\mathcal{T}}^{\bar{\mathcal{R}_i}}m_{{\mathcal{R}},{\mathcal{T}},n}^{\bar{\mathcal{R}}_i}:\mathcal{R}\ni q,\mathcal{R}\cup\bar{\mathcal{R}}_i\not\ni q', m_{{\mathcal{R}},{\mathcal{T}},n}^{\bar{\mathcal{R}}_i}\in\mathcal{M}_{q'}[N],\forall\mathcal{T}\right\}.\label{eqn polynomial set 11}
\end{align}
Partition the set \eqref{eqn polynomial set 11} into subsets according to different powers of factors $\{\alpha_{{\mathcal{R}},{\mathcal{T}}}^{\bar{\mathcal{R}}_i}: \mathcal{R}\cup\bar{\mathcal{R}}_i\not\ni q'\}$. Since polynomials in different subsets are linearly independent due to different powers of factors $\alpha$, we only need to prove the linear independence within each subset.

Consider an arbitrary subset with the following form:
\begin{align}
&\left\{\frac{\alpha_{{\mathcal{R}_0},{\mathcal{T}_0}}^{\bar{\mathcal{R}}_0}\tilde{h}_{q,\mathcal{T}_0}^{\bar{\mathcal{R}}_0}}{\alpha_{{\mathcal{R}_0},{\mathcal{T}_0}}^{\bar{\mathcal{R}}_0}\tilde{h}_{q',\mathcal{T}_0}^{\bar{\mathcal{R}}_0}}\prod_{\substack{\mathcal{R},\bar{\mathcal{R}}_i,\mathcal{T}:\\ \mathcal{R}\cup\bar{\mathcal{R}}_i\not\ni q'}}\left[\alpha_{{\mathcal{R}},{\mathcal{T}}}^{\bar{\mathcal{R}}_i}\tilde{h}_{q',\mathcal{T}}^{\bar{\mathcal{R}}_i}\right]
^{s_{{\mathcal{R}},{\mathcal{T}}}^{i}}: \mathcal{R}_0\ni q,\mathcal{R}_0\cup\bar{\mathcal{R}}_0\not\ni q',\forall\mathcal{T}_0\right\},\label{eqn polynomial set 12}
\end{align}
where the power of  $\alpha_{{\mathcal{R}},{\mathcal{T}}}^{\bar{\mathcal{R}}_i}$ is $s_{{\mathcal{R}},{\mathcal{T}}}^{i}$. To prove the linear independence of polynomials in \eqref{eqn polynomial set 12}, it is equivalent to prove the linear independence of functions in
\begin{align}
\left\{\frac{\tilde{h}_{q,\mathcal{T}_0}^{\bar{\mathcal{R}}_0}}{\tilde{h}_{q',\mathcal{T}_0}^{\bar{\mathcal{R}}_0}}: \bar{\mathcal{R}}_0\not\ni q,\mathcal{R}_0\cup\bar{\mathcal{R}}_0\not\ni q',\forall\mathcal{T}_0\right\}.\label{eqn function}
\end{align}
Assume  there exist some factors $k_{\mathcal{T}_0}^{\bar{\mathcal{R}}_0}$ such that
\begin{align}
\sum_{\substack{\mathcal{R}_0,\bar{\mathcal{R}}_0,\mathcal{T}_0:\\\bar{\mathcal{R}}_0\not\ni q,\mathcal{R}_0\cup\bar{\mathcal{R}}_0\not\ni q'}}
k_{\mathcal{T}_0}^{\bar{\mathcal{R}}_0}\frac{\tilde{h}_{q,\mathcal{T}_0}^{\bar{\mathcal{R}}_0}}{\tilde{h}_{q',\mathcal{T}_0}^{\bar{\mathcal{R}}_0}}\equiv0\label{eqn polynomial set 13}.
\end{align}
Note that $\tilde{h}_{q,\mathcal{T}_0}^{\bar{\mathcal{R}}_0}=\sum_{p=1}^{t}h_{q,T_p}C_{T_p}(\tilde{h}_{q,\mathcal{T}_0}^{\bar{\mathcal{R}}_0})$, where $T_p\in\mathcal{T}_0$ and $C_{T_p}(\tilde{h}_{q,\mathcal{T}_0}^{\bar{\mathcal{R}}_0})$ is the cofactor of $h_{q,T_p}$ in \eqref{eqn proof2 equivalent channel} for $p=1,2,\ldots,t$. We can rewrite \eqref{eqn polynomial set 13} as
\begin{align}
\sum_{\substack{\mathcal{R}_0,\bar{\mathcal{R}}_0,\mathcal{T}_0:\\\bar{\mathcal{R}}_0\not\ni q,\mathcal{R}_0\cup\bar{\mathcal{R}}_0\not\ni q'}}
k_{\mathcal{T}_0}^{\bar{\mathcal{R}}_0}\frac{\tilde{h}_{q,\mathcal{T}_0}^{\bar{\mathcal{R}}_0}}{\tilde{h}_{q',\mathcal{T}_0}^{\bar{\mathcal{R}}_0}}
=&\sum_{\substack{\mathcal{R}_0,\bar{\mathcal{R}}_0,\mathcal{T}_0:\\\bar{\mathcal{R}}_0\not\ni q,\mathcal{R}_0\cup\bar{\mathcal{R}}_0\not\ni q'}}
k_{\mathcal{T}_0}^{\bar{\mathcal{R}}_0}\frac{\sum_{p=1}^{t}h_{q,T_p}C_{T_p}(\tilde{h}_{q,\mathcal{T}_0}^{\bar{\mathcal{R}}_0})}{\tilde{h}_{q',\mathcal{T}_0}^{\bar{\mathcal{R}}_0}}\notag\\
=&\sum_{p=1}^{N_T}h_{q,p}\sum_{\substack{\mathcal{R}_0,\bar{\mathcal{R}}_0,\mathcal{T}_0:\\ \bar{\mathcal{R}}_0\not\ni q,\mathcal{R}_0\cup\bar{\mathcal{R}}_0\not\ni q',\mathcal{T}_0\ni p}}k_{\mathcal{T}_0}^{\bar{\mathcal{R}}_0}\frac{C_{p}(\tilde{h}_{q,\mathcal{T}_0}^{\bar{\mathcal{R}}_0})}{\tilde{h}_{q',\mathcal{T}_0}^{\bar{\mathcal{R}}_0}}
\label{eqn polynomial set 14}.
\end{align}
Since $h_{q,p}$ is independent for different transmitter $p$, \eqref{eqn polynomial set 14} implies that
\begin{align}
\sum_{\substack{\mathcal{R}_0,\bar{\mathcal{R}}_0,\mathcal{T}_0:\\\bar{\mathcal{R}}_0\not\ni q,\mathcal{R}_0\cup\bar{\mathcal{R}}_0\not\ni q',\mathcal{T}_0\ni p}}k_{\mathcal{T}_0}^{\bar{\mathcal{R}}_0}\frac{C_{p}(\tilde{h}_{q,\mathcal{T}_0}^{\bar{\mathcal{R}}_0})}{\tilde{h}_{q',\mathcal{T}_0}^{\bar{\mathcal{R}}_0}}\equiv0,\label{eqn polynomial set 15}
\end{align}
for each transmitter $p$. We can rewrite \eqref{eqn polynomial set 15} as
\begin{align}
\sum_{\substack{\mathcal{R}_0,\bar{\mathcal{R}}_0,\mathcal{T}_0:\\\bar{\mathcal{R}}_0\not\ni q,\mathcal{R}_0\cup\bar{\mathcal{R}}_0\not\ni q',\mathcal{T}_0\ni p,\mathcal{T}_0\ni p_1}}k_{\mathcal{T}_0}^{\bar{\mathcal{R}}_0}\frac{C_{p}(\tilde{h}_{q,\mathcal{T}_0}^{\bar{\mathcal{R}}_0})}{\tilde{h}_{q',\mathcal{T}_0}^{\bar{\mathcal{R}}_0}}
\equiv
-\sum_{\substack{\mathcal{R}_0,\bar{\mathcal{R}}_0,\mathcal{T}_0:\\\bar{\mathcal{R}}_0\not\ni q,\mathcal{R}_0\cup\bar{\mathcal{R}}_0\not\ni q',\mathcal{T}_0\ni p,\mathcal{T}_0\not\ni p_1}}k_{\mathcal{T}_0}^{\bar{\mathcal{R}}_0}\frac{C_{p}(\tilde{h}_{q,\mathcal{T}_0}^{\bar{\mathcal{R}}_0})}{\tilde{h}_{q',\mathcal{T}_0}^{\bar{\mathcal{R}}_0}},\label{eqn polynomial set 16}
\end{align}
for an arbitrary transmitter $p_1\neq p$. Since $\{h_{1,p_1},h_{2,p_1},\ldots,h_{N_R,p_1}\}\setminus\{h_{q,p_1}\}$ only appear on the left side of \eqref{eqn polynomial set 16}, it is easy to see that \eqref{eqn polynomial set 16} holds only when both sides equal zero. Therefore, the summation of functions on the left side of \eqref{eqn polynomial set 16} equals  zero. The same arguments can be applied on these functions again, and we can find that the summation of functions satisfying $\{\bar{\mathcal{R}}_0\not\ni q,\mathcal{R}_0\cup\bar{\mathcal{R}}\not\ni q',\mathcal{T}_0\supseteq\{p,p_1,p_2\}\}$ equals  zero for an arbitrary transmitter $p_2\notin\{p,p_1\}$. Iteratively, we can see that the summation of functions satisfying $\{ \bar{\mathcal{R}}_0\not\ni q,\mathcal{R}_0\cup\bar{\mathcal{R}}_0\not\ni q'\}$ equals  zero for an arbitrary transmitter set $\mathcal{T}$, i.e.,
\begin{align}
\sum_{\substack{\mathcal{R}_0,\bar{\mathcal{R}}_0:\\ \bar{\mathcal{R}}_0\not \ni q,\mathcal{R}_0\cup\bar{\mathcal{R}}_0\not\ni q'}}k_{\mathcal{T}}^{\bar{\mathcal{R}}_0}\frac{C_{p}(\tilde{h}_{q,\mathcal{T}}^{\bar{\mathcal{R}}_0})}{\tilde{h}_{q',\mathcal{T}}^{\bar{\mathcal{R}}_0}}\equiv0.\label{eqn polynomial set 17}
\end{align}

Similar to the derivation of \eqref{eqn polynomial set 16} and \eqref{eqn polynomial set 17}, based on \eqref{eqn polynomial set 17}, we can find that the summation of functions equals  zero for an arbitrary transmitter set $\mathcal{T}$ and an arbitrary $\bar{\mathcal{R}}_i$, s.t. $q\notin \bar{\mathcal{R}}_i,q'\notin\mathcal{R}\cup\bar{\mathcal{R}}_i$. The detailed proof is omitted here. This implies that $k_{\mathcal{T}}^{\bar{\mathcal{R}}_i}\frac{C_{p}(\tilde{h}_{q,\mathcal{T}}^{\bar{\mathcal{R}}_i})}{\tilde{h}_{q',\mathcal{T}}^{\bar{\mathcal{R}}_i}}\equiv0$, and thus $k_{\mathcal{T}}^{\bar{\mathcal{R}}_i}=0$. Therefore, we proved the linear independence of functions in \eqref{eqn function}, and the linear independence of polynomials in \eqref{eqn polynomial set 11}.

Now we consider the polynomials of interference symbols in $\mathcal{M}_q[N+1]$. Given the construction of monomial set \eqref{eqn monomial set 1}, it is easy to see that the polynomials of interference symbols are linearly independent with each other and with the polynomials of the desired signals. Therefore, we finished the proof of linear independence of the polynomials of received symbols at receiver $q$. Similarly, the polynomials of received symbols at other receivers are also linearly independent. Therefore, the received signal matrix of each receiver is full-rank with probability 1 using \cite[Lemma 3]{CoMP}, and each receiver can decode its desired signals successfully. Since each receiver can decode $\binom{N_R-1}{r}\binom{N_T}{t}t N^{\binom{N_R-1}{r+1}\binom{N_T}{t}}$ symbols in $S$ time slots, a per-user DoF of
\begin{align}
d= \frac{\binom{N_R-1}{r}\binom{N_T}{t}t N^{\binom{N_R-1}{r+1}\binom{N_T}{t}}}{\binom{N_R-1}{r}\binom{N_T}{t}t N^{\binom{N_R-1}{r+1}\binom{N_T}{t}}+(N+1)^{\binom{N_R-1}{r+1}\binom{N_T}{t}}}\notag
\end{align}
is achieved. Letting $N\rightarrow\infty$, the per-user DoF of $\frac{\binom{N_R-1}{r}\binom{N_T}{t}t}{\binom{N_R-1}{r}\binom{N_T}{t}t +1}$ is achieved.

\subsection{$r+t\le N_R-2$}
Now, we consider the case when $r+t\le N_R-2$. There are two methods to deliver the messages. The first one is similar to the one used when $r+t\ge N_R$. We first split each message $W_{\mathcal{R},\mathcal{T}}$ into $\binom{N_R-r-1}{t-1}$ submessages, each associated with a unique receiver set $\mathcal{R}\cup\{R_{r+2},R_{r+3}.\ldots, R_{r+t}\}$, where $\mathcal{R}$ is the desired receiver set and $\{R_{r+2},R_{r+3}.\ldots, R_{r+t}\}$ is a set of  arbitrary $t-1$ receivers from the rest $N_R-r-1$ undesired receivers. There are $\binom{N_R}{r+t}$ different receiver sets in total, each having $\binom{r+t}{r+1}\binom{N_T}{t}$ submessages. In the delivery phase, submessages with different receiver sets are delivered individually in a time division manner, and submessages with a same receiver set are delivered together. Through this approach, we can see that the transmitters  send $\binom{r+t}{r+1}\binom{N_T}{t}$ submessages  each time and each receiver in the corresponding receiver set of these submessages desires $\binom{r+t-1}{r}\binom{N_T}{t}$ submessages, while the rest $N_R-r-t$ receivers do not desire any submessage of them. Therefore, we can regard the network each time as a $\binom{N_T}{t}\times \binom{r+t}{r+1}$ cooperative X-multicast network whose per-user achievable DoF is 1 in Case A. Note that each receiver only exists in $\binom{N_R-1}{r+t-1}$ of the $\binom{N_R}{r+t}$ receiver sets in total. Thus, the per-user DoF of $\binom{N_R-1}{r+t-1}/\binom{N_R}{r+t}=\frac{r+t}{N_R}$ is achieved by this method.

The second method is similar to the one in Case B, i.e. interference neutralization is used to neutralize each message at undesired receivers, and then the rest interferences are aligned together by asymptotic interference alignment.

We first consider the case when $t<N_T$. The delivery scheme when $t=N_T$ is slightly different, and will be presented later. When $t<N_T$, each message $W_{{\mathcal{R}},{\mathcal{T}}}$ is encoded into a $\binom{N_R-r-1}{t-1}t N^{(N_R-r-t)(N_T-t+1)}\times 1$ symbol vector
\begin{align}
\mathbf{x}_{{\mathcal{R}},{\mathcal{T}}}=((\mathbf{x}_{{\mathcal{R}},{\mathcal{T}}}^{1})^T,(\mathbf{x}_{{\mathcal{R}},{\mathcal{T}}}^{2})^T,
\ldots,(\mathbf{x}_{{\mathcal{R}},{\mathcal{T}}}^{{\varrho}})^T)^T,\notag
\end{align}
where $\varrho\triangleq\binom{N_R-r-1}{t-1}$, $N\in\mathds{Z}^+$, and $\mathbf{x}_{{\mathcal{R}},{\mathcal{T}}}^{i}$ ($i\in[\varrho]$) is a $t N^{(N_R-r-t)(N_T-t+1)}\times 1$ vector. We use $S\triangleq S_0+\binom{N_R-1}{r+1}\binom{N_R-r-2}{t-1}\binom{N_T}{t-1}(N+1)^{(N_R-r-t)(N_T-t+1)}$-symbol extension here, where $S_0\triangleq\binom{N_R-1}{r}\binom{N_T}{t}\binom{N_R-r-1}{t-1}t N^{(N_R-r-t)(N_T-t+1)}$. Note that $S_0$ is also the total number of symbols desired by each receiver. In each time slot $u$, the received signal at an arbitrary receiver $q\in[N_R]$, denoted as $y_q(u)$, is
\begin{align}
y_q(u)&=\sum_{\mathcal{R}:|\mathcal{R}|=r+1}\sum_{\mathcal{T}:|\mathcal{T}|=t}\sum_{i=1}^\varrho\left[\sum_{p\in \mathcal{T}}h_{qp}(u)\left(\mathbf{v}_{{\mathcal{R}},{\mathcal{T}},p}^i(u)\right)^T\right]\mathbf{x}_{{\mathcal{R}},{\mathcal{T}}}^i,\label{eqn proof3 receive signal}
\end{align}
where $h_{qp}(u)$ is the channel realization, and $\mathbf{v}_{{\mathcal{R}},{\mathcal{T}},p}^i(u)$ is the $tN^{(N_R-r-t)(N_T-t+1)}\times1$ precoding vector of symbol vector $\mathbf{x}_{{\mathcal{R}},{\mathcal{T}}}^i$ at transmitter $p$.

To apply interference neutralization, we consider an arbitrary symbol vector $\mathbf{x}_{{\mathcal{R}},{\mathcal{T}}}^i$ desired by receiver multicast group $\mathcal{R}=\{R_1,R_2,\ldots,R_{r+1}\}$, transmitted by transmitter cooperation group $\mathcal{T}=\{T_{1},T_{2},\ldots,T_{t}\}$ , and whose undesired receiver set is $\bar{\mathcal{R}}\!=\![N_R]\setminus\mathcal{R}\!=\!\{\bar{R}_1,\bar{R}_2,\ldots,\bar{R}_{N_R-r-1}\}$. We assume that each $\mathbf{x}_{{\mathcal{R}},{\mathcal{T}}}^i$ will be neutralized at a distinct receiver set $\bar{\mathcal{R}}_i\subset\bar{\mathcal{R}}$ with $|\bar{\mathcal{R}}_i|=t-1$. Consider an arbitrary $\mathbf{x}_{{\mathcal{R}},{\mathcal{T}}}^i$ with $\bar{\mathcal{R}}_i=\{\bar{R}_{i,1},\bar{R}_{i,2},\ldots,\bar{R}_{i,t-1}\}$. Then, the precoder must satisfy
\begin{align}
\sum_{p\in \mathcal{T}}h_{qp}(u)v_{{\mathcal{R}},{\mathcal{T}},p,n}^i(u)=0, \forall q\in \bar{\mathcal{R}}_i,\forall n\in[tN^{(N_R-r-t)(N_T-t+1)}],\forall u\in[S],\label{eqn proof3 neutralization}
\end{align}
where $v_{{\mathcal{R}},{\mathcal{T}},p,n}^i(u)$ is the $n$-th element of $\mathbf{v}_{{\mathcal{R}},{\mathcal{T}},p}^i(u)$.  Consider the following matrix:
\begin{align}
\left(
\begin{matrix}
h_{\bar{R}_{i,1},T_1}&h_{\bar{R}_{i,1},T_2}&\cdots&h_{\bar{R}_{i,1},T_t}\\
h_{\bar{R}_{i,2},T_1}&h_{\bar{R}_{i,2},T_2}&\cdots&h_{\bar{R}_{i,2},T_t}\\
\cdots&\cdots&\cdots&\cdots\\
h_{\bar{R}_{i,t-1},T_1}&h_{\bar{R}_{i,t-1},T_2}&\cdots&h_{\bar{R}_{i,t-1},T_t}\\
a_{1}&a_{2}&\cdots&a_{t}
\end{matrix}
\right)\triangleq \mathbf{H}_{\bar{\mathcal{R}}_i,\mathcal{T}},\label{eqn proof3 neutralization matrix}
\end{align}
for any $\{a_{1},a_2,\ldots,a_{t}\}$. Define $c_p$ as the cofactor of $a_p$ such that
\begin{align}
\sum_{p=1}^{t}a_{p}c_{p}=\det(\mathbf{H}_{\bar{\mathcal{R}}_i,\mathcal{T}}).\label{eqn proof3 neutralization determinant}
\end{align}
We then design $v_{{\mathcal{R}},{\mathcal{T}},p,n}^i(u)$ as
\begin{align}
v_{{\mathcal{R}},{\mathcal{T}},p,n}^i(u)=\alpha_{{\mathcal{R}},{\mathcal{T}}}^{\tilde{\mathcal{R}}_i}(u)c_p(u) z_{{\mathcal{R}},{\mathcal{T}},n}^{\bar{\mathcal{R}}_i}(u),\label{eqn proof3 v}
\end{align}
where $\alpha_{{\mathcal{R}},{\mathcal{T}}}^{\tilde{\mathcal{R}}_i}(u)$ is chosen i.i.d. from a continuous distribution for all $\{\mathcal{R,T},\bar{\mathcal{R}}_i,u\}$, $c_p(u)$ is the cofactor $c_p$ by taking channel realization $\{h_{qp}(u)\}$ into \eqref{eqn proof3 neutralization matrix} and \eqref{eqn proof3 neutralization determinant}, and $z_{{\mathcal{R}},{\mathcal{T}},n}^{\bar{\mathcal{R}}_i}(u)$ will be determined later. By such construction, the condition in \eqref{eqn proof3 neutralization} is satisfied. For example, at receiver $\bar{R}_{i,1}$, we have
\begin{align}
&\sum_{p\in \mathcal{T}}h_{\bar{R}_{i,1},p}(u)\alpha_{{\mathcal{R}},{\mathcal{T}}}^{\bar{\mathcal{R}}_i}(u)c_p(u)z_{{\mathcal{R}},{\mathcal{T}},n}^{\bar{\mathcal{R}}_i}(u)\notag\\
=\,&\alpha_{{\mathcal{R}},{\mathcal{T}}}^{\bar{\mathcal{R}}_i}(u)z_{{\mathcal{R}},{\mathcal{T}},n}^{\bar{\mathcal{R}}_i}(u)\sum_{p\in \mathcal{T}}h_{\bar{R}_{i,1},p}(u)c_p(u)\notag\\
=\,&\alpha_{{\mathcal{R}},{\mathcal{T}}}^{\bar{\mathcal{R}}_i}(u)z_{{\mathcal{R}},{\mathcal{T}},n}^{\bar{\mathcal{R}}_i}(u)
\cdot\begin{vmatrix}
h_{\bar{R}_{i,1},T_1}(u)&h_{\bar{R}_{i,1},T_2}(u)&\cdots&h_{\bar{R}_{i,1},T_t}(u)\\
h_{\bar{R}_{i,2},T_1}(u)&h_{\bar{R}_{i,2},T_2}(u)&\cdots&h_{\bar{R}_{i,2},T_t}(u)\\
\cdots&\cdots&\cdots&\cdots\\
h_{\bar{R}_{i,t-1},T_1}(u)&h_{\bar{R}_{i,t-1},T_2}(u)&\cdots&h_{\bar{R}_{i,t-1},T_t}(u)\\
h_{\bar{R}_{i,1},T_1}(u)&h_{\bar{R}_{i,1},T_2}(u)&\cdots&h_{\bar{R}_{i,1},T_t}(u)
\end{vmatrix}\notag\\
=\,&0.
\end{align}

By the above construction of precoders, it can be seen that symbol vectors $\mathbf{x}_{{\mathcal{R}},{\mathcal{T}}}^i$ unwanted by receiver $q\in\bar{\mathcal{R}}_i$ are all neutralized. Then, the received signal at an arbitrary receiver $q\in[N_R]$ can be rewritten as
\begin{align}
y_q(u)=&\sum_{\mathcal{R}:|\mathcal{R}|=r+1,\mathcal{R}\ni q}\sum_{\mathcal{T}:|\mathcal{T}|=t}\sum_{i=1}^\varrho\left[\sum_{p\in \mathcal{T}}h_{qp}(u)\left(\mathbf{v}_{{\mathcal{R}},{\mathcal{T}},p}^i(u)\right)^T\right]\mathbf{x}_{{\mathcal{R}},{\mathcal{T}}}^i\notag\\
&+\sum_{\substack{\bar{\mathcal{R}}_i,\mathcal{R}:\\|\mathcal{R}|=r+1,\mathcal{R}\cup\bar{\mathcal{R}}_i\not\ni q}}\sum_{\mathcal{T}:|\mathcal{T}|=t}\left[\sum_{p\in \mathcal{T}}h_{qp}(u)\left(\mathbf{v}_{{\mathcal{R}},{\mathcal{T}},p}^i(u)\right)^T\right]\mathbf{x}_{{\mathcal{R}},{\mathcal{T}}}^i,\label{eqn proof3 receive signal after neutralization}
\end{align}
where the first term is the desired messages of receiver $q$ and the second term is the residual interferences.

Now, we aim to apply asymptotic interference alignment to align the interference term in \eqref{eqn proof3 receive signal after neutralization} at the same sub-space. In specific, symbol vector $\mathbf{x}_{{\mathcal{R}},{\mathcal{T}}}^{i}$ is aligned with other symbol vectors which have the same receiver multicast group $\mathcal{R}$, neutralized at the same receiver set $\bar{\mathcal{R}}_i$, and only differ from one transmitter at transmitter cooperation set $\mathcal{T}$. For an arbitrary $\mathbf{x}_{\mathcal{R},\mathcal{T}}^i$, consider the following monomial sets:
\begin{align}
\mathcal{M}_{\mathcal{R},\mathcal{T}^c}^{\bar{\mathcal{R}}_i}\left[N\right]
=\left\{\prod_{\substack{p,q:\\p\notin\mathcal{T}^c,q\notin\mathcal{R}\cup\bar{\mathcal{R}}_i}}\left[\alpha_{{\mathcal{R}},{\{p\}\cup\mathcal{T}^c}}^{\bar{\mathcal{R}}_i}\tilde{h}_{q,\{p\}\cup\mathcal{T}^c}^{\bar{\mathcal{R}}_i}\right]^{s_{{\mathcal{R}},{\{p\}\cup\mathcal{T}^c}}^{i,q}}:
1\le s_{{\mathcal{R}},{\{p\}\cup\mathcal{T}^c}}^{i,q}\le N \right\},\label{eqn monomial set 2c}
\end{align}
where $\mathcal{T}^c\subset\mathcal{T},|\mathcal{T}^c|=t-1$ and $\tilde{h}_{q,\mathcal{T}}^{\bar{\mathcal{R}}_i}$ is defined as
\begin{align}
\tilde{h}_{q,\mathcal{T}}^{\bar{\mathcal{R}}_i}
\triangleq\sum_{p=T_1}^{T_t}h_{q,p}c_{p}
=\begin{vmatrix}
h_{\bar{R}_{i,1},T_1}&h_{\bar{R}_{i,1},T_2}&\cdots&h_{\bar{R}_{i,1},T_t}\\
h_{\bar{R}_{i,2},T_1}&h_{\bar{R}_{i,2},T_2}&\cdots&h_{\bar{R}_{i,2},T_t}\\
\cdots&\cdots&\cdots&\cdots\\
h_{\bar{R}_{i,t-1},T_1}&h_{\bar{R}_{i,t-1},T_2}&\cdots&h_{\bar{R}_{i,t-1},T_t}\\
h_{q,T_1}&h_{q,T_2}&\cdots&h_{q,T_t}
\end{vmatrix}.\label{eqn proof3 equivalent channel}
\end{align}
There are $t$ different monomial sets for symbol vector $\mathbf{x}_{{\mathcal{R}},{\mathcal{T}}}^{i}$, each with cardinality $N^{(N_R-r-t)(N_T-t+1)}$. Define $\mathcal{M}_{\mathcal{R},\mathcal{T}}^{\bar{\mathcal{R}}_i}\left[N\right]=\{\mathcal{M}_{\mathcal{R},\mathcal{T}^c}^{\bar{\mathcal{R}}_i}\left[N\right]:\mathcal{T}^c\subset \mathcal{T},|\mathcal{T}^c|=t-1\}$. For each element $v_{{\mathcal{R}},{\mathcal{T}},p,n}^i(u)$ in $\mathbf{v}_{{\mathcal{R}},{\mathcal{T}},p}^i(u)$ , the element $z_{{\mathcal{R}},{\mathcal{T}},n}^{\bar{\mathcal{R}}_i}(u)$ in \eqref{eqn proof3 v} is given by a unique monomial $m_{{\mathcal{R}},{\mathcal{T}},n}^{\bar{\mathcal{R}}_i}(u)$ in $\mathcal{M}_{\mathcal{R},\mathcal{T}}^{\bar{\mathcal{R}}_i}\left[N\right]$ by taking $\{h_{qp}(u)\}$ and $\{\alpha_{{\mathcal{R}},{\mathcal{T}}}^{\bar{\mathcal{R}}_i}(u)\}$ into \eqref{eqn monomial set 2c}. By such construction, the summation $\sum_{p\in \mathcal{T}}h_{qp}(u)v_{{\mathcal{R}},{\mathcal{T}},p,n}^i(u)$ for $\mathbf{x}_{{\mathcal{R}},{\mathcal{T}}}^i$ such that $q\notin \mathcal{R}\cup\bar{\mathcal{R}}_i$ is $\alpha_{{\mathcal{R}},{\mathcal{T}}}^{\bar{\mathcal{R}}_i}(u)\tilde{h}_{q,\mathcal{T}}^{\bar{\mathcal{R}}_i}(u)m_{{\mathcal{R}},{\mathcal{T}},n}^{\bar{\mathcal{R}}_i}(u)$ and satisfies
\begin{align}
\alpha_{{\mathcal{R}},{\mathcal{T}}}^{\bar{\mathcal{R}}_i}(u)\tilde{h}_{q,\mathcal{T}}^{\bar{\mathcal{R}}_i}(u)m_{{\mathcal{R}},{\mathcal{T}},n}^{\bar{\mathcal{R}}_i}(u)\in\mathcal{M}_{\mathcal{R},\mathcal{T}}^{\bar{\mathcal{R}}_i}\left[N+1\right](u),\notag
\end{align}
In specific, denote $m_{{\mathcal{R}},{\mathcal{T}^c},n}^{\bar{\mathcal{R}}_i}(u)$ as the monomial selected by $v_{{\mathcal{R}},{\mathcal{T}},p,n}^i(u)$ in an arbitrary $\mathcal{M}_{\mathcal{R},\mathcal{T}^c}^{\bar{\mathcal{R}}_i}\left[N\right](u)$, where $\mathcal{T}^c\subset\mathcal{T}$. We have
\begin{align}
\alpha_{{\mathcal{R}},{\mathcal{T}}}^{\bar{\mathcal{R}}_i}(u)\tilde{h}_{q,\mathcal{T}}^{\bar{\mathcal{R}}_i}(u)m_{{\mathcal{R}},{\mathcal{T}^c},n}^{\bar{\mathcal{R}}_i}(u)\in\mathcal{M}_{\mathcal{R},\mathcal{T}^c}^{\bar{\mathcal{R}}_i}\left[N+1\right](u),\notag
\end{align}
which means that the symbols intended for the same receiver multicast group $\mathcal{R}$, neutralized at the same receiver set $\bar{\mathcal{R}}_i$, with its precoder $z_{{\mathcal{R}},{\mathcal{T}},n}^{\bar{\mathcal{R}}_i}$ constructed from the same monomial set $\mathcal{M}_{\mathcal{R},\mathcal{T}^c}^{\bar{\mathcal{R}}_i}\left[N\right]$ are aligned in the same subspace with dimension $(N+1)^{(N_R-r-t)(N_T-t+1)}$.

By the design of both interference neutralization and interference alignment, the received signal at receiver $q$ is given by
\begin{align}
y_q(u)=&\sum_{\mathcal{R}:|\mathcal{R}|=r+1, \mathcal{R}\ni q}\sum_{\mathcal{T}:|\mathcal{T}|=t}\sum_{i=1}^\varrho\left[\sum_{p\in \mathcal{T}}h_{qp}(u)\left(\mathbf{v}_{{\mathcal{R}},{\mathcal{T}},p}^i(u)\right)^T\right]\mathbf{x}_{{\mathcal{R}},{\mathcal{T}}}^i\notag\\
&+\sum_{\substack{\bar{\mathcal{R}}_i,\mathcal{R}:\\|\mathcal{R}|=r+1,\mathcal{R}\cup\bar{\mathcal{R}}_i\not\ni q}}\sum_{\mathcal{T}:|\mathcal{T}|=t}\left[\sum_{p\in \mathcal{T}}h_{qp}(u)\left(\mathbf{v}_{{\mathcal{R}},{\mathcal{T}},p}^i(u)\right)^T\right]\mathbf{x}_{{\mathcal{R}},{\mathcal{T}}}^i\notag\\
=&\sum_{\mathcal{R}:|\mathcal{R}|=r+1, \mathcal{R}\ni q}\sum_{\mathcal{T}:|\mathcal{T}|=t}\sum_{i=1}^\varrho\left[\sum_{p\in \mathcal{T}}h_{qp}(u)\left(\mathbf{v}_{{\mathcal{R}},{\mathcal{T}},p}^i(u)\right)^T\right]\mathbf{x}_{{\mathcal{R}},{\mathcal{T}}}^i\notag\\
&+\sum_{\substack{\bar{\mathcal{R}}_i,\mathcal{R}:\\|\mathcal{R}|=r+1,\mathcal{R}\cup\bar{\mathcal{R}}_i\not\ni q}}\sum_{\mathcal{T}^c:|\mathcal{T}^c|=t-1}\sum_{m(u)\in\mathcal{M}_{\mathcal{R},\mathcal{T}^c}^{\bar{\mathcal{R}}_i}\left[N+1\right](u)}m(u)x_{m(u)},\label{eqn proof3 receive signal after alignment}
\end{align}
where $x_{m(u)}$ is the sum of interference symbols whose received factor is $m(u)$ at receiver $q$. To successfully decode the $\binom{N_R-1}{r}\binom{N_T}{t}$ desired messages of receiver $q$, we need to assure that the $S \times S$ received signal matrix whose column vectors are
\begin{align}
&\left\{\left(\alpha_{{\mathcal{R}},{\mathcal{T}}}^{\bar{\mathcal{R}}_i}(u)\tilde{h}_{q,\mathcal{T}}^{\bar{\mathcal{R}}_i}(u)m_{{\mathcal{R}},{\mathcal{T}},n}^{\bar{\mathcal{R}}_i}(u)\right)_{u=1}^S:|\mathcal{R}|=r+1,\right.\notag\\
&\quad \mathcal{R}\ni q,|T|=t,i\in[\varrho],n\in[tN^{(N_R-r-t)(N_T-t+1)}]\Big\}\notag\\
&\cup\left\{\left(m(u)\right)_{u=1}^S:m(u)\in\mathcal{M}_{\mathcal{R},\mathcal{T}^c}^{\bar{\mathcal{R}}_i}\left[N+1\right](u), \mathcal{R}\cup\bar{\mathcal{R}}_i\not\ni q,|\mathcal{T}^c|=t-1\right\}\notag
\end{align}
is full-rank with probability 1. Since the construction method of $\{\mathbf{v}_{{\mathcal{R}},{\mathcal{T}},p}^i(u)\}$ and the formation of $\{\alpha_{{\mathcal{R}},{\mathcal{T}}}^{\bar{\mathcal{R}}_i}(u)\tilde{h}_{q,\mathcal{T}}^{\bar{\mathcal{R}}_i}(u)m_{{\mathcal{R}},{\mathcal{T}},n}^{\bar{\mathcal{R}}_i}(u)\}$ and $\{m(u)\}$  are the same at each time slot $u$, based on \cite[Lemma 3]{CoMP}, we only need to prove the linear independence of these polynomial functions, which is given by
\begin{align}
&\left\{\alpha_{{\mathcal{R}},{\mathcal{T}}}^{\bar{\mathcal{R}}_i}\tilde{h}_{q,\mathcal{T}}^{\bar{\mathcal{R}}_i}m_{{\mathcal{R}},{\mathcal{T}},n}^{\bar{\mathcal{R}}_i}:\mathcal{R}\ni q,|\mathcal{R}|=r+1,|\mathcal{T}|=t,i\in[\varrho],n\in[tN^{(N_R-r-t)(N_T-t+1)}]\right\}\notag\\
&\cup \left\{m_{{\mathcal{R}},{\mathcal{T}^c},n}^{\bar{\mathcal{R}}_i}:m_{{\mathcal{R}},{\mathcal{T}^c},n}^{\bar{\mathcal{R}}_i}\in\mathcal{M}_{\mathcal{R},\mathcal{T}^c}^{\bar{\mathcal{R}}_i}\left[N+1\right], \mathcal{R}\cup \bar{\mathcal{R}}_i\not\ni q,|\mathcal{T}^c|=t-1\right\}.\notag
\end{align}

First, we can see that polynomials corresponding to different $\mathcal{R}$ and $\bar{\mathcal{R}}_i$ are linearly independent because they have different factors $\left\{\alpha_{{\mathcal{R}},{\mathcal{T}}}^{\bar{\mathcal{R}}_i}\right\}$. Then, we only need to consider polynomials corresponding to the same $\mathcal{R}$ and $\bar{\mathcal{R}}_i$. Let us first consider polynomials of desired symbols corresponding to an arbitrary $\mathcal{R}$ and $\bar{\mathcal{R}}_i$ ($\mathcal{R}\ni q$), i.e.
\begin{align}
&\bigcup_{\mathcal{T}:|\mathcal{T}|=t}\left\{\alpha_{{\mathcal{R}},{\mathcal{T}}}^{\bar{\mathcal{R}}_i}\tilde{h}_{q,\mathcal{T}}^{\bar{\mathcal{R}}_i}m_{{\mathcal{R}},{\mathcal{T}},n}^{\bar{\mathcal{R}}_i}: m_{{\mathcal{R}},{\mathcal{T}},n}^{\bar{\mathcal{R}}_i}\in \mathcal{M}_{\mathcal{R},\mathcal{T}}^{\bar{\mathcal{R}}_i}\left[N\right]\right\}\notag\\
=&\bigcup_{\mathcal{T}:|\mathcal{T}|=t}\bigcup_{\substack{\mathcal{T}^c:\\\mathcal{T}^c\subset\mathcal{T},|\mathcal{T}^c|=t-1}}\left\{\alpha_{{\mathcal{R}},{\mathcal{T}}}^{\bar{\mathcal{R}}_i}\tilde{h}_{q,\mathcal{T}}^{\bar{\mathcal{R}}_i}m_{{\mathcal{R}},{\mathcal{T}^c},n}^{\bar{\mathcal{R}}_i}:m_{{\mathcal{R}},{\mathcal{T}^c},n}^{\bar{\mathcal{R}}_i}\in \mathcal{M}_{\mathcal{R},\mathcal{T}^c}^{\bar{\mathcal{R}}_i}\left[N\right]\right\}\notag\\
=&\bigcup_{\mathcal{T}^c:|\mathcal{T}^c|=t-1}\left\{\alpha_{{\mathcal{R}},{\mathcal{T}^c\cup\{p\}}}^{\bar{\mathcal{R}}_i}\tilde{h}_{q,\mathcal{T}^c\cup\{p\}}^{\bar{\mathcal{R}}}m_{{\mathcal{R}},{\mathcal{T}^c},n}^{\bar{\mathcal{R}}_i}:p\notin \mathcal{T}^c,m_{{\mathcal{R}},{\mathcal{T}^c},n}^{\bar{\mathcal{R}}_i}\in \mathcal{M}_{\mathcal{R},\mathcal{T}^c}^{\bar{\mathcal{R}}_i}\left[N\right]\right\}.\label{eqn polynomial set 23}
\end{align}
Partitioning these polynomials into subsets w.r.t. $\mathcal{T}^c$ as in \eqref{eqn polynomial set 23}, it can be seen that polynomials $\{\alpha_{{\mathcal{R}},{\mathcal{T}^c\cup\{p\}}}^{\bar{\mathcal{R}}_i}\tilde{h}_{q,\mathcal{T}^c\cup\{p\}}^{\bar{\mathcal{R}}_i}m_{{\mathcal{R}},{\mathcal{T}^c},n}^{\bar{\mathcal{R}}_i}\}$ for different $\mathcal{T}^c$ are linearly independent. This is because the polynomials for each $\mathcal{T}^c$ have a unique factor set $\left\{\alpha_{\mathcal{R},{\mathcal{T}^c\cup\{p\}}}^{\bar{\mathcal{R}}_i}: p\notin \mathcal{T}^c\right\}$.

Then, we only need to consider the linear independence of polynomials with the same $\mathcal{R}$, $\bar{\mathcal{R}}_i$ and $\mathcal{T}^c$ ($\mathcal{R}\ni q$):
\begin{align}
&\left\{\alpha_{{\mathcal{R}},{\mathcal{T}^c\cup\{p\}}}^{\bar{\mathcal{R}}_i}\tilde{h}_{q,\mathcal{T}^c\cup\{p\}}^{\bar{\mathcal{R}}_i}m_{{\mathcal{R}},{\mathcal{T}^c},n}^{\bar{\mathcal{R}}_i}:p\notin \mathcal{T}^c,m_{{\mathcal{R}},{\mathcal{T}^c},n}^{\bar{\mathcal{R}}_i}\in \mathcal{M}_{\mathcal{R},\mathcal{T}^c}^{\bar{\mathcal{R}}_i}\left[N\right]\right\}.\label{eqn polynomial set 24}
\end{align}
Note that $h_{qp}$ only exists in polynomials whose symbols are transmitted by $\mathcal{T}^c\cup\{p\}$, i.e.
\begin{align}
\left\{\alpha_{{\mathcal{R}},{\mathcal{T}^c\cup\{p\}}}^{\bar{\mathcal{R}}_i}\tilde{h}_{q,\mathcal{T}^c\cup\{p\}}^{\bar{\mathcal{R}}_i}m_{{\mathcal{R}},{\mathcal{T}^c},n}^{\bar{\mathcal{R}}_i}: m_{{\mathcal{R}},{\mathcal{T}^c},n}^{\bar{\mathcal{R}}_i}\in \mathcal{M}_{\mathcal{R},\mathcal{T}^c}^{\bar{\mathcal{R}}_i}\left[N\right]\right\}.\label{eqn polynomial set 25}
\end{align}
Therefore, it can be seen that polynomials $\{\alpha_{{\mathcal{R}},{\mathcal{T}^c\cup\{p\}}}^{\bar{\mathcal{R}}_i}\tilde{h}_{q,\mathcal{T}^c\cup\{p\}}^{\bar{\mathcal{R}}_i}m_{{\mathcal{R}},{\mathcal{T}},n}^{\bar{\mathcal{R}}_i}\}$ are linearly independent for different $p$. This implies that we only need to prove the linear independence of polynomials for the same $\mathcal{R}$, $\bar{\mathcal{R}}_i$, $\mathcal{T}^c$ and $p$ ($\mathcal{R}\ni q$), which is equivalent to prove the linear independence of polynomials in $\mathcal{M}_{\mathcal{R},\mathcal{T}^c}^{\bar{\mathcal{R}}_i}\left[N\right]$ in \eqref{eqn monomial set 2c}. Given that each $\alpha_{{\mathcal{R}},{\{p\}\cup\mathcal{T}^c}}^{\bar{\mathcal{R}}_i}\tilde{h}_{q,\{p\}\cup\mathcal{T}^c}^{\bar{\mathcal{R}}_i}$ in \eqref{eqn monomial set 2c} has a unique channel coefficient $h_{q,p}$, we can guarantee the Jacobian matrix of polynomials $\{\alpha_{{\mathcal{R}},{\{p\}\cup\mathcal{T}^c}}^{\bar{\mathcal{R}}_i}\tilde{h}_{q,\{p\}\cup\mathcal{T}^c}^{\bar{\mathcal{R}}_i}: q\notin\mathcal{R}\cup\bar{\mathcal{R}}_i,p\notin \mathcal{T}^c\}$ are full row rank. Using  \cite[Theorem 3]{algebraic} in Page 135 and \cite[Lemma 1]{CoMP}, it can be seen that these polynomials are algebraically independent. Thus, polynomials in $\mathcal{M}_{\mathcal{R},\mathcal{T}^c}^{\bar{\mathcal{R}}_i}\left[N\right]$ are linearly independent, and we finished the proof of linear independence of polynomials whose symbols are desired by receiver $q$.

We can directly apply the arguments above to the polynomials corresponding to the interference, and show that these polynomials are also linearly independent. Therefore, we finished the proof of linear independence of the polynomials of received symbols at receiver $q$. Similarly, the polynomials of received symbols at other receivers are also linearly independent. Therefore, the received signal matrix at each receiver is full rank with probability 1 using \cite[Lemma 3]{CoMP}, and each receiver can decode its desired signals successfully.

Since each receiver can decode $\binom{N_R-1}{r}\binom{N_T}{t}\binom{N_R-r-1}{t-1}t N^{(N_R-r-t)(N_T-t+1)}$ symbols in $S=\binom{N_R-1}{r}\binom{N_T}{t}\binom{N_R-r-1}{t-1}t N^{(N_R-r-t)(N_T-t+1)}+\binom{N_R-1}{r+1}\binom{N_R-r-2}{t-1}\binom{N_T}{t-1}(N+1)^{(N_R-r-t)(N_T-t+1)}$-symbol extension, a per-user DoF of
\begin{align}
d=\frac{\binom{N_R-1}{r}\binom{N_T}{t}\binom{N_R-r-1}{t-1}t N^{(N_R-r-t)(N_T-t+1)}}
{S}\notag
\end{align}
is achieved. Letting $N\rightarrow\infty$, the per-user DoF of
\begin{align}
d = \frac{\binom{N_R-1}{r}\binom{N_T}{t}\binom{N_R-r-1}{t-1}t}
{\binom{N_R-1}{r}\binom{N_T}{t}\binom{N_R-r-1}{t-1}t+\binom{N_R-1}{r+1}\binom{N_R-r-2}{t-1}\binom{N_T}{t-1}}\label{eqn dof 2}
\end{align}
is achieved.

Next, we consider the case when $t=N_T$ where interference alignment is not applied. Each message $W_{{\mathcal{R}},{[N_T]}}$ is encoded into a $\binom{N_R-r-1}{N_T-1}\times 1$ symbol vector
\begin{align}
\mathbf{x}_{{\mathcal{R}},{[N_T]}}=(x_{{\mathcal{R}},{[N_T]}}^1,x_{{\mathcal{R}},{[N_T]}}^2,\ldots,x_{{\mathcal{R}},{[N_T]}}^\rho)^T\notag
\end{align}
where $\rho\triangleq\binom{N_R-r-1}{N_T-1}$. A symbol extension of $S=\binom{N_R-1}{r}\binom{N_R-r-1}{N_T-1}+\binom{N_R-1}{r+1}\binom{N_R-r-2}{N_T-1}$ is used here.  In each time slot $u$, the received signal at an arbitrary receiver $q\in[N_R]$, denoted by $y_q(u)$, is given by
\begin{align}
y_q(u)&=\sum_{\mathcal{R}:|\mathcal{R}|=r+1}\sum_{i=1}^\rho\left[\sum_{p\in [N_T]}h_{qp}(u)v_{{\mathcal{R}},{[N_T]},p}^i(u)\right]x_{{\mathcal{R}},{[N_T]}}^i\label{eqn proof2 receive signal}
\end{align}
where $h_{qp}(u)$ is the channel realization, and $v_{{\mathcal{R}},{[N_T]},p}^i(u)$ is the precoder of symbol $x_{{\mathcal{R}},{[N_T]}}^i$ at transmitter $p$.

To apply interference neutralization, we consider an arbitrary symbol $x_{{\mathcal{R}},{[N_T]}}^i$ desired by receiver multicast group $\mathcal{R}=\{R_1,R_2,\ldots,R_{r+1}\}$, and whose undesired receiver set is $\bar{\mathcal{R}}=\{\bar{R}_1,\bar{R}_2,\ldots,\bar{R}_{N_R-r-1}\}$. We assume that $x_{{\mathcal{R}},{[N_T]}}^i$ will be neutralized at a distinct receiver set $\bar{\mathcal{R}}_i\subset\bar{\mathcal{R}}$ with $|\bar{\mathcal{R}}_i|=N_T-1$. Consider an arbitrary $x_{{\mathcal{R}},{[N_T]}}^i$ with $\bar{\mathcal{R}}_i=\{\bar{R}_{i,1},\bar{R}_{i,2},\ldots,\bar{R}_{i,N_T-1}\}$. We must have
\begin{align}
\sum_{p\in [N_T]}h_{qp}(u)v_{{\mathcal{R}},{[N_T]},p}^i(u)=0, \forall q\in \bar{\mathcal{R}}_i,\forall u\in[\rho]\label{eqn proof4 neutralization}
\end{align}
Consider the following $N_T\times N_T$ matrix:
\begin{align}
\left(
\begin{matrix}
h_{\bar{R}_{i,1},1}&h_{\bar{R}_{i,1},2}&\cdots&h_{\bar{R}_{i,1},N_T}\\
h_{\bar{R}_{i,2},1}&h_{\bar{R}_{i,2},2}&\cdots&h_{\bar{R}_{i,2},N_T}\\
\cdots&\cdots&\cdots&\cdots\\
h_{\bar{R}_{i,N_T-1},1}&h_{\bar{R}_{i,N_T-1},2}&\cdots&h_{\bar{R}_{i,N_T-1},N_T}\\
a_{1}&a_{2}&\cdots&a_{N_T}
\end{matrix}
\right)\triangleq \mathbf{H}_{\bar{\mathcal{R}}_i},\label{eqn proof4 neutralization matrix}
\end{align}
for any $\{a_{1},a_2,\ldots,a_{N_T}\}$. Define $c_p$ as the cofactor of $a_p$ such that
\begin{align}
\sum_{p=1}^{N_T}a_{p}c_{p}=\det(\mathbf{H}_{\bar{\mathcal{R}}_i}).\label{eqn proof4 neutralization determinant}
\end{align}
We then design $v_{{\mathcal{R}},{[N_T]},p}^i(u)$ as
\begin{align}
v_{{\mathcal{R}},{[N_T]},p}^i(u)=\alpha_{{\mathcal{R}},{[N_T]}}^{\bar{\mathcal{R}}_i}(u)c_p(u),\label{eqn proof4 v}
\end{align}
where $\alpha_{{\mathcal{R}},{[N_T]}}^{\bar{\mathcal{R}}_i}(u)$ is chosen i.i.d. from a continuous distribution for all $\{\mathcal{R},\bar{\mathcal{R}}_i,u\}$, and $c_p(u)$ is the cofactor $c_p$ by taking channel realization $\{h_{qp}(u)\}$ into \eqref{eqn proof4 neutralization matrix} and \eqref{eqn proof4 neutralization determinant}. By such construction, the condition in \eqref{eqn proof4 neutralization} is satisfied. For example, at receiver $\bar{R}_{i,1}$, we have
\begin{align}
&\sum_{p\in [N_T]}h_{\bar{R}_{i,1},p}(u)\alpha_{{\mathcal{R}},{[N_T]}}^{\bar{\mathcal{R}}_i}(u)c_p(u)\notag\\
=\,&\alpha_{{\mathcal{R}},{[N_T]}}^{\bar{\mathcal{R}}_i}(u)\sum_{p\in [N_T]}h_{\bar{R}_{i,1},p}(u)c_p(u)\notag\\
=\,&\alpha_{{\mathcal{R}},{[N_T]}}^{\bar{\mathcal{R}}_i}(u)\cdot\begin{vmatrix}
h_{\bar{R}_{i,1},1}(u)&h_{\bar{R}_{i,1},2}(u)&\cdots&h_{\bar{R}_{i,1},N_T}(u)\\
h_{\bar{R}_{i,2},1}(u)&h_{\bar{R}_{i,2},2}(u)&\cdots&h_{\bar{R}_{i,2},N_T}(u)\\
\cdots&\cdots&\cdots&\cdots\\
h_{\bar{R}_{i,N_T-1},1}(u)&h_{\bar{R}_{i,N_T-1},2}(u)&\cdots&h_{\bar{R}_{i,N_T-1},N_T}(u)\\
h_{\bar{R}_{i,1},1}(u)&h_{\bar{R}_{i,1},2}(u)&\cdots&h_{\bar{R}_{i,1},N_T}(u)
\end{vmatrix}\notag\\
=\,&0.
\end{align}
%For notation simplicity, we define $\tilde{h}_{q,\mathcal{T}}^{\mathcal{R}'_i}$ as
%\begin{align}
%\tilde{h}_{q,\mathcal{T}}^{\mathcal{R}'_i}\triangleq\sum_{p=T_1}^{T_t}h_{q,p}c_{p}=
%\begin{vmatrix}
%h_{R'_{i,1},1}&h_{R'_{i,1},2}&\cdots&h_{R'_{i,1},N_T}\\
%h_{R'_{i,2},1}&h_{R'_{i,2},2}&\cdots&h_{R'_{i,2},N_T}\\
%\cdots&\cdots&\cdots&\cdots\\
%h_{R'_{i,N_T-1},1}&h_{R'_{i,N_T-1},T_2}&\cdots&h_{R'_{i,N_T-1},N_T}\\
%h_{q,1}&h_{q,T_2}&\cdots&h_{q,T_t}
%\end{vmatrix}.\label{eqn proof4 equivalent channel}
%\end{align}
By the above construction of precoders, it can be seen that symbol $x_{{\mathcal{R}},{[N_T]}}^i$ unwanted by receiver $q\in\bar{\mathcal{R}}_i$ are all neutralized. Then, the received signal at receiver $q$ can be rewritten as
\begin{align}
y_q(u)=&\sum_{\mathcal{R}:\mathcal{R}\ni q,|\mathcal{R}|=r+1}\sum_{i=1}^\rho\left[\sum_{p\in [N_T]}h_{qp}(u)v_{{\mathcal{R}},{[N_T]},p}^i(u)\right]x_{{\mathcal{R}},{[N_T]}}^i\notag\\
&+\sum_{\mathcal{R},\bar{\mathcal{R}}_i:\mathcal{R}\cup\bar{\mathcal{R}}_i\not\ni q}\left[\sum_{p\in [N_T]}h_{qp}(u)v_{{\mathcal{R}},{[N_T]},p}^i(u)\right]x_{{\mathcal{R}},{[N_T]}}^i,\label{eqn proof4 receive signal after neutralization}
\end{align}
where the first term is the desired messages and the second term is the interference. To guarantee the decodability of receiver $q$, we need to assure the following $S\times S$ received signal matrix whose column vectors are
\begin{align}
&\left\{\left(\sum_{p\in[N_T]}h_{qp}(u)v_{{\mathcal{R}},{[N_T]},p}^i(u)\right)_{u=1}^S: \mathcal{R}\ni q,|\mathcal{R}|=r+1,i\in[\rho]\right\}\notag\\
&\cup\left\{\left(\sum_{p\in[N_T]}h_{qp}(u)v_{{\mathcal{R}},{[N_T]},p}^i(u)\right)_{u=1}^S:\mathcal{R}\cup\bar{\mathcal{R}}_i\not\ni q\right\}\notag
\end{align}
to be full-rank with probability 1. Since the construction method of $\{v_{{\mathcal{R}},{[N_T]},p}^i(u)\}$ and the formation of $\{\sum_{p\in[N_T]}h_{qp}(u)v_{{\mathcal{R}},{[N_T]},p}^i(u)\}$ are the same at each time slot $u$, based on \cite[Lemma 3]{CoMP}, we only need to prove the linear independence of these polynomial functions, which is given by
\begin{align}
&\left\{\alpha_{{\mathcal{R}},{[N_T]}}^{\bar{\mathcal{R}}_i}\sum_{p=1}^{N_T}h_{qp}c_p\right\}_{\mathcal{R},i:\mathcal{R}\ni q,|\mathcal{R}|=r+1,i\in[\rho]}\cup\left\{\alpha_{{\mathcal{R}},{[N_T]}}^{\bar{\mathcal{R}}_i}\sum_{p=1}^{N_T}h_{qp}c_p\right\}_{\mathcal{R},\bar{\mathcal{R}}_i:\mathcal{R}\cup\bar{\mathcal{R}}_i\not\ni q}.\notag
\end{align}
Since each polynomial has a unique factor $\alpha_{{\mathcal{R}},{[N_T]}}^{\bar{\mathcal{R}}_i}$, it is obvious that these polynomials are linearly independent. Similarly, the polynomials of received symbols at other receivers are also linearly independent. Therefore, the received signal matrix at each receiver is full rank with probability 1 using \cite[Lemma 3]{CoMP}, and each receiver can decode its desired signals successfully. Since each receiver can decode $\binom{N_R-1}{r}\binom{N_R-r-1}{N_T-1}$ symbols in $S=\binom{N_R-1}{r}\binom{N_R-r-1}{N_T-1}+\binom{N_R-1}{r+1}\binom{N_R-r-2}{N_T-1}$-symbol extension, the per-user DoF of
\begin{align}
d=\frac{\binom{N_R-1}{r}\binom{N_R-r-1}{N_T-1}}{\binom{N_R-1}{r}\binom{N_R-r-1}{N_T-1}+\binom{N_R-1}{r+1}\binom{N_R-r-2}{N_T-1}}\label{eqn dof 3}
\end{align}
is achieved.

Combining \eqref{eqn dof 2} and \eqref{eqn dof 3}, we obtain the per-user DoF of
\begin{align}
d=\frac{\binom{N_R-1}{r}\binom{N_T}{t}\binom{N_R-r-1}{t-1}t}
{\binom{N_R-1}{r}\binom{N_T}{t}\binom{N_R-r-1}{t-1}t+\binom{N_R-1}{r+1}\binom{N_R-r-2}{t-1}\binom{N_T}{t-1}}\label{eqn dof 4}
\end{align}
 for $t\le N_T$.

It can be seen that \eqref{eqn dof 4} is an increasing function of $t$ when $N_R-N_T-r-1<0$, and is a decreasing function of $t$ when $N_R-N_T-r-1>0$. Intuitively, the achievable per-user DoF should be a non-decreasing function of transmitter cooperation size $t$. Thus, to obtain a reasonable DoF, we introduce the following proposition.
\begin{proposition}\label{proposition 1}
Any achievable DoF of the $\binom{N_T}{t'}\times\binom{N_R}{r+1}$ cooperative X-multicast channel can be achieved in the $\binom{N_T}{t}\times\binom{N_R}{r+1}$ cooperative X-multicast channel, where $t'<t$.
\end{proposition}
\begin{proof}
Consider an arbitrary $\binom{N_T}{t}\times\binom{N_R}{r+1}$ cooperative X-multicast channel. Split each message $W_{{\mathcal{R}},{\mathcal{T}}}$ into $\binom{t}{t'}$ submessages, each associated with a unique transmitter set $\mathcal{T}'\subset \mathcal{T}$ with $|\mathcal{T}'|=t'$, and will be transmitted by transmitter set $\mathcal{T}'$ only. We denote submessages of $W_{{\mathcal{R}},{\mathcal{T}}}$ delivered by transmitter set  $\mathcal{T}'$ as $W_{{\mathcal{R}},{\mathcal{T}}}^{\mathcal{T}'}$. Then, in the delivery phase for an arbitrary transmitter set $\mathcal{T}'$, each transmitter in $\mathcal{T}'$ will generate and transmit a super-message desired by receiver set $\mathcal{R}$:
\begin{align}
\hat{W}_{{\mathcal{R}},{\mathcal{T}'}}=\left\{W_{{\mathcal{R}},{\mathcal{T}}}^{\mathcal{T}'}: \mathcal{T}\supset\mathcal{T}'\right\}.\notag
\end{align}
Through this approach, each super-message $\hat{W}_{{\mathcal{R}},{\mathcal{T}'}}$ is available at transmitter set $\mathcal{T}'$ ($|\mathcal{T}'|=t'$) and desired by receiver set $\mathcal{R}$. The network topology has changed to the $\binom{N_T}{t'}\times\binom{N_R}{r+1}$ cooperative X-multicast channel. Therefore, any achievable DoF in the $\binom{N_T}{t'}\times\binom{N_R}{r+1}$ cooperative X-multicast channel can be achieved in the original $\binom{N_T}{t}\times\binom{N_R}{r+1}$ cooperative X-multicast channel.
\end{proof}

Based on Proposition \ref{proposition 1}, combining the achievable per-user DoF of $\frac{r+t}{N_R}$ in the first method and \eqref{eqn dof 4} in the second method, we obtain the achievable per-user DoF when $r+t\le N_R-2$ as
\begin{align}
d=\max\left\{d'_{r,t},\frac{r+t}{N_R}\right\},\notag
\end{align}
where $d'_{r,t}$ is given in Lemma \ref{lemma dof}.

By combining the results in all three parts, Lemma \ref{lemma dof} is proved.

\section*{Appendix B: Optimality (Proof of Corollary \ref{coro optimality})}
We prove the optimality of the proposed caching and delivery scheme presented in Section \ref{section cache placement} and \ref{section delivery}. We consider the following four cases.

\subsubsection{$N_R\mu_R+N_T\mu_T\ge N_R$}
Letting $l=s_1=1,s_2=0$ in \eqref{eqn converse} in Theorem \ref{thm 2}, we have:
\begin{align}
\tau^*(\mu_R,\mu_T)\ge 1-\mu_R.\label{eqn 1-mur}
\end{align}
Now consider the achievable upper bound of NDT. Subtracting \eqref{eqn tmin2} from \eqref{eqn tmin1}, we obtain
\begin{align}
1-\mu_R\le&\sum_{t=1}^{N_T}\binom{N_T}{t}a_{0,t}+\sum_{r=1}^{N_R-1}\sum_{t=1}^{N_T}\left[\binom{N_R}{r}-\binom{N_R-1}{r-1}\right]\binom{N_T}{t}a_{r,t}\notag\\
=&\sum_{r=0}^{N_R-1}\sum_{t=1}^{N_T}\binom{N_R-1}{r}\binom{N_T}{t}a_{r,t}.\label{eqn opt 11}
\end{align}
Substituting \eqref{eqn opt 11} into \eqref{eqn tau v}, we obtain that the achievable NDT must satisfy
\begin{align}
\tau&\ge 1-\mu_R+\sum_{\{(r,t):r+t<N_R\}}\binom{N_R-1}{r}\binom{N_T}{t}\left(\frac{1}{d_{r,t}}-1\right)a_{r,t}\notag\\
&\ge1-\mu_R,\label{eqn opt 12}
\end{align}
and hence $\tau_U\ge1-\mu_R$. Note that the two equalities in \eqref{eqn opt 12} can be achieved at the same time when $N_R\mu_R+N_T\mu_T\ge N_R$.

In specific, when $N_T\ge N_R$, consider the following file splitting ratios which satisfy constraints \eqref{eqn tmin1}\eqref{eqn tmin2}\eqref{eqn tmin3}:
\begin{align}
a^*_{N_R,0}=\mu_R,a^*_{0,N_R}=\frac{1-\mu_R}{\binom{N_T}{N_R}}, \textrm{ and others being 0}.\notag
\end{align}
Substituting it into \eqref{eqn tmin}, we have:
\begin{align}
\tau_U=\binom{N_T}{N_R}a^*_{0,N_R}=1-\mu_R,\notag
\end{align}
which coincides with lower bound \eqref{eqn 1-mur}, and thus is optimal.

When $N_T< N_R$, consider the following file splitting ratios:
\begin{align}
&a^*_{N_R,0}=1-(1-\mu_R)\frac{N_R}{N_T},a^*_{N_R-N_T,N_T}=\frac{1-\mu_R}{\binom{N_R-1}{N_R-N_T}},\textrm{ and others being 0}.\notag
\end{align}
Similar arguments when $N_T\ge N_R$ can be applied here again, and are omitted. Thus, the optimality when $N_T\mu_T+N_R\mu_R\ge N_R$ is proved.

\subsubsection{$(\mu_R,\mu_T)=(0,1)$}
Substituting $l=s_1=\min\{N_T,N_R\},s_2=N_R-\min\{N_T,N_R\}$ into \eqref{eqn converse}, we have:
\begin{align}
\tau^*\ge \frac{N_R}{\min\{N_T,N_R\}}.\label{eqn converse broadcast}
\end{align}
Now consider the achievable upper bound of NDT. Since there is no cache storage at receivers, which implies that $a_{r,t}=0$ for $r>0$, the achievable NDT in \eqref{eqn tmin} reduces to
\begin{align}
\tau_U=\min\sum_{t=1}^{N_T}\frac{\binom{N_T}{t}}{d_{0,t}}a_{0,t},\label{eqn opt 21}
\end{align}
and constraint \eqref{eqn tmin1} reduces to
\begin{align}
\sum_{t=1}^{N_T}\binom{N_T}{t}a_{0,t}=1.\label{eqn opt 22}
\end{align}
It can be seen from \eqref{eqn tau ip} that $d_{0,t}\le \min\{1,N_T/N_R\}$ for $t\in [N_T]$, thus we have
\begin{align}
\sum_{t=1}^{N_T}\frac{\binom{N_T}{t}}{d_{0,t}}a_{0,t}\ge\frac{1}{\min\{1,\frac{N_T}{N_R}\}}\sum_{t=1}^{N_T}\binom{N_T}{t}a_{0,t}=\frac{N_R}{\min\{N_T,N_R\}}.\label{eqn opt 23}
\end{align}
Note that the equality in \eqref{eqn opt 23} can be achieved by letting file splitting ratios satisfy
\begin{align}
a^*_{0,N_T}=1,\textrm{ and others being 0}\notag
\end{align}
in \eqref{eqn tmin}. Thus, the optimality when $(\mu_R,\mu_T)=(0,1)$ is proved.

\subsubsection{$(\mu_R,\mu_T)=(0,1/N_T)$}
substituting $l=s_1=1,s_2=N_R-1$ into \eqref{eqn converse}, we have:
\begin{align}
\tau^*\ge \frac{N_T+N_R-1}{N_T}.\label{eqn converse xchannel}
\end{align}
Now consider the achievable upper bound of NDT. In this case, the only feasible file splitting ratios are given by
\begin{align}
a^*_{0,1}=1/N_T,\textrm{ and others being 0.}\notag
\end{align}
Substituting it into \eqref{eqn tmin}, the achievable NDT is given by
\begin{align}
\tau=\frac{N_T}{\frac{N_T}{N_T+N_R-1}}a^*_{0,1}=\frac{N_T+N_R-1}{N_T}\notag,
\end{align}
which coincides with the lower bound \eqref{eqn converse xchannel} and thus is optimal.

\subsubsection{$\mu_R+N_T\mu_T=1$ when intra-file coding is not allowed in the caching functions}
Substituting $l=s_1=1,s_2=N_R-1$ into \eqref{eqn taul2}, we have:
\begin{align}
\tau^*\ge (N_T+N_R-1)\mu_T=\frac{N_T+N_R-1}{N_T}(1-\mu_R)\label{eqn converse nointralfile}.
\end{align}
Now consider the achievable upper bound of NDT. In this case, the only feasible file splitting ratios are given by
\begin{align}
a^*_{0,1}=\frac{1-\mu_R}{N_T},a^*_{N_R,0}=\mu_R,\textrm{ and others being 0.}\notag
\end{align}
Substituting it into \eqref{eqn tmin}, the achievable NDT is given by
\begin{align}
\tau=\frac{N_T}{\frac{N_T}{N_T+N_R-1}}a^*_{0,1}=\frac{N_T+N_R-1}{N_T}(1-\mu_R)\notag,
\end{align}
which coincides with the lower bound \eqref{eqn converse nointralfile} and thus is optimal.

Summarizing all the four cases above, we finished the proof of Corollary \ref{coro optimality}.

\section*{Appendix C: Maximum Multiplicative Gap (Proof of Corollary \ref{coro gap})}
Given the fact that NDT is optimal when $N_T\mu_T+N_R\mu_R\ge N_R$, we only need to prove the multiplicative gap when $N_T\mu_T+N_R\mu_R<N_R$. Denote $g$ as the multiplicative gap. We consider three cases to prove Corollary \ref{coro gap}: (1) $N_T\ge N_R$; (2) $N_T< N_R$ and $\mu_T\ge 1/N_T$; (3) $N_T< N_R$ and $\mu_T< 1/N_T$.

\setcounter{subsection}{0}
\subsection{$N_T\ge N_R$}

Using \eqref{eqn 1-mur}, we have:
\begin{align}
g\le\frac{1}{1-\mu_R}\cdot\min_{\{a_{r,t}\}}\sum_{r=0}^{N_R-1}\sum_{t=1}^{N_T}\frac{\binom{N_R-1}{r}\binom{N_T}{t}}{d_{r,t}}a_{r,t},\label{eqn gap}
\end{align}
Consider the following file splitting ratios in \eqref{eqn gap}:
\begin{align}
&a_{N_R,0}=\mu_R,a_{0,1}=\frac{N_R(1-\mu_R)-N_T\mu_T}{N_T(N_R-1)},a_{0,N_R}=\frac{N_T\mu_T+\mu_R-1}{\binom{N_T}{N_R}(N_R-1)}, \textrm{ and others being 0,}\notag
\end{align}
then we have the following upper bound:
\begin{align}
g&\le\frac{1}{1-\mu_R}\cdot\frac{N_T}{d_{0,1}}\frac{N_R(1-\mu_R)-N_T\mu_T}{N_T(N_R-1)}+\frac{1}{1-\mu_R}\cdot\frac{\binom{N_T}{N_R}}{d_{0,N_R}}\frac{N_T\mu_T+\mu_R-1}{\binom{N_T}{N_R}(N_R-1)}\notag\\
&=\frac{1}{1-\mu_R}\cdot\frac{N_T+N_R-1}{N_R-1}\left(\frac{N_R}{N_T}(1-\mu_R)-\mu_T\right)+\frac{1}{1-\mu_R}\cdot\frac{N_T\mu_T+\mu_R-1}{N_R-1}\notag\\
&=1+\frac{N_R}{N_T}-\frac{\mu_T}{1-\mu_R}\notag\\
&\le2.\label{eqn gap2}
\end{align}
Therefore, the multiplicative gap is within 2 when $N_T\ge N_R$.

\subsection{$N_T<N_R$ and $\mu_T\ge \frac{1}{N_T}$}

we consider six cases to discuss the multiplicative gap $g$: (1) $N_R\le 1.8N_T$, (2) $N_R>1.8N_T,\mu_R\le\frac{1}{2N_R-N_T}$, (3) $N_R>1.8N_T,N_T=2,\frac{1}{2N_R-2}<\mu_R<\frac{1}{4}$, (4) $N_R>1.8N_T,N_T=2,\mu_R\ge\frac{1}{4}$, (5) $N_R>1.8N_T,N_T\ge3,\frac{1}{2N_R-N_T}<\mu_R<\frac{N_T-\sqrt{2N_T^2-2N_T}}{2N_T-N_T^2}$, (6) $N_R>1.8N_T,N_T\ge3,\mu_R\ge\frac{N_T-\sqrt{2N_T^2-2N_T}}{2N_T-N_T^2}$.

\subsubsection{$N_R\le 1.8N_T$}

Letting file splitting ratios $a_{N_R,0}=\mu_R,a_{0,1}=\frac{1-\mu_R}{N_T}$ and others being 0 in \eqref{eqn tmin}, we have $\tau^*\le\frac{N_T+N_R-1}{N_T}(1-\mu_R)$. Comparing the lower bound \eqref{eqn 1-mur}, we have
\begin{align}
g\le\frac{N_T+N_R-1}{N_T}<2.8\label{eqn gap2.8}
\end{align}

\subsubsection{$N_R>1.8N_T,\mu_R\le\frac{1}{2N_R-N_T}$}

Letting $l=s_1=N_T,s_2=N_R-N_T$ in \eqref{eqn converse}, we have $\tau^*\ge\frac{1}{N_T}\left\{N_R-\left((N_R-N_T)N_R+\frac{N_T^2+N_T}{2}\right)\mu_R\right\}$. Using the same upper bound of $\tau^*$ as in case (1), i.e. $\tau^*\le\frac{N_T+N_R-1}{N_T}(1-\mu_R)$, we have
\begin{align}
g&\le\frac{(N_T+N_R-1)(1-\mu_R)}{N_R-\left((N_R-N_T)N_R+\frac{N_T^2+N_T}{2}\right)\mu_R}\notag\\
&\le\frac{(N_T+N_R-1)(2N_R-N_T)}{N_R^2-N_T^2/2-N_T/2}\notag\\
&\le\frac{N_TN_R+2N_R^2-2N_R+N_T}{N_R^2-N_T^2/2-N_T/2}\notag\\
&\le\frac{2N_R^2+N_R^2/1.8-2N_R+N_R/1.8}{N_R^2-N_R^2/(2\times 1.8^2)-N_R/(2\times 1.8)}\notag\\
&=\frac{(4+2/1.8)N_R-(4-2/1.8)}{(2-1/1.8^2)N_R-1/1.8}\notag\\
&=\frac{4+2/1.8}{2-1/1.8^2}+\frac{\frac{4+2/1.8}{2\times 1.8-1/1.8}-(4-2/1.8)}{(2-1/1.8^2)N_R-1/1.8}\notag\\
&<\frac{4+2/1.8}{2-1/1.8^2}<3.1\label{eqn gap3.1}
\end{align}

\subsubsection{$N_R>1.8N_T,N_T=2,\frac{1}{2N_R-2}<\mu_R<\frac{1}{4}$}

Denote $\hat{\mu_R}=\lfloor \mu_RN_R\rfloor /N_R$. Since the achievable upper bound $\tau_U$ is a decreasing function of $\mu_R$ and $\mu_T$, we have:
\begin{align}
\tau^*(\mu_R,\mu_T)\le \tau_U(\mu_R,\mu_T)\le \tau_U(\mu_R,1/N_T)\le \tau_U(\hat{\mu_R},1/N_T).\label{eqn gap 3 middle}
\end{align}
Here $\tau_U(\hat{\mu_R},1/N_T)$ is upper bounded by
\begin{align}
\tau_U(\hat{\mu_R},1/N_T)&\le\frac{N_T-1+\frac{N_R}{N_R\hat{\mu_R}+1}}{N_T}(1-\hat{\mu_R)}\notag\\
&\le\frac{N_T+\frac{N_R}{N_R\hat{\mu_R}+1}}{N_T}\notag\\
&\le\frac{N_T+\frac{N_R}{N_R\mu_R}}{N_T}=1+\frac{1}{N_T\mu_R},\label{eqn gap 3 upper}
\end{align}
by letting file splitting ratio $a_{N_R\hat{\mu_R},1}=\frac{1}{N_T\binom{N_R}{N_R\hat{\mu_R}}}$, and others being 0 in \eqref{eqn tmin}. Letting $l=s_1=N_T,s_2=\lfloor\frac{1}{2\mu_R}-1\rfloor$ in \eqref{eqn converse}, we have 
\begin{align}
\tau^*&\ge\frac{1}{N_T}\left\{\left(N_T+\lfloor\frac{1}{2\mu_R}-1\rfloor\right)-\left(\lfloor\frac{1}{2\mu_R}-1\rfloor^2+N_T^2/2+N_T/2+N_T\lfloor\frac{1}{2\mu_R}-1\rfloor\right)\mu_R\right\}\notag\\
&\ge\frac{1}{N_T}\left\{\left(N_T+\frac{1}{2\mu_R}-1-1\right)-\left((\frac{1}{2\mu_R}-1)^2+N_T^2/2+N_T/2+N_T(\frac{1}{2\mu_R}-1)\right)\mu_R\right\}\notag\\
&=\frac{1}{N_T}\left\{\left(N_T+\frac{1}{2\mu_R}-2\right)-\left(\frac{1}{4\mu_R^2}+N_T^2/2-N_T/2+1+(N_T/2-1)/\mu_R\right)\mu_R\right\}\notag\\
&=\frac{1}{N_T}\left\{N_T/2-1+\frac{1}{4\mu_R}-(N_T^2/2-N_T/2+1)\mu_R\right\}=\frac{1}{N_T}\left\{\frac{1}{4\mu_R}-2\mu_R\right\}\label{eqn gap 3 lower}
\end{align}
Comparing \eqref{eqn gap 3 upper} and \eqref{eqn gap 3 lower}, we have
\begin{align}
g\le\frac{N_T+\frac{1}{\mu_R}}{\frac{1}{4\mu_R}-2\mu_R}=\frac{2\mu_R+1}{1/4-2\mu_R^2}< \frac{2\times 1/4+1}{1/4-2\times (1/4)^2}=12.\label{eqn gap12}
\end{align}

\subsubsection{$N_R>1.8N_T,N_T=2,\mu_R\ge\frac{1}{4}$}

By the convexity of the achievable upper bound $\tau_U$, we have
\begin{align}
&\tau^*(\mu_R,\mu_T)\notag\\
\le &\tau_U(\mu_R,\mu_T)\notag\\
\le &\tau_U(\mu_R,1/N_T)\notag\\
\le &\tau_U(1,1/N_T)+\frac{\tau_U(1/4,1/N_T)-\tau_U(1,1/N_T)}{3/4}(1-\mu_R)\notag\\
=&\frac{\tau_U(1/4,1/N_T)}{3/4}(1-\mu_R).\label{eqn gap 4 upper1}
\end{align}
Denoting $\hat{\mu_R}=\lfloor \frac{1}{4}N_R\rfloor /N_R$ and using \eqref{eqn gap 3 middle}\eqref{eqn gap 3 upper}, we have
\begin{align}
\tau_U(1/4,1/N_T)&\le \tau_U(\hat{\mu_R},1/N_T)\notag\\
&\le\tau_U(\hat{\mu_R},1/N_T)\notag\\
&\le 1+\frac{1}{N_T/4}.\label{eqn gap 4 upper2}
\end{align}
Using \eqref{eqn 1-mur}\eqref{eqn gap 4 upper1}\eqref{eqn gap 4 upper2}, we have
\begin{align}
g\le\frac{4}{3}\left(1+\frac{1}{N_T/4}\right)=4.\label{eqn gap4}
\end{align}

\subsubsection{$N_R>1.8N_T,N_T\ge3,\frac{1}{2N_R-N_T}<\mu_R<\frac{N_T-\sqrt{2N_T^2-2N_T}}{2N_T-N_T^2}$}

Denote $\hat{\mu_R}=\lfloor \mu_RN_R\rfloor /N_R$ and $\mu_R^0=\frac{N_T-\sqrt{2N_T^2-2N_T}}{2N_T-N_T^2}$. Similar to \eqref{eqn gap 3 middle}\eqref{eqn gap 3 upper}, we have
\begin{align}
\tau^*(\mu_R,\mu_T)\le \tau_U(\hat{\mu_R},1/N_T)\le 1+\frac{1}{N_T\mu_R}.\label{eqn gap 5 upper}
\end{align}
Letting $l=s_1=N_T,s_2=\lfloor\frac{1}{2\mu_R}-N_T/2\rfloor$ in \eqref{eqn converse}, we have 
\begin{align}
\tau^*&\ge\frac{1}{N_T}\left\{\left(N_T+\lfloor\frac{1}{2\mu_R}-N_T/2\rfloor\right)\right.\notag\\
&\quad\qquad\left.-\left(\lfloor\frac{1}{2\mu_R}-N_T/2\rfloor^2+N_T^2/2+N_T/2+N_T\lfloor\frac{1}{2\mu_R}-N_T/2\rfloor\right)\mu_R\right\}\notag\\
&\ge\frac{1}{N_T}\left\{\left(N_T+\frac{1}{2\mu_R}-N_T/2-1\right)\right.\notag\\
&\quad\qquad\left.-\left((\frac{1}{2\mu_R}-N_T/2)^2+N_T^2/2+N_T/2+N_T(\frac{1}{2\mu_R}-N_T/2)\right)\mu_R\right\}\notag\\
&=\frac{1}{N_T}\left\{\left(N_T/2+\frac{1}{2\mu_R}-1\right)-\left(\frac{1}{4\mu_R^2}+N_T/2+\frac{N_T^2}{4}\right)\mu_R\right\}\notag\\
&=\frac{1}{N_T}\left\{N_T/2-1+\frac{1}{4\mu_R}-(N_T/2+\frac{N_T^2}{4})\mu_R\right\}\notag\\
&\ge\frac{1}{N_T}\left\{N_T/2-1+\frac{1}{4\mu_R}-\frac{N_T^2+2N_T}{4}\mu_R^0\right\}\notag\\
&=\frac{1}{N_T}\left\{N_T/2-1+\frac{1}{4\mu_R}-\frac{N_T+2}{4}\frac{N_T-\sqrt{2N_T^2-2N_T}}{2-N_T}\right\}\notag\\
&\ge\frac{1}{N_T}\left\{N_T/2-1+\frac{1}{4\mu_R}-N_T/4\right\}\notag\\
&=\frac{1}{N_T}\left\{N_T/4-1+\frac{1}{4\mu_R}\right\}\label{eqn gap 5 lower}
\end{align}
Combining \eqref{eqn gap 5 upper} and \eqref{eqn gap 5 lower}, we have
\begin{align}
g&\le\frac{1+\frac{1}{N_T\mu_R}}{1/4-1/N_T+\frac{1}{4N_T\mu_R}}\notag\\
&=4+\frac{4/N_T}{1/4-1/N_T+\frac{1}{4N_T\mu_R}}\notag\\
&=4+\frac{4}{N_T/4-1+\frac{1}{4\mu_R}}\notag\\
&<4+\frac{4}{N_T/4-1+\frac{1}{4\mu_R^0}}<7\label{eqn gap7}
\end{align}

\subsubsection{$N_R>1.8N_T,N_T\ge3,\mu_R\ge\frac{N_T-\sqrt{2N_T^2-2N_T}}{2N_T-N_T^2}$}

By the convexity of the achievable upper bound $\tau_U$, we have
\begin{align}
&\tau^*(\mu_R,\mu_T)\notag\\
\le &\tau_U(\mu_R,\mu_T)\notag\\
\le &\tau_U(\mu_R,1/N_T)\notag\\
\le &\tau_U(1,1/N_T)+\frac{\tau_U(\mu_R^0,1/N_T)-\tau_U(1,1/N_T)}{1-\mu_R^0}(1-\mu_R)\notag\\
=&\frac{\tau_U(\mu_R^0,1/N_T)}{1-\mu_R^0}(1-\mu_R).\label{eqn gap 6 upper1}
\end{align}
Denote $\hat{\mu_R}^0=\lfloor \mu_R^0 N_R\rfloor /N_R$. Using \eqref{eqn gap 3 middle}\eqref{eqn gap 3 upper}, we have
\begin{align}
\tau_U(\mu_R^0,1/N_T)\le \tau_U(\hat{\mu_R}^0,1/N_T)\le 1+\frac{1}{N_T\mu_R^0}\label{eqn gap 6 upper2}
\end{align}
Combining \eqref{eqn 1-mur}, \eqref{eqn gap 6 upper1} and \eqref{eqn gap 6 upper2}, we have
\begin{align}
g\le \frac{1+\frac{1}{N_T\mu_R^0}}{1-\mu_R^0}<3.8\label{eqn gap3.8}
\end{align}

Thus, by combining the above six cases, the multiplicative gap is within 12 when $N_T<N_R$ and $\mu_T\ge\frac{1}{N_T}$.

\subsection{$N_T<N_R$ and $\mu_T<\frac{1}{N_T}$}

Consider the following file splitting ratios in \eqref{eqn tmin}:
\begin{align}
a_{0,1}=\frac{1-\mu_R}{N_T},a_{N_R,0}=\mu_R,\textrm{ and others being 0.}\notag
\end{align}
Then, we have the following achievable NDT:
\begin{align}
\tau= \frac{N_T}{\frac{N_T}{N_T+N_R-1}}\frac{1-\mu_R}{N_T}=\frac{N_T+N_R-1}{N_T}(1-\mu_R).\notag
\end{align}
Comparing to \eqref{eqn 1-mur}, we have
\begin{align}
g\le\frac{N_T+N_R-1}{N_T}.\notag
\end{align}

Summarizing all the analysis above, Corollary \ref{coro gap} is proved.

\section*{Appendix D: Optimization of File Splitting Ratios in the $2\times 2$ Network (Proof of Corollary \ref{coro 2x2})}
The LP problem in the $2\times 2$ network is expressed as
\begin{align}
\min\, &\tau_2=3a_{0,1}+a_{0,2}+2a_{1,1}+a_{1,2}\label{eqn optimization nt2}\\
\textrm{s.t. }&2a_{0,1}+a_{0,2}+4a_{1,1}+2a_{1,2}+a_{2,0}+2a_{2,1}+a_{2,2}=1,\label{eqn:total cachent2}\\
&2a_{1,1}+a_{1,2}+a_{2,0}+2a_{2,1}+a_{2,2}\le\mu_R, \label{eqn:receiver cachent2}\\
&a_{0,1}+a_{0,2}+2a_{1,1}+2a_{1,2}+a_{2,1}+a_{2,2}\le\mu_T.  \label{eqn:transmitter cachent2}
\end{align}
Subtracting \eqref{eqn:receiver cachent2} from \eqref{eqn:total cachent2}, we have
\begin{align}
1-\mu_R\le2a_{0,1}+a_{0,2}+2a_{1,1}+a_{1,2}.\label{eqn 1-murnt2}
\end{align}
Substituting \eqref{eqn 1-murnt2} into the objective function in \eqref{eqn optimization nt2}, we get
\begin{subequations}\label{eqn tau 1-mur nt2}
\begin{align}
\tau_2&\ge1-\mu_R+a_{0,1}\label{eqn tau 1-mur nt21}\\
&\ge1-\mu_R.\label{eqn tau 1-mur nt22}
\end{align}
\end{subequations}
Now we discuss the solution in regions $\mathcal{R}_{22}^1$ and $\mathcal{R}_{22}^2$ individually.

\textit{Region $\mathcal{R}_{22}^1$}:\quad In this region, the equality in \eqref{eqn tau 1-mur nt2} holds if $a_{0,1}=0$. This can be satisfied when the file splitting ratios $\{a_{r,t}\}$ take the following values:
\begin{align}
a^*_{2,0}=\mu_R,a^*_{0,2}=1-\mu_R\ \textrm{and other ratios are 0}.\notag
\end{align}

\textit{Region $\mathcal{R}_{22}^2$}:\quad In this region, subtracting \eqref{eqn 1-murnt2} from \eqref{eqn:transmitter cachent2}, we can get
\begin{align}
&\quad\mu_T-(1-\mu_R)\notag\\
&\ge-a_{0,1}+a_{1,2}+a_{2,1}+a_{2,2},\notag\\
&\ge -a_{0,1}\notag
\end{align}
or equivalently
\begin{align}
a_{0,1}\ge1-\mu_R-\mu_T.\label{a01nt2}
\end{align}
Substituting \eqref{a01nt2} into \eqref{eqn tau 1-mur nt21}, we can get
\begin{align}
\tau_2\ge1-\mu_R+1-\mu_R-\mu_T=2(1-\mu_R)-\mu_T.
\end{align}
In this region, the minimum NDT $\tau_2=2(1-\mu_R)-\mu_T$ can be achieved. The solution for splitting ratios is not unique but must satisfy $a^*_{0,1}=1-\mu_R-\mu_T$ and $a^*_{2,1}=a^*_{2,2}=a^*_{1,2}=0$. Here we choose one feasible solution to be
\begin{align}
&a^*_{0,1}=1-\mu_R-\mu_T,a^*_{2,0}=\mu_R,a^*_{0,2}=2\mu_T-(1-\mu_R)\textrm{ and other ratios are 0}.\notag
\end{align}

\section*{Appendix E: Optimal Solution of File Splitting Ratios in the $3\times 3$ Network}
The optimal solution of file splitting ratios in the $3\times3$ network is given below, where all the regions are defined in Corollary \ref{coro 3x3}.

\textit{Region $\mathcal{R}^1_{33}$}: The optimal splitting ratios are not unique but must satisfy
\begin{align}
&a^*_{1,1}=a^*_{0,1}=a^*_{0,2}=0,\label{eqn 3311}\\
&a^*_{3,0}+3a^*_{3,1}+3a^*_{3,2}+a^*_{3,3}+6a^*_{2,1}+6a^*_{2,2}+2a^*_{2,3}+3a^*_{1,2}+a^*_{1,3}=\mu_R.\label{eqn 3312}
\end{align}
One feasible solution is
\begin{align}
a^*_{3,0}=\mu_R,a^*_{0,3}=1-\mu_R,
\end{align}
and other ratios are 0.

\textit{Region $\mathcal{R}^2_{33}$}: The optimal splitting ratios are not unique but must satisfy
\begin{align}
&a^*_{0,1}=a^*_{0,2}=a^*_{3,1}=a^*_{3,2}=a^*_{3,3}=a^*_{2,2}=a^*_{2,3}=a^*_{1,3}=0,\notag\\
&a^*_{1,1}=\frac{1}{3}-\frac{\mu_R}{3}-\frac{\mu_T}{3},\notag\\
&a^*_{3,0}+6a^*_{2,1}+3a^*_{1,2}=2\mu_R+\mu_T-1,\notag\\
&3a^*_{2,1}+6a^*_{1,2}+a^*_{0,3}=\mu_R+2\mu_T-1.\notag
\end{align}
One feasible solution is
\begin{align}
&a^*_{1,1}=\frac{1}{3}-\frac{\mu_R}{3}-\frac{\mu_T}{3},a^*_{3,0}=2\mu_R+\mu_T-1,a^*_{0,3}=\mu_R+2\mu_T-1,
\end{align}
and other ratios are 0.

\textit{Region $\mathcal{R}^3_{33}$}: The optimal splitting ratios are unique and given by
\begin{align}
a^*_{1,1}=\frac{\mu_R}{3},a^*_{0,2}=1-2\mu_R-\mu_T,a^*_{0,3}=3\mu_R+3\mu_T-2,
\end{align}
and other ratios being 0.

\textit{Region $\mathcal{R}^4_{33}$}: The optimal splitting ratios are unique and given by
\begin{align}
a^*_{1,1}=\frac{\mu_R}{3},a^*_{0,1}=\frac{2}{3}-\mu_R-\mu_T,a^*_{0,2}=\mu_T-\frac{1}{3},
\end{align}
and other ratios being 0.

\textit{Region $\mathcal{R}^5_{33}$}: The optimal splitting ratios are unique and given by
\begin{align}
a^*_{1,1}=\frac{\mu_R}{3}+\mu_T-\frac{1}{3},a^*_{0,1}=1-\mu_R-2\mu_T,a^*_{3,0}=1-3\mu_T,
\end{align}
and other ratios being 0.

\bibliographystyle{IEEEtran}
\bibliography{IEEEabrv,journal}

\end{document}